\documentclass{biometrika-full}

\usepackage{amsmath}

\usepackage{newtxtext}
\usepackage[subscriptcorrection]{newtxmath}

\graphicspath{{./art/}}

\usepackage{mystyle-bka}
\usepackage{titlesec}
\titlespacing{\paragraph}{%
  0pt}{
  0.2\baselineskip}{
  1em}

\titlespacing*{\section}
  {0pt}{\baselineskip}{\baselineskip}
\titlespacing*{\subsection}
  {0pt}{\baselineskip}{\baselineskip}
\begin{document}

\jname{}
\jyear{}
\jvol{}
\jnum{}
\cyear{}
\accessdate{}

\received{}
\revised{}

\markboth{Du et~al.}{Post-integrated inference}

\title{Assumption-Lean Post-Integrated Inference\\ with Surrogate Control Outcomes}

\author{Jin-Hong Du}
\affil{Department of Statistics and Actuarial Science \& Institute of Data Science,\\ The University of Hong Kong, Hong Kong SAR, China}

\author{\and Kathryn Roeder \and Larry Wasserman}

\affil{Department of Statistics and Data Science, Carnegie Mellon University,\\ Pittsburgh, PA 15213, USA
}

\maketitle

\begin{abstract}
    Data integration methods aim to extract low-dimensional embeddings from high-dimensional outcomes to remove unwanted variations, such as batch effects and unmeasured covariates, across heterogeneous datasets. However, multiple hypothesis testing after integration can be biased due to data-dependent processes.
    We introduce a robust post-integrated inference method that accounts for latent heterogeneity by utilizing control outcomes. 
    Leveraging causal interpretations, we derive nonparametric identifiability of the direct effects using negative control outcomes.
    By utilizing surrogate control outcomes as an extension of negative control outcomes, we develop semiparametric inference on projected direct effect estimands, accounting for hidden mediators, confounders, and moderators.
    These estimands remain statistically meaningful under model misspecifications and with error-prone embeddings.
    We provide bias quantifications and finite-sample linear expansions with uniform concentration bounds. 
    The proposed doubly robust estimators are consistent and efficient under minimal assumptions and potential misspecification, facilitating data-adaptive estimation with machine learning algorithms. 
    Our proposal is evaluated using random forests through simulations and analysis of single-cell CRISPR perturbed datasets, which may contain potential unmeasured confounders.
\end{abstract}



\section{Introduction}    
In the big data era, integrating information from multiple heterogeneous sources has become increasingly crucial for achieving larger sample sizes and more diverse study populations. 
    The applications of data integration are in a variety of fields, including but not limited to causal inference on heterogeneous populations \citep{shi2023data}, survey sampling \citep{yang2020doubly}, health policy \citep{paddock2024statistical}, retrospective psychometrics \citep{howe2023retrospective}, and multi-omics biological science \citep{du2022robust}.
    Data integration methods have been proposed to mitigate the unwanted effects of heterogeneous datasets and unmeasured covariates, recovering the common variation across datasets.
    However, a critical and often overlooked question is whether reliable statistical inference can be made from integrated data.
    Directly performing statistical inference on integrated outcomes and covariates of interest fails to account for the complex correlation structures introduced by the data integration process, often leading to improper analyses that incorrectly assume the corrected data points are independent \citep{li2023overcoming}.

While data integration is widely utilized in various fields, our paper focuses on a challenging scenario involving high-dimensional outcomes. 
Particularly in the context of genomics, experimental constraints often necessitate the collection of data in multiple batches \citep{luo2018batch,luecken2022benchmarking}.
Batch correction and data integration methods are commonly used in genomics to recover the \emph{low-dimensional embeddings} or manifolds of each observation from the \emph{high-dimensional outcomes}.
The naive approach uses a batch indicator as a covariate in a regression model for inference, which may not be sufficient for adjusting for batch effects and unmeasured covariates \citep{li2023overcoming}.
Instead, two-step methods are commonly employed in practice as a separate data preprocessing step to produce integrated data, which can then be utilized for downstream inference.
For instance, design-based methods, such as Combat \citep{johnson2007adjusting} and BUS \citep{luo2018batch}, incorporate the batch or unknown subtype indicator into hierarchical Bayesian models, providing location and scale corrections.
Additionally, design-free methods, including RUV \citep{gagnon2012using} and SVA \citep{leek2012sva}, directly estimate the latent confounding factors, allowing users to utilize the estimated latent variables as additional covariates for downstream inference.
These methods apply to samples that share the same underlying biological variability, which is our focus in this paper; see \Cref{fig:overview} for an illustration.

\begin{figure}[!t]
    \centering
    \includegraphics[width=0.7\linewidth]{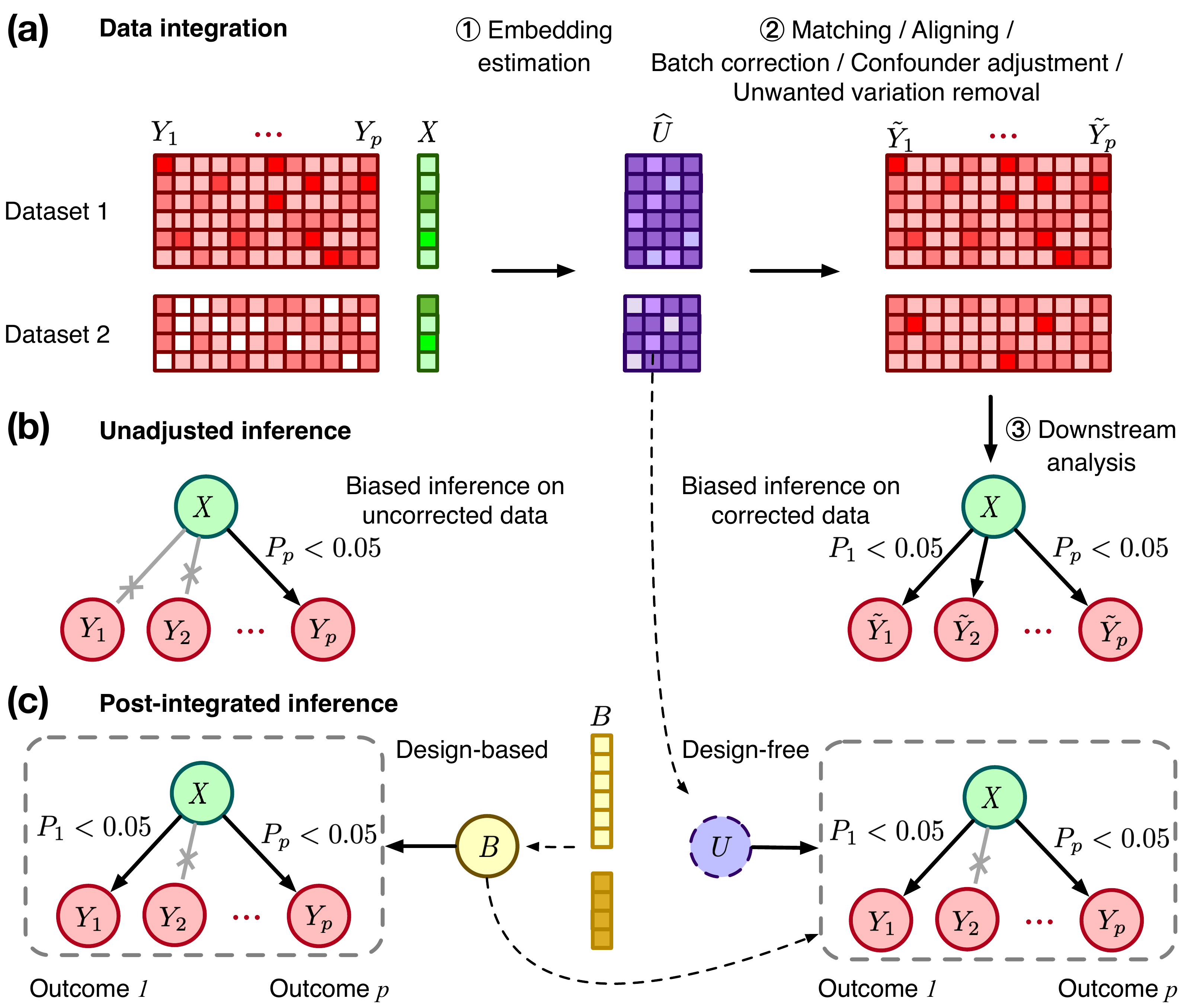}
    \caption{
    Overview of the post-integrated inference problem.    
    \textbf{(a)} Data integration utilizes multiple outcomes $Y=(Y_1,\ldots,Y_p)^{\top}$ and covariate $X$ of interest to estimate the unaligned embeddings $\hat{U}$ (e.g., top principal components of $Y$), and provides integrated outcomes $\tilde{Y}=(\tilde{Y}_1,\ldots,\tilde{Y}_p)^{\top}$ for downstream analysis.
    \textbf{(b)} Inference on the direct associations between $Y_j$'s and $X$, and those between $\tilde{Y}_j$'s and $X$ may be biased because of batch effects and observational dependency induced by data integration processes, respectively.
    \textbf{(c)} Post-integrated inference includes two strategies: the design-based approach that includes a batch indicator through a statistical model and the design-free approach that first estimates the latent embedding $\hat{U}$ and then treats it as extra covariates as a proxy of $U$ for downstream inference (the batch indicator can also be used as an observed confounder), where the latter is our focus.}
    \label{fig:overview}
\end{figure}

Despite different procedures and output formats, nearly all batch correction methods utilize information from multiple outcomes to estimate and align the underlying embeddings of observations. This approach is closely related to unmeasured confounder adjustment, particularly when each observation is viewed as a single dataset. Over the past decades, researchers have explored various methods to address unmeasured confounders in statistical analysis. In the presence of multiple outcomes, deconfounding techniques primarily employ two strategies: incorporating known negative control outcomes or leveraging sparsity assumptions \citep{wang2017confounder,zhou2024promises}. Additionally, a line of research on proximal causal inference uses both negative control outcomes and/or exposures for deconfounding \citep{miao2018identifying}; see a review of related work in \Cref{app:related-work}. 
This paper focuses specifically on the surrogate control approach in the context of multiple outcomes, where the goal is to directly estimate and adjust for latent factors that may confound the relationships between treatment outcomes.

Mathematically, a high-dimensional outcome vector $Y \in \RR^p$ is often related to a covariate vector $X\in\RR^d$ and an unobserved low-dimensional latent vector $U \in \RR^r$.
Here, $X$ includes variables such as disease status or treatment, and $U$, frequently referred to as the embedding vector, captures both the batch effects and the unmeasured covariates.
Both of them serve as a compact representation of the outcome $Y$, with the dimensionality of the outcome space being significantly larger than that of the covariate and latent space, i.e., $p \gg d$ and $p \gg r$.
Differences in how data are collected across datasets can result in shifts or distortions in the distribution of the unobserved variable $U$, and can potentially affect the distribution of $X$ as well.
Our primary interest lies in the direct associations or causal relationships between the outcome $Y_j$ and the covariate $X$ for $j=1,\ldots,p$, after adjusting for the difference induced by unwanted variation $U$.
When $X$ and $U$ are independent, the problem would be trivial because the direct effects can be estimated by regressing $Y_j$'s on $X$.
However, when $X$ and $U$ are dependent, the direct regression approach targets the total effects and provides a biased estimate of the direct effects.
Hence, proper data integration methods need to estimate $U$ for outcome alignment from different sources and for multiple hypothesis testing.

A common strategy to perform inference on the direct effects of $X$ on $Y$ is to employ two-step procedures, which practitioners widely favor. 
These methods first estimate the latent embedding $\hat{U}$ from high-dimensional (negative control) outcomes and then treat it as an additional covariate for downstream inference.
However, this sequential approach propagates estimation error: the uncertainty from the first step is typically ignored in the second, potentially invalidating the final statistical conclusions.
Specifically, the estimation of latent embeddings $U$ and the subsequent statistical inference are both contingent on the assumptions made by their respective models. 
If either model is misspecified, the final inference results can be significantly biased. 
For instance, varying choices of hyperparameters, such as the latent dimension, can affect the accuracy of the first-stage estimation.
It is, therefore, critical to understand whether such approaches are effective in more general settings and how to address the limitations of existing post-integrated inference methods under possible misspecification.

In this paper, we investigate the validity of statistical inference on high-dimensional outcomes when unobserved confounding effects are estimated from surrogate control outcomes, thereby providing theoretical guarantees for two-step post-integrated inference. 
Further, we aim to provide a framework that not only ensures effective batch correction but also maintains the integrity and reliability of statistical inference under minimal assumptions about the data-generating processes and prediction models.
This will allow researchers to retain the statistical power of their analyses while providing greater confidence in the validity of their inferences from integrated data.

To demonstrate the challenges in post-integrated inference, we analyze high-throughput single-cell CRISPR data from a study on gene perturbations related to autism spectrum disorder and their effects on neuronal differentiation (\Cref{sec:real-data}). 
In this example, one cell can be viewed as a single dataset, where the heterogeneity among cells may not be fully explained by observed covariates.
Our analysis focuses on testing nonlinear associations between 4163 genes and \emph{PTEN} perturbation after accounting for covariates in neural development and unwanted variations from heterogeneous observations.

\begin{figure}[!t]
    \centering
    \includegraphics[width=0.7\linewidth]{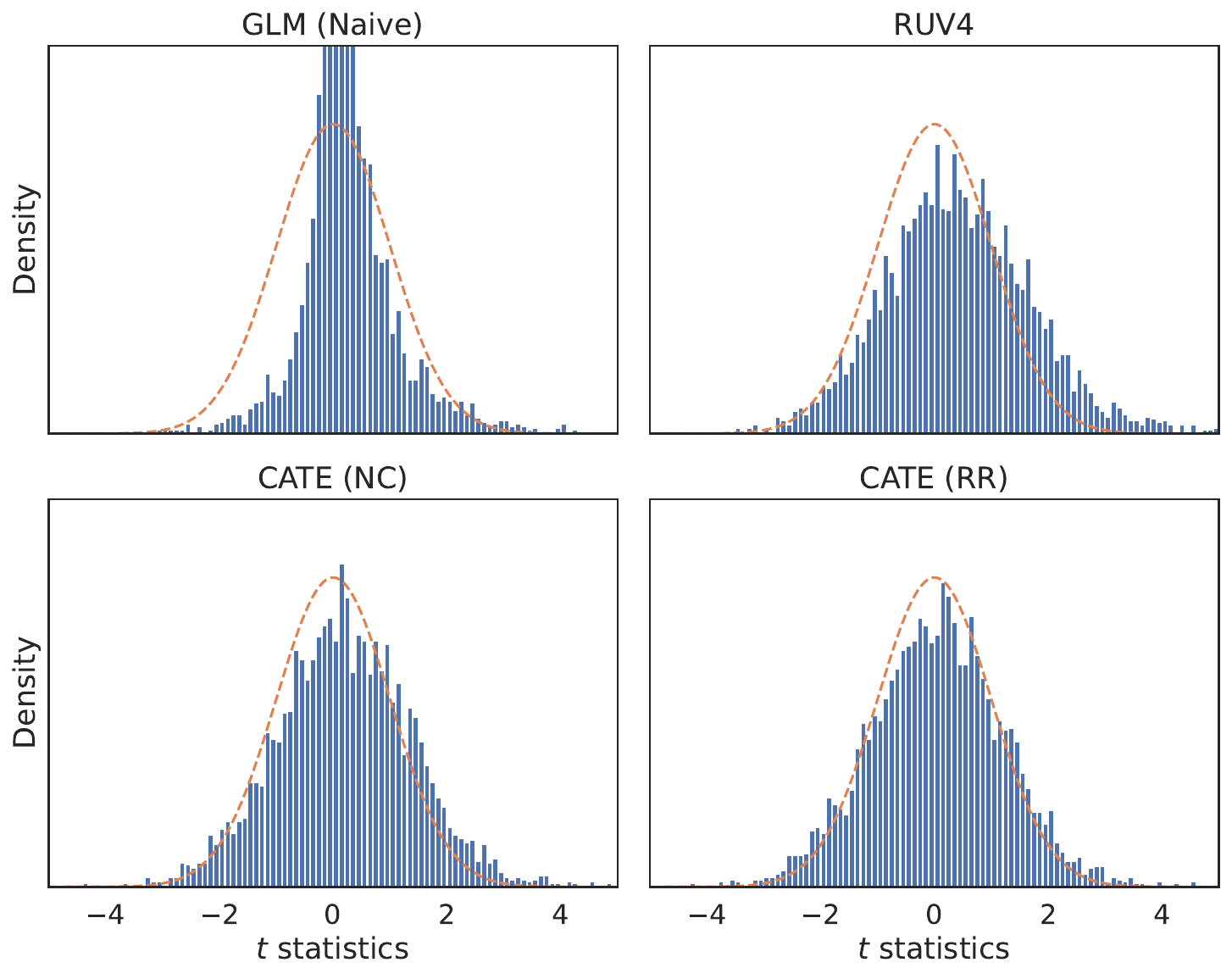}
    \caption{Histogram of $t$-statistics of \emph{PTEN} perturbation on 8320 cells and 4163 genes by four different confounder adjustment methods.
    The orange dashed curves represent the density of the standard normal distribution.
    See \Cref{sec:real-data} for more details about the methods and experiment setting.}
    \label{fig:stat-other}
\end{figure}

We would expect that different well-calibrated methods, when applied to the \emph{same} data, would all have sufficient power to detect a small proportion of stronger signals on the remaining outcomes $Y_{\cC^c}$.
\Cref{fig:stat-other} illustrates t-statistic distributions from four different methods. The unadjusted inference method yields overly conservative test statistic distributions compared to the expected $\cN(0,1)$ distribution. While batch correction and confounder adjustment methods produce distributions closer to the standard normal, some show anti-conservative tendencies. 
Importantly, only about half of the significant tests (p-values$ <0.05$) are consistent across the three confounder adjustment methods, raising concerns about their reliability.
This inconsistency stems from varied model assumptions and algorithms tailored to specific data models, which may be misspecified for sparse single-cell data or due to inaccurate estimation of the number of latent factors.
This unreliability motivates us to develop a robust statistical framework that utilizes embeddings from existing data integration methods to mitigate misspecification issues, thereby ensuring valid statistical inference and enhancing current post-integration inference methodologies.

Our work makes three main contributions. 
First, we establish nonparametric identification conditions using negative control outcomes; building on this, we propose an assumption-lean post-integrated inference (\PII) framework that relaxes standard requirements by leveraging surrogate control outcomes to address hidden mediators/confounders/moderators and batch-driven heterogeneity.
Second, we quantify the statistical error induced by using estimated embeddings in target estimands, providing bias characterization under general conditions and deterministic finite-sample bounds for linear and partial linear settings.
Third, we develop efficient semiparametric inference procedures for estimands involving estimated covariates under both linear and nonlinear link functions, enabling valid asymptotic inference for multiple treatments and high-dimensional outcomes.

\begin{figure}[!t]
    \centering
    \begin{tikzpicture}[scale=0.15]
    \tikzstyle{every node}+=[inner sep=0pt]
    \node at (15,-15) {(a)};
    \draw [black] (37.1,-33.2) circle (3);
    \draw (37.1,-33.2) node {$Y$};
    \draw [black] (21.4,-18.2) circle (3);
    \draw (21.4,-18.2) node {$X$};
    \draw [black,dashed] (21.4,-33.2) circle (3);
    \draw (21.4,-33.2) node {$U$};
    \draw [black] (23.57,-20.27) -- (34.93,-31.13);
    \fill [black] (34.93,-31.13) -- (34.7,-30.21) -- (34.01,-30.94);
    \draw [black] (24.4,-33.2) -- (34.1,-33.2);
    \fill [black] (34.1,-33.2) -- (33.3,-32.7) -- (33.3,-33.7);
    \draw [black] (21.4,-30.2) -- (21.4,-21.2);
    \fill [black] (21.4,-30.2) -- (20.9,-29.4) -- (21.9,-29.4);
    \end{tikzpicture}        
    \hspace{1cm}
    \begin{tikzpicture}[scale=0.15]
    \tikzstyle{every node}+=[inner sep=0pt]
    \node at (15,-15) {(b)};
    \draw [black] (37.1,-33.2) circle (3);
    \draw (37.1,-33.2) node {$Y$};
    \draw [black] (21.4,-18.2) circle (3);
    \draw (21.4,-18.2) node {$X$};
    \draw [black,dashed] (21.4,-33.2) circle (3);
    \draw (21.4,-33.2) node {$U$};
    \draw [black] (23.57,-20.27) -- (34.93,-31.13);
    \fill [black] (34.93,-31.13) -- (34.7,-30.21) -- (34.01,-30.94);
    \draw [black] (24.4,-33.2) -- (34.1,-33.2);
    \fill [black] (34.1,-33.2) -- (33.3,-32.7) -- (33.3,-33.7);
    \draw [black] (21.4,-30.2) -- (21.4,-21.2);
    \fill [black] (21.4,-21.2) -- (20.9,-22) -- (21.9,-22);
    \end{tikzpicture}
    \hspace{1cm}
    \begin{tikzpicture}[scale=0.15]
    \tikzstyle{every node}+=[inner sep=0pt]
    \node at (15,-15) {(c)};
    \draw [black] (37.1,-33.2) circle (3);
    \draw (37.1,-33.2) node {$Y$};
    \draw [black] (21.4,-18.2) circle (3);
    \draw (21.4,-18.2) node {$X$};
    \draw [black,dashed] (21.4,-33.2) circle (3);
    \draw (21.4,-33.2) node {$U$};
    \draw [black] (23.57,-20.27) -- (34.93,-31.13);
    \fill [black] (34.93,-31.13) -- (34.7,-30.21) -- (34.01,-30.94);
    \draw [black] (24.4,-33.2) -- (34.1,-33.2);
    \fill [black] (34.1,-33.2) -- (33.3,-32.7) -- (33.3,-33.7);    
    \end{tikzpicture}

    \vspace{3mm}
    
        \begin{tikzpicture}[scale=0.15]
        \tikzstyle{every node}+=[inner sep=0pt]
        \node at (10,-13) {(d)};
        \draw [black] (17.1,-14.5) circle (3);
        \draw (17.1,-14.5) node {$X$};
        \draw [black,dashed] (17.1,-31.5) circle (3);
        \draw (17.1,-31.5) node {$U$};
        \draw [black] (34.2,-31.5) circle (3);
        \draw (34.2,-31.5) node {$Y_{\cC}$};
        \draw [black] (34.2,-14.5) circle (3);
        \draw (34.2,-14.5) node {$Y_{\cC^c}$};
        \draw [black] (20.1,-14.5) -- (31.2,-14.5);
        \fill [black] (31.2,-14.5) -- (30.4,-14) -- (30.4,-15);
        \draw [black] (19.23,-29.38) -- (32.07,-16.62);
        \fill [black] (32.07,-16.62) -- (31.15,-16.82) -- (31.86,-17.53);
        \draw [black] (20.1,-31.5) -- (31.2,-31.5);
        \fill [black] (31.2,-31.5) -- (30.4,-31) -- (30.4,-32);
        \draw [black] (17.1,-28.5) -- (17.1,-17.5);
        \fill [black] (17.1,-17.5) -- (16.6,-18.3) -- (17.6,-18.3);
        \end{tikzpicture}
        \hspace{1cm}
        \begin{tikzpicture}[scale=0.15]
        \tikzstyle{every node}+=[inner sep=0pt]
        \node at (10,-13) {(e)};
        \draw [black] (17.1,-14.5) circle (3);
        \draw (17.1,-14.5) node {$X$};
        \draw [black,dashed] (17.1,-31.5) circle (3);
        \draw (17.1,-31.5) node {$U$};
        \draw [black] (34.2,-31.5) circle (3);
        \draw (34.2,-31.5) node {$Y_{\cC}$};
        \draw [black] (34.2,-14.5) circle (3);
        \draw (34.2,-14.5) node {$Y_{\cC^c}$};
        \draw [black] (20.1,-14.5) -- (31.2,-14.5);
        \fill [black] (31.2,-14.5) -- (30.4,-14) -- (30.4,-15);
        \draw [black] (19.23,-29.38) -- (32.07,-16.62);
        \fill [black] (32.07,-16.62) -- (31.15,-16.82) -- (31.86,-17.53);
        \draw [black] (20.1,-31.5) -- (31.2,-31.5);
        \fill [black] (31.2,-31.5) -- (30.4,-31) -- (30.4,-32);
        \draw [black,dashed] (17.1,-28.5) -- (17.1,-17.5);
        \fill [black] (17.1,-17.5) -- (16.6,-18.3) -- (17.6,-18.3);
        \fill [black] (17.1,-28.5) -- (16.6,-27.3) -- (17.6,-27.3);
        \end{tikzpicture}
    \caption{
    Post-integrated inference where the latent embedding $U$ is (a) a mediator that contributes to the indirect effect from treatment $X$ to outcome $Y$;
    (b) a confounder that affects both $X$ and $Y$; 
    and (c) a moderator that affects $Y$ but is not on the causal pathway from $X$ to $Y$.
    An embedding function $Y_{\cC}\mapsto U$ can be estimated when $Y_{\cC}\indep (Y_{\cC^c},X) \mid U$ and $Y_{\cC}\not\indep U$ if $Y_{\cC}$ are the (d) negative control outcomes when $U$ causally cause $X$ and (e) surrogate control outcomes when the causal order betwen $U$ and $X$ is arbitrary.
    }\label{fig:batch-correction}
\end{figure}

\section{Post-Integrated inference}\label{sec:PII}

    \subsection{Nonparametric identification with negative control outcomes}\label{subsec:nonpara}

    In this section, we consider post-integrated inference with negative control outcomes.
    Similar to the causal inference analysis with observational data \citep{imbens2015causal,kennedy2022semiparametric}, we consider the case when the latent variable $U$ is a confounder as in \Cref{fig:batch-correction}(b).
    Let $\cX\subseteq\RR^d$ and $\cU\subseteq\RR^r $ be the support of $X$ and $U$, respectively.
    We use $f$ to denote a generic (conditional) probability density or mass function and require causal assumptions on the observational data $(X,U,Y)$ and counterfactual outcome $Y(x)$ when $X$ is interpreted as treatment.
    
    \begin{assumption}\label{asm:causal}
        \hypertarget{asm:consistency}{(i)} (Consistency) When $X=x$, $Y=Y(x)$.
        \hypertarget{asm:positivity}{(ii)} (Positivity) $f(x\mid u) >0$ for all $u\in \cU$.
        \hypertarget{asm:nuc}{(iii)} (Latent ignorability) $X\indep Y(x)\mid U$ for all $x\in\cX$.
    \end{assumption}

\Cref{asm:causal}\hyperlink{asm:consistency}{(i)} requires that no interference among the subjects, meaning that a subject's outcome is affected by its treatment but not by how others are treated.
    \Cref{asm:causal}\hyperlink{asm:positivity}{(ii)} suggests that $X=x$ can be observed at any confounding levels of $U$ with a positive probability.
    \Cref{asm:causal}\hyperlink{asm:nuc}{(iii)} ensures that the treatment assignment is fully determined by the confounder $U$.
    These assumptions are required to estimate the counterfactual distribution of $Y(x)$ with observed variables $(X,U,Y)$ by the g-formula $ f_{Y(x)}(y) = \int f(y \mid u,x) f(u) \rd u$.
    In our problem, because $U$ is not observed, all information contained in the observed data is captured by $f(y, x)$, and one has to solve for $f(y, x, u)$ or, equivalently $f(u \mid y, x)$ from the integral equation:
    \begin{align}
        f(y, x) &= \int f(y, x, u) \rd u . \label{eq:fyx}
    \end{align}
    In general, the joint distribution $f(y, x, u)$ cannot be uniquely determined.    
    With an auxiliary variable $Z$, the approach by \citet[Theorem 1]{miao2023identifying} identifies the treatment effect from any admissible distribution under exclusion restriction, equivalence, and completeness assumptions. Here, a joint distribution $\tf (y, x, u)$ is \emph{admissible} if it conforms to the observed data distribution $f (y, x)$, that is, $f (y, x) = \int\tf (y, x, u)\rd u$.
    In particular, we leverage negative control outcomes, which can be viewed as a non-differentiable proxy for the unmeasured confounder $U$. 
    Formally, let $\cC\subset \{1,\ldots,p\}$ be the index set of known negative control outcomes and let $\cC^c = \{1,\ldots,p\} \setminus \cC$ be its complement. This partitions the full outcome vector into the negative controls $Y_{\cC} = (Y_j)_{j\in\cC}$ and the primary outcomes $Y_{\cC^c} = (Y_j)_{j\in\cC^c}$. 
    The key property of these controls is their conditional independence from the treatment and primary outcomes, given the confounder: $(Y_{\cC^c},X) \indep Y_{\cC}\mid U$.

    With this structure, the approach by \citet{miao2023identifying} applies directly to our problem by setting their auxiliary variable $Z=Y_{\cC}$. 
    However, by specializing to this negative control setting, we can extend their results to identify the counterfactual distributions under weaker assumptions, which we detail next.
    To present our first result on identification with negative control outcomes, we let $f (y, x, u ; \alpha)$ denote a model for joint distribution indexed by a possibly infinite-dimensional parameter $\alpha$, and conditional and marginal distributions are defined analogously. We require \Cref{asm:neg-outcomes}.
    \begin{assumption}\label{asm:neg-outcomes}
        The following hold for a set of control outcomes $\cC\subset\{1,\ldots,p\}$ and for any $\alpha$:  
        \hypertarget{asm:indep}{(i)} (Negative control outcomes) $(Y_{\cC^c},X)\indep Y_{\cC}\mid U$.
        \hypertarget{asm:equiv}{(ii)} (Equivalence) Any $\tf(y_{\cC}, u )$  that solves $f(y_{\cC};\alpha)=\int \tf(y_{\cC}, u ; \alpha) \rd u$ can be written as $\tf(y_{\cC}, u ) = f (y_{\cC}, v^{-1}(u); \alpha )$ for some invertible but not necessarily known function $v$.
        \hypertarget{asm:comp}{(iii)} (Completeness) 
        For all $u\in\cU$, $f(u)>0$; for any square-integrable function $g$, $\EE [ g(U) \mid Y_{\cC},X=x; \alpha] = 0$ almost surely if and only if $g(U) = 0$ almost surely.
    \end{assumption}

    The causal diagram under \Cref{asm:neg-outcomes}\hyperlink{asm:indep}{(i)} is \Cref{fig:batch-correction}(d).
    \Cref{asm:neg-outcomes}\hyperlink{asm:equiv}{(ii)} is a high-level assumption stating that at any level of covariates, the joint distribution of control outcomes and confounders lies in a class where each model is identified upon a one-to-one transformation of $U$.    
    In contrast to \citet[Assumption 2 (ii)]{miao2023identifying} that concern the joint distribution of $(X,U, Y_{\cC})$, \Cref{asm:neg-outcomes}\hyperlink{asm:equiv}{(ii)} only requires equivalence on the joint distribution of $(U,Y_{\cC})$; though we also require an extra completeness assumption on $U$ in \Cref{asm:neg-outcomes}\hyperlink{asm:comp}{(iii)} to recover an equivalent distribution of $(X,U)$.
    The completeness property plays a pivotal role in statistics \citep{lehmann2012completeness}.
    Intuitively, it precludes the degeneration of the (conditional) distributions on their supports, which guarantees the uniqueness of the solution to certain linear integral equations.    
    At different levels of $X$, \Cref{asm:neg-outcomes}\hyperlink{asm:equiv}{(ii)} requires that any infinitesimal variability in $U$ is accompanied by variability in $Y_{\cC}$, which implicitly requires the dimension of $Y_{\cC}$ to be larger than the one of $U$.
    The completeness is viewed as a regularity condition, and more detailed discussions can be found in \citet[Appendix 2]{miao2023identifying}.

    \begin{remark}[A practical interpretation of \Cref{asm:neg-outcomes}]
    To provide intuition for \Cref{asm:neg-outcomes}, consider the single-cell CRISPR experiment from \Cref{sec:real-data}. 
    The treatment ($X$) is a target gene knockdown, the latent state ($U$) captures cell-specific properties like developmental stage, the negative controls ($Y_{\cC}$) are housekeeping genes, and the primary outcomes ($Y_{\cC^c}$) are genes in the target's known pathway.

    \Cref{asm:neg-outcomes}\hyperlink{asm:indep}{(i)} requires that the housekeeping genes ($Y_{\cC}$) are not directly regulated by the treatment ($X$), but do depend on the cell state ($U$). 
    This is plausible for genes maintaining core cellular functions. 
    A violation could occur if the perturbation were toxic enough to disrupt these core functions.
    \Cref{asm:neg-outcomes}\hyperlink{asm:equiv}{(ii)}-\hyperlink{asm:comp}{(iii)} are technical conditions ensuring the housekeeping genes ($Y_{\cC}$) are informative enough to identify the low-dimensional cell state ($U$). 
    This holds if different cell states correspond to distinct expression signatures in the controls. A violation could occur if the controls respond to the cell state in a redundant manner, for instance, if their joint expression signature cannot distinguish metabolic stress from the G1 phase.
    \end{remark}
    
    Building upon the approach by \citet{miao2023identifying}, we propose a modified identification approach.

    \begin{theorem}[Nonparametric identification]\label{thm:iden-neg-outcomes}
        Under \Cref{asm:causal,asm:neg-outcomes}, for any admissible distribution $\tf(y_{\cC}, u)$ that solves $f(y_{\cC}) = \int \tf(y_{\cC}, u)\rd u$ and let $\tf(u) = \int \tf(y_{\cC}, u) \rd y_{\cC}$, there exist a unique solution $\tf(x\mid u)$ to the equation
        \begin{align}
            f(x) &= \int  \tf(x\mid u) \tf(u)\rd u.\label{eq:fyx_proxy}
        \end{align}
        Let $\tf(y_{\cC}, u \mid x) = \tf(y_{\cC}, u)\tf(x\mid u)/f(x)$, then there exists a unique solution $\tf (y_{\cC^c} \mid x, u)$ to the equation
        \begin{align}
            f(y\mid x) &= \int \tf(y_{\cC^c} \mid x,u)\tf(y_{\cC}, u \mid x) \rd u,\label{eq:fy_x_proxy}
        \end{align}
        Further, the potential outcome distribution is identified by
        \[f_{Y(x)}(y) = \int \tf(y_{\cC^c} \mid u,x)\tf(y_{\cC}, u) \rd u. \]
    \end{theorem}

    \Cref{thm:iden-neg-outcomes} suggests that if the joint distribution of $(Y_{\cC},U)$ can be estimated up to inverse transformation, then one can recover the joint distribution of potential outcome $Y(x)$.
    Based on \Cref{thm:iden-neg-outcomes}, an operational strategy is given in two steps.
    The first step is to derive $\tf(y_{\cC},u)$, which retrieves a proxy of $U$ using the information from multiple control outcomes $Y_{\cC}$.
    Given $\tf(y_{\cC},u)$, the conditional treatment distribution $\tf(x\mid u)$ and the condition outcome distribution can be obtained by solving integral equations \eqref{eq:fyx_proxy} and \eqref{eq:fy_x_proxy}.
    Even though $\tf(y_{\cC},u)$ might not be unique, the estimated conditional distributions $\tf(x\mid u)$ and $\tf (y_{\cC^c} \mid x, u)$ are guaranteed to be unique for any given $\tf(y_{\cC},u)$.
    Motivated by the nonparametric identification condition presented in \Cref{thm:iden-neg-outcomes}, we will provide a detailed description of the deconfounding strategy for recovering the true main effect under more relaxed assumptions in the next subsection.

    The deconfounding strategy given in \Cref{thm:iden-neg-outcomes} is similar to previous negative control outcome approaches \citep{wang2017confounder,zhou2024promises} under parametric modeling assumptions.
    There are two related nonparametric causal inference frameworks, as described below.

    (i) \emph{Auxiliary variables framework \citep[Section 3]{miao2023identifying}}:
    Our nonparametric identification strategy builds upon the auxiliary variables framework; however, we extend their framework to identify counterfactual distributions under weaker assumptions in \Cref{subsec:nonpara}. 
    More specifically, Theorem 1 of \citet{miao2023identifying} aims to recover the joint distribution of three variables $(Z,X,U)$, where $Z$ is an auxiliary variable that satisfies the exclusion restriction condition $Z\indep Y_{\cC^c}\mid (X,U)$.
    When $Z$ is the negative control outcome $Y_{\cC}$, we are able to factorize the joint distribution into two conditional distributions of $X\mid U$ and $Y_{\cC}\mid U$.
    This property allows us to derive nonparametric identification with weaker assumptions in \Cref{thm:iden-neg-outcomes}.
    
    (ii) \emph{Proximal inference framework \citep{miao2018identifying,miao2024confounding}}:
    Our identification result in \Cref{thm:iden-neg-outcomes} is related to, but distinct from, proximal causal inference.
    Proximal inference provides nonparametric identification by solving an integral equation under the existence of a bridge function linking negative control outcomes and negative control exposures.
    In contrast, our approach leverages variables that simultaneously serve as negative control outcomes and negative control exposures.
    In exchange for this stronger requirement, \Cref{thm:iden-neg-outcomes} identifies the distribution of counterfactuals, whereas proximal inference typically targets the expected counterfactuals.

    In the setting of \Cref{fig:batch-correction}(d), the multivariate surrogate controls $Y_{\cC}$ satisfy both roles because there are no direct causal arrows between $Y_{\cC}$ and $Y_{\cC^c}$ and no arrow $X\to U$.
    While one could partition $Y_{\cC}$ into two disjoint subsets to play the respective roles of negative control outcomes and exposures for applying proximal inference, our approach avoids this partitioning in such applications.
    We acknowledge that in other settings it may be easier to find separate negative control outcomes and exposures (as required by proximal inference) than to find variables satisfying both roles simultaneously.

    \begin{remark}[Deconfounding with multiple treatments]
        In addition to negative control outcomes, one can also deconfound using information from null treatments \citep[Section 4]{miao2023identifying}.
        When there is a single outcome, and the information of confounders solely comes from multiple (null) treatments, we can marginalize the unknown conditional distribution $f(u \mid y, x)$ over the response $y$ to obtain $f(u \mid x) = \int f(u\mid y,x) f(y\mid x) \rd y$.
        This suggests a two-stage procedure as in Section 4 of \cite{miao2023identifying}, for successively identifying solutions $f(u,x)$ and $f(y \mid u, x)$ from two integral equations: $f(x) = \int f(u, x) \rd u$ and $f(y \mid x)  = \int f(y \mid u, x) f(u \mid x) \rd u $.
        The information used to estimate the confounders in their setting is from multiple null treatments, rather than multiple outcomes.
        For this reason, they require strong assumptions to distinguish the set of confounded treatments associated with confounders.
    \end{remark}

\subsection{Assumption-Lean semiparametric inference}\label{subsec:semipara}

    The nonparametric identification results aim to reveal the counterfactual distributions from confounded observational data, which is useful for designing general deconfounding strategies, yet remains impractical.    
    Below, we provide the semiparametric tools to adjust confounding effects in practice and perform valid and efficient inference, even if the strong causal assumptions are not perfectly met.
    A leading example of semiparametric regression models is the partially linear regression \citep{robinson1988root,hardle2000partially}:
    \begin{align}
        \EE[Y \mid X, U] &=\beta^{\top} X + h(U), \label{eq:partial-linear-model}
    \end{align}
    where $Y$ is a high-dimensional vector of responses, $X$ is a low-dimensional vector of covariates (including the treatment of interest), $U\in\RR^r$ is a low-dimensional latent vector, i.e., an unmeasured confounder, $\beta\in\RR^{d\times p}$ is the coefficient to be estimated, and $h:\RR^r \rightarrow \RR^p$ is an unknown function.
    Over the past few decades, considerable attention has been given to estimating and testing partially linear models.

    When $U$ is known, the coefficient $\beta$ can be obtained with the double residual methodology \citep{robinson1988root}, by noting that
    \[
        \EE[Y\mid X,U] - \EE[Y \mid U] = \beta^{\top} (X - \EE[X\mid U]),
    \]
    More specifically, the double residual methodology proceeds in two steps: (1) regressing $Y$ on $U$ to obtain the residual $Y- \hat{\EE}[Y\mid U]$, and regress $X$ on $U$ to obtain the residual $X-\hat{\EE}[X\mid U]$; and (2) regressing the residual $Y-\hat{\EE}[Y\mid U]$ on the residual $X-\hat{\EE}[X\mid U]$. 
    Here, the notation $\hat{\EE}$ denotes the estimated regression function.
    The resulting regression coefficient is an estimator of $\beta$.
    Intuitively, this procedure removes the confounding effect of $U$ by taking the residuals, so that the final regression only captures the relationship between $X$ and $Y$ conditional on $U$, which is $\beta$ under the partial linear model assumption.
    In the special case with binary treatments, the resulting estimator is called the E-estimator \citep{robins1992estimating}.

    Even when the model \eqref{eq:partial-linear-model} is misspecified, the estimator from the two-step procedure is directly informative about the conditional association between $X$ and $U$.
     Under mild moment assumptions on the conditional covariance matrix of $X$ given $U$, it returns a meaningful estimand 
    \begin{align}
        {\beta} &= \EE[\Cov(X\mid {U})]^{-1}\EE[\Cov(X, \EE[Y \mid X,U] \mid {U})] \notag\\
        &= \EE[\Cov(X\mid {U})]^{-1}\EE[\Cov(X, Y \mid {U})], \label{eq:beta}
    \end{align}
    which itself does not crucially rely on the restrictions imposed by the outcome model \eqref{eq:partial-linear-model}.
    Furthermore, this also allows us to relax the causal relationship as detailed in \Cref{rem:causal-rela} and the strict requirement for negative controls as defined in \Cref{asm:neg-outcomes}\hyperlink{asm:indep}{(i)}. 
    We can now proceed using the more general \emph{surrogate control outcomes}, which are defined under the weaker causal assumption that no restrictions are placed on the causal order between $U$ and $X$, as demonstrated in \Cref{fig:batch-correction}(e).

    \begin{remark}[Relexation of causal relationship]\label{rem:causal-rela}
        Under the causal setting in \Cref{subsec:nonpara}, $X$ is the only shared parent of $U$ and $Y$.
        If $U$ is a moderator as in \Cref{fig:batch-correction}, adjusting for $U$ can also help to reduce the variance.
        If $U$ is a mediator, estimand \eqref{eq:beta} coincides with the controlled direct effect under the partial linear model \eqref{eq:partial-linear-model}.
        If $U$ is a confounder, it is necessary to adjust for $U$ to have a proper interpretation of the main effect of $X$ on $Y$.
        However, in general, we will not be certain whether $U$ is a confounder or not, even if $U$ were not missing.
        In particular, each entry of $U$ can either be a confounder, a mediator, or a moderator (as in \Cref{fig:batch-correction}(e)).        
        When targeting the estimand \eqref{eq:beta}, we do not need to impose specific causal assumptions.
        In contrast, \eqref{eq:beta} allows us to relax the relationship between $U$ and $X$, as long as the variability of $X$ given $U$ persists.
    \end{remark}

   Compared to causal frameworks in the previous subsection, a key practical advantage of the above strategy is its reduced reliance on strong causal assumptions that the causal diagram has to be correctly specified. 
    The statistical estimation and inference procedures targeting the projected direct effect \eqref{eq:beta} are model-free and assumption-lean, providing meaningful results even if the underlying causal model is partially misspecified.
   Because $U$ is unmeasured, we rely on the strategy offered by \Cref{thm:iden-neg-outcomes} to estimate and perform inference with surrogate control outcomes. Our deconfounding procedure is summarized in \Cref{alg:deconfounder} for general link functions.
   Below, we describe the main steps of the procedure with an identity link as a special case.

    (i) \emph{Reduction}:
    Suppose that $\cC\subseteq\{1,\ldots,p\}$ is the set of surrogate control outcomes such that $\beta_{\cC} = 0$.
    In the first step, we aim to estimate $U$ from the surrogate control outcomes $Y_{\cC}$ independently of $X$.
    To distinguish from the previous causal setting, we call $U$ the embedding of $Y_{\cC}$.
    This typically involves learning some (nonlinear) embedding map $f_e:\RR^{|\cC|}\rightarrow\RR^r$ with $Y_{\cC}\mapsto U$.

    One can use the same set of data to learn the embedding function $\hf_e$ and obtain the transformed embedding $\hat{U}= \hf_e(Y_{\cC})$.
    For example, perform the principal component analysis and use the first few principal components as the estimation embedding $\hat{U}$.
    In a more general scenario, we can also borrow extra datasets to estimate the embedding function.
    For genomic studies, many single-cell atlases of healthy cells can be used to estimate it, which helps to improve the estimation of latent embedding and is commonly used in practice for transfer learning \citep{seurat}.

    \begin{remark}[surrogate control genes]\label{rem:ncg}
        For genomic studies, housekeeping genes can serve as surrogate control outcomes.
        Furthermore, although most genes are measured, typically only the top thousands of highly variable genes are used for subsequent differential expression testing.
        It is believed that the remaining genes with low expression exhibit similar behavior under different experimental conditions.
        As we demonstrate later in \Cref{sec:real-data}, we can ideally utilize these extra genes as surrogate control outcomes to improve statistical inference.
        Of course, there is a chance that some of the genes with low expression are indeed affected by the conditions; our framework would still provide reasonable interpretability as well as robustness against such misspecification of the surrogate controls.
    \end{remark}

    \begin{algorithm}[!t]\small
    \caption{Post-Integrated inference (\PII) with surrogate control outcomes}
        \label{alg:deconfounder}
        \begin{algorithmic}[1]
        \REQUIRE
        A data set $\cD$ that contains $N$ i.i.d. samples of $(X,Y)\in\RR^{d}\times\RR^p$,
        a set of control genes $\cC\subset\{1,\ldots,p\}$,
        and a user-specified link function $g$.
        \smallskip

        \STATE Split sample  $\cD=\cD_0\cup\cD_1$ with $|\cD_0|=m,|\cD_1|=n$ and $N=m+n$; otherwise set $\cD=\cD_0=\cD_1$ and $N=m=n$.        
 
        \STATE \textbf{Estimation of the embedding functional:} 
        Based on samples in $\cD_0$, obtain an estimate $\hf_e:\RR^{|\cC|}\rightarrow\RR^r$ for the embedding map $f_e:Y_{\cC}\mapsto U$.
    
        \STATE \textbf{Extract estimated latent embeddings:} Compute $\hat{U} = \hf_e(Y_{\cC})$ on $\cD_1$.
        
        \STATE \textbf{Semiparametric inference of the main effect estimand:} 
        Use \Cref{alg:semi} to estimate
        \begin{align*}
            \tbeta_{\cdot j} &= \EE[\Cov(X \mid \hat{U} )]^{-1} \EE[\Cov(X, g(\EE[Y_j \mid X,\hat{U}] \mid \hat{U}))]\quad (j\in\cC^c),
        \end{align*}
        and the empirical variance.
        Construct the confidence interval or compute p-values according to the asymptotic distribution of $\tbeta$.
        
        \ENSURE Return the confidence intervals or p-values.
        \end{algorithmic}
    \end{algorithm}

    (ii) \emph{Estimation}:
    In the second stage, recall that our target estimand is $\beta$ in \eqref{eq:beta}.
    Because $U$ is unobserved, the best we can do is to use $\hat{U}$ as the estimated embedding and focus on the estimand:
    \begin{align}
        \tbeta_{\cdot j} &= \EE[\Cov(X \mid \hat{U} )]^{-1} \EE[\Cov(X, Y_j \mid \hat{U})]\quad (j\in\cC^c). \label{eq:beta-Uhat}
    \end{align}
    This estimand quantifies the conditional associations of $X$ and $Y$ given $\hat{U}$.
    One would typically restrict the estimation of main effects to the complement set of control genes $\cC^c$, while for notational simplicity, we simply set $\tbeta_{\cdot \cC} = 0_{d\times |\cC|}$ and present the main effect matrix $\tbeta\in\RR^{d\times p}$ in its entirety.
    Note that for $j\in\cC$, one always has $\beta_{\cdot j} = 0_d$, because $\EE[Y_{j}\mid X,U] = \EE[Y_j \mid U]$ does not depend on $X$ and the conditional covariance between $X$ and $\EE[Y_{j}\mid X,U]$ is always zero.

    (ii) \emph{Inference}:
    In the last step, to provide uncertainty quantification, we rely on the efficient influence function for $\tbeta$, similar to E-estimator \citep{chernozhukov2018double} and two-stage least squares estimators \citep{robins1992estimating,vansteelandt2022assumption}.
    The details of semiparametric inference will be given later in \Cref{subsec:dr-semi} and \Cref{subsec:dr-semi-nonlinear} for linear and nonlinear link functions, respectively.

    \begin{remark}[Assumption-lean and model-free inference]\label{rem:model-free}
         The above procedure is minimally dependent on assumptions regarding the data-generating process. It operates independently of any underlying data model, making it truly model-free.
         To compute an estimate of \eqref{eq:beta-Uhat}, arbitrary nonparametric methods can be employed to estimate the nuisance regression function. Inference can then be performed using the efficient influence function within the semiparametric framework \citep{vansteelandt2022assumption}.
         As we will see in the next section, this approach only requires mild moment conditions on the true regression function and consistency assumptions on the nuisance function estimation.
    \end{remark}

    The procedure is straightforward to understand. However, caution is warranted for nuisance regression functions and variance estimation \citep{vansteelandt2022assumption}.
    To understand the exact conditions under which this method is effective, a more sophisticated analysis is required to quantify the bias using estimated latent embeddings. 
    Additionally, theoretical guarantees of valid inference need to take into account the presence of multivariate covariates and multiple outcomes. 
    The next section serves these purposes.

\section{Statistical properties with estimated embeddings}\label{sec:theory}
    \subsection{Bias of main effects}\label{subsec:bias}

    Before presenting our analysis of the estimation errors, we introduce several technical assumptions.
    To begin with, we consider a common probability space $(\Omega,\sF,\PP)$ and let $\hat{U}_m$ explicitly indicate the dependency of $\hat{U}$ on $m\in\NN$, which is the sample size used to estimate the embedding functional $\hf_e$.
    In general, $\hat{U}_m$ can have different dimensions than $U$; to ease our theoretical analysis, we will treat the latent dimension $r$ as known so that $\hat{U}_m\in\RR^r$.
    As we will see later, such a requirement can be weakened under certain working models.
    Let $\{\sF_m\}_{m\in\NN}$ be a filtration generated by $\{\hat{U}_m\}_{m\in\NN}$ such that $\sF_m=\sigma(\hat{U}_m)$ and $\sF_1\subseteq\sF_{2}\subseteq\cdots$, and define the sub-$\sigma$-field $\sF_{\infty}=\sigma(\cup_m\sF_m) \subseteq \sF$.
    We require the following assumption.

    \begin{assumption}[Latent embedding estimation]\label{asm:latent}
        There exists a $\sF_{\infty}$-measurable and invertible function $v$ such that $\hat{U}_m \to v(U)$ almost surely.
        Further, $\ell_m = \|\hat{U}_m - v(U)\|_{\Lp{2}} <\infty$.
    \end{assumption}

    In many scenarios, when we have prior information on the embedding function $f_e$, both the number of latent dimensions and the embedding can be consistently estimated.
    For example, consistent estimation of the number of latent variables has been well established under factor models \citep{bai2002determining} and under mixture models \citep{chen2012inference}.
    Generally, a rate of $\ell_m = \Op( m ^{-\frac{1}{2}})$ can be obtained for factor analysis when there are sufficiently many surrogate control outcomes such that $|\cC|>m$ \citep{bai2012statistical}.
    For mixture models, this reduces the need to estimate the cluster membership because one can treat the one-hot vector of cluster memberships as the embedding and the cluster centers as the loading, akin to factor analysis.
    When $f_e$ is estimated nonparametrically by $\hf_e$, the estimated embedding $\hat{U}_m$ can be viewed as nonparametrically generated covariates.
    In this context, \Cref{asm:latent} only requires the (conditional) $L_2$-norm of the estimation error $\hat{f}_e-f_e$ decays to zero in probability to ensure meaningful and accurate estimation of $U$, which is weaker than Assumption 2 of \citet{mammen2012nonparametric} that requires the (conditional) $L_{\infty}$-norm of $\hat{f}_e-f_e$ is $\op(1)$.
    Finally, we also note that one can utilize additional data sources to obtain a more accurate estimate of $\hf_e$ with a larger sample size $m$.
    In many applications, such as single-cell data analysis, the embedding function can be derived from previous studies, ensuring that $m$ is sufficiently large.

    The following \Cref{asm:moment} imposes a boundedness condition on the population quantities, and \Cref{asm:Lipshitz-cond-mean} imposes a smoothness assumption on the regression function.

    \begin{assumption}[Regularity conditions] \label{asm:moment}
        There exists constants $\bar{\sigma}\geq \sigma>0$ and $M>0$ such that $\sigma I_d \preceq \EE[\Cov(X\mid U)] \preceq \bar{\sigma} I_d$, $\sigma I_d \preceq \EE[\Cov(X\mid \hat{U}_m)]$, $\|\beta\|_{2,\infty} \leq M$, $\|X\|_{\Lp{2}}\leq M,\max_{j\in \cC^c}\|Y_j\|_{\Lp{2}}\leq~M$.
    \end{assumption}

    \begin{assumption}[Lipschitzness of regression functions]\label{asm:Lipshitz-cond-mean}
        The regression functions satisfy Lipschitz conditions:
        \begin{align*}
            \|\EE[X\mid U=u_1] - \EE[X\mid U=u_2]\| &\leq L_X\|u_1 - u_2\|\\
            \|\EE[Y_j\mid X, U=u_1] - \EE[Y_j\mid X, U=u_2]\| &\leq L_{Y}\|u_1 - u_2\|\quad (\forall\ j\in\cC^c),
        \end{align*}
        almost surely for all $u_1,u_2\in \cU$ and some constants $L_X$ and $L_{Y}$.
    \end{assumption}

    \Cref{asm:Lipshitz-cond-mean} imposes certain smoothness restrictions on the conditional expectation.    
    In certain applications, the Lipschitz condition holds for many continuous multivariate distributions.    
    For example, suppose $W$ and $V$ are jointly normally distributed with
    \begin{align*}
        \begin{pmatrix}
            W\\
            V
        \end{pmatrix} \sim 
        \cN\left( 
            \begin{pmatrix}
                \mu_W\\ \mu_V
            \end{pmatrix},
            \begin{pmatrix}
            \Sigma_W & \Sigma_{WV}\\
            \Sigma_{WV}^{\top} & \Sigma_V
        \end{pmatrix}   
        \right).
    \end{align*}
    Then $h(v)=\EE[W\mid V=v] = \mu_W + \Sigma_{WV}\Sigma_V^{-1} (v-\mu_V)$ is $L$-Lipschitz in $\ell_2$-norm, with $L= \|\Sigma_{WV}\Sigma_V^{-1}\|$.
    Other examples of such a regression function include the posterior mean of the exponential and Poisson distributions under their conjugate prior, as in Bayesian inference.
    Similar conditions have been employed for nonparametric regression with generated covariates; see, for example, Assumption 4 in \citet{mammen2012nonparametric}.
    In particular, \citet{mammen2012nonparametric} require differentiability and Lipschitz condition in $\ell_{\infty}$ of the condition expectation, which is much stronger than \Cref{asm:Lipshitz-cond-mean}.

    Consider two population coefficients $\beta$ and $\tbeta$ as defined in \eqref{eq:beta} and \eqref{eq:beta-Uhat}, respectively.
    We next quantify the difference between the two in \Cref{thm:err-bound-beta-U}.

    \begin{theorem}[Bias for two-stage regression with estimated covariates]\label{thm:err-bound-beta-U}
        Under \Cref{asm:latent,asm:moment,asm:Lipshitz-cond-mean}, when $\|\EE[X\mid \hat{U}] - \EE[X\mid U]\|_{\Lp{2}} < \sigma/(2M)$, it holds that
        \begin{align*}
            \max_{j\in \cC^c}\|\tbeta_{\cdot j} - \beta_{\cdot j}\|&\lesssim 
            \left(\| X \|_{\Lp{2}}(L_X^{\frac{1}{2}}+L_{Y}^{\frac{1}{2}})  +  \max_{j\in\cC^c}\| Y_j \|_{\Lp{2}} L_{Y}^{\frac{1}{2}} \right)\ell_m .
        \end{align*}
    \end{theorem}
    
    \Cref{thm:err-bound-beta-U} suggests that the upper bound of estimation error using estimated embeddings is related to the second moments of $X$ and $Y$, as well as the accuracy of latent embedding estimation.
    This deterministic result only concerns the population quantities.
    Given i.i.d. samples of $(X,U,Y)$, the corresponding estimator of $\beta_{\cdot j}$ based on finite samples is
    \begin{align}
        b_{\cdot j} &= (\PP_n\{ (X - \hat{\EE}[X\mid {U}] )^{\otimes 2} \} )^{-1} \PP_n\{ (X - \hat{\EE}[X\mid {U}] )(Y_j - \hat{\EE}[Y_j\mid U] ) \} ,\label{eq:b-U}
    \end{align}
    where $A^{\otimes2} = AA^{\top}$ denotes Gram matrix of $A^{\top}$, and $\hat{\EE}[X\mid U]$ and $\hat{\EE}[Y\mid U]$ are the estimated nuisance functions.
    Because $U$ is unobserved, we treat $\hat{U}$ as the truth and estimate $\tbeta_{\cdot j}$ with:
    \begin{align}
        \tilde{b}_{\cdot j} &= (\PP_n\{ (X - \hat{\EE}[X\mid \hat{U}] )^{\otimes 2} \} )^{-1} \PP_n\{ (X - \hat{\EE}[X\mid \hat{U}] )(Y_j - \hat{\EE}[Y_j\mid \hat{U}] ) \}. \label{eq:b-Uhat}
    \end{align}

    As an example, we consider a special case when the regression functions are linear models.
    To distinguish from previous notations, we denote the latent embedding matrix $\bU\in\RR^{n\times r}$ and its estimate $\hat{\bU}\in\RR^{n\times \hat{r}}$, where the latter may have a different dimension $\hat{r}$ than the truth $r$.
    \Cref{lem:est-err-U_hat} below shows that we are still able to quantify the empirical estimation error of the main effects in terms of the estimation error of linear projection matrices in finite samples.

    \begin{lemma}[Empirical bias with estimated embeddings under linear nuisance estimators]\label{lem:est-err-U_hat}        
        Define $S = \PP_n\{ (X - {\EE}[X\mid {U}] )^{\otimes 2} \}$, $\tilde{S} = \PP_n\{ (X - {\EE}[X\mid \hat{U}] )^{\otimes 2} \}$, and $\Gamma = \diag(\PP_n\{YY^{\top}\})$.
        Assume $S$ and $\tilde{S}$ have full rank, and $\kappa(S)\|P_{\hat{\bU}}^{\perp}-P_{\bU}^{\perp}\|<1$, where for any matrix $A\in\RR^{n\times p}$,
        $P_{A}=A(A^{\top}A)^{-1}A^{\top}$ denotes the projection matrix and $\kappa(A) = \|A\|\|A^{-1}\|$ denotes the condition number of matrix $A$.
        When the nuisance estimators $\hat{\EE}[X\mid {U}]$, $\hat{\EE}[X\mid \hat{U}]$, and $\hat{\EE}[Y\mid U]$ are linear functions, it holds that
        \begin{align*}
            \max_{j\in\cC^c}\|\tilde{b}_{\cdot j} - b_{\cdot j}\|&\leq \left(\|b\|_{2,\infty}  + \|S\|_{\oper}^{-\frac{1}{2}}\|\Gamma\|_{\infty}\right) \frac{\kappa(S) \|P_{\hat{\bU}}^{\perp}-P_{\bU}^{\perp}\|}{1 - \kappa(S)  \|P_{\hat{\bU}}^{\perp}-P_{\bU}^{\perp}\|}, 
        \end{align*}
        where $\|A\|_{2,\infty} = \max_{j\in\{1,\ldots,p\}}\|A_{\cdot j}\|$ is the maximum column euclidean norm for $A\in\RR^{n\times p}$.
    \end{lemma}

    To the best of our knowledge, 
    \Cref{lem:est-err-U_hat} provides the first deterministic, finite-sample result that quantifies the estimation error for the main effects when substituting the true latent variables with estimated embeddings.
    The result also applies to the partial linear model \eqref{eq:partial-linear-model} when it is assumed.
    Compared to \Cref{thm:err-bound-beta-U}, \Cref{lem:est-err-U_hat} suggests that the rate condition of $\hat{\bU}$ can be weakened to the rate condition of the linear projection $P_{\hat{\bU}}^{\perp}$.
    The conclusion of \Cref{lem:est-err-U_hat} is fully deterministic, and its proof relies on the backward error analysis in numerical linear algebra \citep{trefethen2022numerical}.
    The dimension of the estimated embedding is allowed to differ from the truth, as long as the column space of $\hat{U}$ captures essential information of the column space of $U$.
    Analogously, it is possible to relax \Cref{asm:latent} to varying latent dimension settings for \Cref{thm:err-bound-beta-U} under general data models.
    In this regard, one can consider a decomposition of $\lim_m\hat{U}_m = T + A$, where $T$ and $A$ are a sufficient statistic and an ancillary statistic, respectively, as when $U$ is viewed as a parameter.
    We leave such an extension as future work.

    \subsection{Doubly robust semiparametric inference}\label{subsec:dr-semi}

    In the previous section, we showed that the target estimands $\tbeta$ and $\beta$ are similar whenever $\hat{U}$ is consistent to $U$ up to any invertible transformation. 
    Based on the estimated embedding $\hat{U}$, our target of estimation and inference becomes $\tbeta$ as defined in \eqref{eq:beta-Uhat}.
    To consider potential nonparametric models for the nuisance functions, in what follows, we require the estimated nuisance functions $\hat{\EE}[X\mid U]$ and $\hat{\EE}[Y\mid U]$ to be computed from independent samples of $\PP_n$.
    The required independence is standard in recent developments of double machine learning and causal inference \citep{vansteelandt2022assumption,kennedy2022semiparametric} where sample splitting and cross-fitting can be used to fulfill this requirement, though one can also restrict to Donsker classes to avoid sample splitting.

    Before we inspect the estimation error of $\tilde{b}$ to the target estimand $\tbeta$, we introduce one extra assumption on the moments and consistency of nuisance estimation.

    \begin{assumption}[Bounded moments and consistency]\label{asm:nuisance}
        There exists $\delta\in(0,1]$, $M>0$, such that 
        \[
            \|X - \EE[X\mid \hat{U}]\|_{\Lp{2(1+\delta^{-1})}} 
            \vee
            \|X - \hat{\EE}[X\mid \hat{U}]\|_{\Lp{2(1+\delta^{-1})}}
            \vee
            \|Y - {\EE}[Y\mid \hat{U}]\|_{\Lp{2(1+\delta^{-1})}} < M,
        \]
        \[
            \|\EE[X\mid \hat{U}] - \hat{\EE}[X\mid \hat{U}]\|_{\Lp{2(1+\delta)}}, \| \|{\EE}[Y\mid \hat{U} ] - \hat{\EE}[Y\mid \hat{U}]\|_{\infty}\|_{\Lp{2(1+\delta)}}=\op(1).
        \]
    \end{assumption}

    Let $O=(X,\hat{U},Y)\in\RR^d\times\RR^r\times\RR^p$ denote the observation when the estimated embedding function $\hf_e$ is treated as fixed.
    The following theorem demonstrates the linear expansion of the estimator $\tilde{b}$ and provides an error bound for the residual term with high probability.

    \begin{theorem}[Linear expansion]\label{thm:DR-linear}
        Consider the above inferential procedure, suppose \Cref{asm:moment,asm:nuisance} hold  and two nuisance functions $\hat{\EE}[X\mid \hat{U}]$ and $\hat{\EE}[Y\mid \hat{U}]$ are estimated from independent samples of $\PP_n$.
        Then, the estimator $\tilde{b}$ admits a linear expansion:
        \begin{align*}
            \sqrt{n}(\tilde{b} - \tbeta) &= \sqrt{n}\tilde{\Sigma}^{-1}(\PP_n-\PP)\{\tilde{\varphi}(O;{\PP})\} + \xi,
        \end{align*}
        where $\tilde{\Sigma} = \EE[\Cov(X\mid \hat{U})]$ and $\tilde{\varphi}$ is the influence function of $\tSigma\tbeta$ defined as
        \begin{align}
            \tilde{\varphi}(O;\PP) &= (X - {\EE}[X\mid \hat{U}]) ( (Y - \EE[Y \mid X]) - \tbeta^{\top}(X - {\EE}[X\mid \hat{U}]) )^{\top} .\label{eq:tvarphi-linear}
        \end{align}
        For any $\epsilon>0$, there exists a constant $C=C(\epsilon,\sigma,M,L)$, such that with probability at least $1- \epsilon$, the remainder term $\xi$ satisfies that 
        \begin{align*}
            \|\xi\|_{2,\infty}
            &\leq 
             C \{ \|(\PP_n-\PP)\{(X-\EE[X\mid \hat{U}])^{\otimes 2}\}\|_{\oper} \\
             &\qquad + 
              \| {\EE}[X\mid \hat{U}] - \hat{\EE}[X\mid \hat{U}]\|_{\Lp{2(1+\delta)}} + \| \|\EE[Y\mid \hat{U}]-  \hat{\EE}[Y\mid \hat{U}] \|_{\infty}\|_{\Lp{2(1+\delta)}} \} \\
            &\qquad + C\sqrt{n}\{
            \|\EE[X\mid \hat{U}] - \hat{\EE}[X\mid \hat{U}]\|_{\Lp{2}}^2 + ML\|{\EE}[Y\mid \hat{U}] - \hat{\EE}[Y\mid \hat{U}]\|_{\Lp{2},\infty}^2 \notag\\
            &\qquad + \| \EE[Y\mid \hat{U}]-  \hat{\EE}[Y\mid \hat{U}]\|_{\Lp{2},\infty}  \|\EE[X\mid \hat{U}] - \hat{\EE}[X\mid \hat{U}]\|_{\Lp{2}} \}.
        \end{align*}
    \end{theorem}
    
    \Cref{thm:DR-linear} provides a non-asymptotic uniform error bound for the residual terms over multiple outcomes.
    With the law of large numbers and the consistency in \Cref{asm:nuisance}, we know that the first term of the upper bound is $\op(1)$.
    On the other hand, the secondary term is also negligible under specific rate conditions on the estimation errors of nuisances. 
    Considering an asymptotic regime when viewing $m$ and $p$ as sequences indexed by $n$ and $n,m,p\rightarrow\infty$, the above result suggests the asymptotic normality, as presented in the following corollary.

    \begin{corollary}[Doubly robust inference with estimated emebeddings]\label{cor:inference}
        Under conditions in \Cref{thm:DR-linear}, if further, the estimation error rates of nuisance functions satisfy that
        $\|\EE[X\mid \hat{U}] - \hat{\EE}[X\mid \hat{U}]\|_{\Lp{2}}^2=\op(n^{-\frac{1}{2}})$, $\|{\EE}[Y\mid \hat{U}] - \hat{\EE}[Y\mid \hat{U}]\|_{\Lp{2},\infty}^2 =\op(n^{-\frac{1}{2}})$, $\| \EE[Y\mid \hat{U}] - \hat{\EE}[Y\mid \hat{U}] \|_{\Lp{2},\infty}  \|\EE[X\mid \hat{U}] - \hat{\EE}[X\mid \hat{U}]\|_{\Lp{2}} =\op(n^{-\frac{1}{2}})$, then the estimator $\tilde{b}$ is asymptotically normal:
        \[\sqrt{n}(\tilde{b}_{\cdot j} - \tbeta_{\cdot j}) \to \cN_d (0, \tSigma^{-1}\VV\{\tilde{\varphi}_{\cdot j}(O;{\PP})\}\tSigma^{-1})\text{ in distribution}\quad (j=1,\ldots,p).\]        
        Furthermore, if the conditions of \Cref{thm:err-bound-beta-U} hold with $\ell_m=o(n^{-\frac{1}{2}})$, then we have
        \begin{align*}
            \sqrt{n} (\tilde{b}_{\cdot j} - {\beta}_{\cdot j}) \to \cN_d (0, \tSigma^{-1}\VV\{\tilde{\varphi}_{\cdot j}(O;{\PP})\}\tSigma^{-1})\text{ in distribution}\quad (j=1,\ldots,p).
        \end{align*}
    \end{corollary}

    In the presence of the estimated embedding $\hat{U}_m$, the influence function $\tilde{\varphi}$ implicitly depends on the sample size $m$.
    Therefore, establishing the asymptotic normality requires verification of the Lindeberg condition for a triangular array of random variables.
    In \Cref{cor:inference}, the rate of estimation for the two nuisance functions may be slower than the parametric rate $n^{-\frac{1}{2}}$, as long as each individual estimation rate is faster than $n^{-\frac{1}{4}}$. This flexibility enables us to employ more versatile machine learning algorithms for nuisance function estimation while maintaining the validity of our inference.
    Furthermore, \Cref{cor:inference} suggests that efficient influence regarding the true main effect $\beta$ is possible when the rate of consistently estimating the embedding is $\ell_m=\op(n^{-\frac{1}{2}})$.
    As discussed above, under factor models, one has $\ell_m = \Op( m^{-\frac{1}{2}})$, this requires $n = o(m)$, i.e., the factor loadings need to be estimated from more observations than those used for the estimation and inference of $\tilde{b}$.

    Based on \Cref{cor:inference}, the data-adaptive procedure to obtain the confidence intervals and p-values is given in \Cref{alg:semi-linear}.
    To fulfill the independence assumptions, one can use cross-fitting to ensure that different samples are used for step \ref{step1} and step \ref{step2}.
    When this holds, the following proposition shows that overall Type-I error control can be controlled at the desired level.
    In \Cref{prop:simul-inference}, when the unit vector $v$ is chosen to be the basis vector, it reduces to testing whether a specific covariate has zero association with individual outcomes.
    
    \begin{proposition}[Multiple linear hypothesis testing]\label{prop:simul-inference}
        Let $t_j = \sqrt{n}\VV_n\{\tilde{\varphi}_{\cdot j}(O;\hat{\PP})\}^{\frac{1}{2}}\hSigma^{-1}$ $(\tilde{b}_{\cdot j} - \tbeta_{\cdot j})$ be the standardized vector.
        For any unit vector $v\in\RR^d$, consider the hypothesis $\cH_{0j}:v^{\top}\beta_{\cdot j} =0$.
        Let $\cN_p=\{j\mid v^{\top}\beta_{\cdot j}=0,j=1,\ldots,p\}$ be the true null hypotheses.
        Under the assumptions of \Cref{cor:inference}, as $m,n,p,|\cN_p|\rightarrow\infty$ such that $\ell_m=o(n^{-1/2})$, it holds that $|\cN_p|^{-1}\sum_{j\in\cN_p} \ind\{|v^{\top}t_j|>z_{\frac{\alpha}{2}}\} \to \alpha$ in probability.
    \end{proposition}

    \begin{remark}[Multiple testing]\label{rmk:mul-test}
        The condition $ \| \|{\EE}[Y\mid \hat{U} ] - \hat{\EE}[Y\mid \hat{U}]\|_{\infty}\|_{\Lp{2(1+\delta)}}=\op(1)$ in \Cref{asm:nuisance} controls the envelope of the regression function estimation errors.
        This is useful when the number of outcomes $p$ grows with the sample size $n$, when multiple testing procedures based on multiplier bootstrap can be applied to control both the family-wise error rate and the false discovery rate \citep{du2024causal}.
        Alternatively, one can apply the Benjamini–Hochberg procedure for multiple testing corrections.
    \end{remark}

    \begin{algorithm}[t]\small
        \caption{Semiparametric inference for main effects}
        \label{alg:semi-linear}
        \begin{algorithmic}[1]
        \REQUIRE Reponses $Y$, covariate $X$, and estimated latent embedding $\hat{U}$.
        
        \STATE\label{step1} Use machine learning methods to obtain nuisance estimates $\hat{\EE}[Y\mid \hat{U}]$ and~$\hat{\EE}[X\mid\hat{U}]$.

        \STATE\label{step2} Fit a linear regression of $Y - \hat{\EE}[Y\mid \hat{U}] \sim X - \hat{\EE}[X\mid \hat{U}]$ without an intercept to obtain an estimate $\tilde{b}$ as defined in \eqref{eq:b-Uhat} of $\tbeta$ as defined in \eqref{eq:b-Uhat}.

        \STATE Estimate the variance of $\tilde{b}_{\cdot j}$ by $\hat{S}_j/n$ based on \Cref{thm:DR-linear}, where $\hat{S}_j = \hSigma^{-1}\VV_n\{\tilde{\varphi}_{\cdot j}(O;\hat{\PP})\}\hSigma^{-1} $ and $\hSigma=\PP_n\{ (X - \hat{\EE}[X \mid \hat{U}])^{\otimes 2}\}$. 

        \ENSURE Confidence intervals and p-values based on asymptotic null distribution $\tilde{b}_{\cdot j} \overset{\cdot}{\sim}\cN_d (\tbeta_{\cdot j}, \frac{\hat{S}_j}{n})$. 
        \end{algorithmic}
    \end{algorithm}

\section{Simulation}\label{sec:simu}

    We generate the data from generalized partial linear models.
    The covariate $X\in\RR$ is sampled from $\cN(0,1)$; 
    the latent variable $U = X\alpha + \epsilon\in\RR^r$ is a linear function of $X$, where $r=10$, $\alpha_{1j}\sim\Unif(-1,1)$ and $\epsilon_j\sim\cN(0,\sigma^2_\epsilon)$ independently for $j\in[r]$; 
    and the response is generated from generalized linear models with a Logistic link $\mathrm{logit} (\EE[Y \mid X,U]) = X\beta + U\eta$, where $\beta_{1j}\sim 2\times \mathrm{Bernoulli}(0.2)$ and $\sqrt{r}\cdot\eta_{ij}\sim\Unif(-1,1)$ independently for $i\in[r]$ and $j\in\{1,\ldots,p\}$.
    We generate $p=1000$ outcomes and use 500 null outcomes as the surrogate outcomes.
    
    We evaluate four methods: (1) \GLM ($X$): naive generalized linear models that use Logistic regression that only uses observed covariate $X$ to predict $Y$; 
    (2) \GLM ($X,U$): oracle Logistic regression that uses both observed covariate $X$ and latent variable $U$ to predict $Y$;
    (3) \PII ($X,U$): the proposed post-integrated inference method that uses observed covariate $X$ and latent embedding $U$ to predict $Y$; 
    and (4) \PII ($X,\hat{U}$): the proposal method that uses the first $r$ PCs of the outcome matrix are selected as $\hat{U}$.

    For \PII, we use the random forest to estimate the nuisance functions $\EE[X\mid U]$, $\EE[Y\mid X,U]$, and $\EE[g(\EE[Y\mid X,U])\mid U]$ and apply extrapolated cross-validation (\ECV) \citep{du2024extrapolated} to select the hyperparameter that minimizes the estimated mean squared error.
    \ECV allows us to use a smaller number of trees for estimating the out-of-sample prediction errors based on out-of-bag observations and extrapolate the risk estimation up to a larger number of trees consistently without sample splitting.
    In our experiment, we use $25$ trees to perform ECV and the hyperparameters we consider include: `max\_depth' in $\{1,3,5\}$ for the depth of each tree, `max\_samples' in $\{0.25,0.5,0.75,1\}$ for bootstrap samples and the number of trees in $\{1,\ldots,50\}$.

    \begin{figure}[!t]
        \centering
        \includegraphics[width=\linewidth]{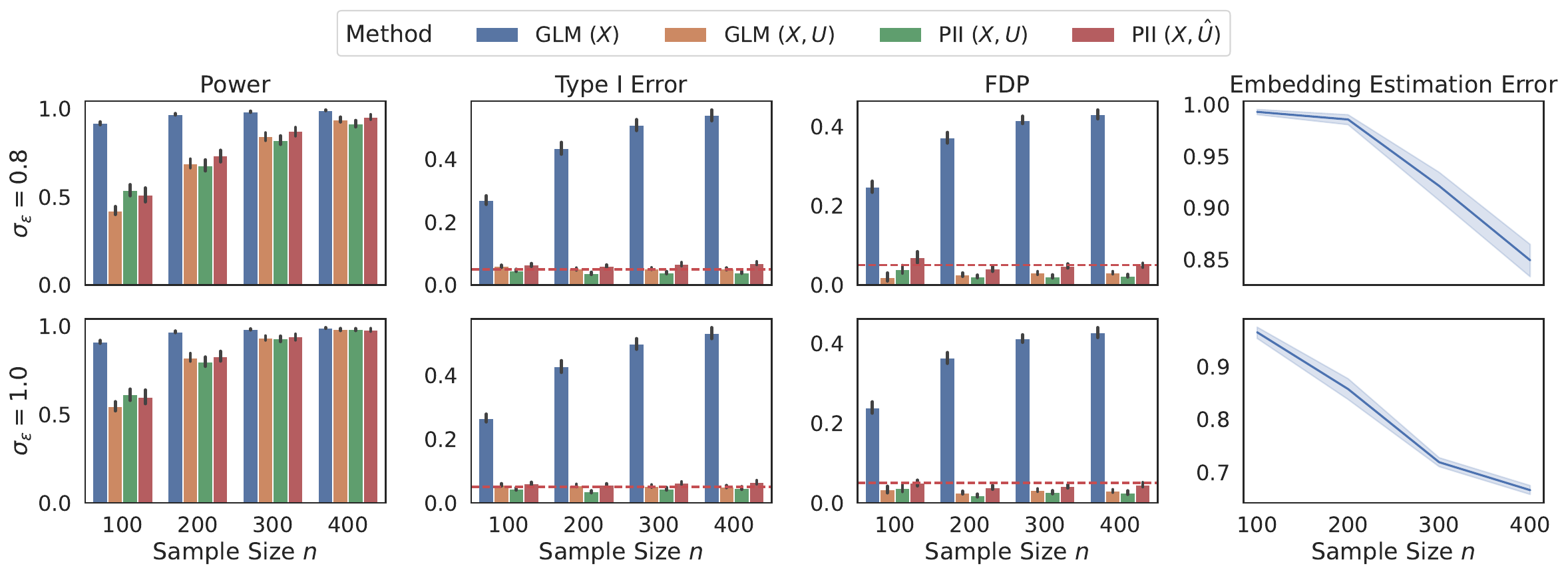}
        \caption{
        Simulation results with $500$ surrogate control outcomes out of a total of $1000$ outcomes.
        For \PII, the nuisance functions are estimated using random forests.
        The data model is the Logistic regression model.
        The first and second rows have noise levels $\sigma_{\epsilon}=0.8$ and $\sigma_{\epsilon}=1$, respectively,  for the latent variables.
        }
        \label{fig:simu-1}
    \end{figure}

    To compare the performance of different methods, the power, type-I error, and false discovery proportion (FDP) for hypothesis testing are analyzed.
    For both the type-I error and power, we set the significance level to be 0.05.
    For FDP, we use the Benjamini-Hochberg procedure with FDR controlled at 0.05.
    As shown in the first two columns of \Cref{fig:simu-1}, the \GLMnaive regression method fails to control the inflated type-I error, resulting in numerous false positives. Furthermore, as the sample size increases, this method becomes even more anti-conservative.
    Conversely, the \GLMoracle regression method exhibits tight control over type-I error, as expected. When the latent embedding $U$ is known, we observe that \PII also effectively controls type-I error. Additionally, under certain conditions, \PII provides greater power than the \GLMoracle. 
    This may be attributed to \PII's ability to address the effect of collinearity between $X$ and $U$ on the nonlinear outcome models through a two-step procedure, whereas \GLMoracle does not, leading to conservative results; see \Cref{app:ex-simu} for more discussions.

    When the latent embedding $U$ is unknown, we evaluate the performance of the estimated $\hat{U}$.
    As shown in the third panel of \Cref{fig:simu-1}, the error of embedding projection matrix $\|P_{\hat{\bU}} - P_{\bU}\|_{\oper}$ decreases rapidly as the sample size $n$ increases. 
    When $U$ can be well approximated, \PII experiences a slightly inflated type-I error because it targets the modified main effect $\tbeta$ instead of the true effect $\beta$. 
    However, the statistical error remains reasonable, the FDP is controlled at the desired level, and \PII achieves greater power compared to the oracle \GLM in many cases.
    Lastly, \PII exhibits greater power when the conditional variation of $X$ given $U$ is large (i.e., $\VV(\epsilon)$ is relatively larger than the linear projected signal strength $\|\gamma\|$).
    One could potentially use the ratio of these two quantities as a metric to quantify the level of confounding.

    Finally, we assessed the robustness of \PII to misspecified surrogate controls. We simulated two scenarios: one in which the control set was contaminated with non-null genes, and another in which controls were selected empirically based on outcome variability. As shown in \Cref{fig:mis}, \PII demonstrates strong robustness in both settings, maintaining excellent error control with minimal loss of power. This underscores the method's practical reliability when the set of surrogate controls is imperfectly defined.

\section{Application on single-cell CRISPR data analysis}\label{sec:real-data}

    In a recent study by \citet{lalli2020high}, the molecular mechanisms of genes associated with neurodevelopmental disorders, particularly Autism Spectrum Disorder (ASD), were investigated using a modified CRISPR-Cas9 system. Experiments focused on 13 ASD-linked gene knockdowns in Lund Human Mesencephalic neural progenitor cells, with gene expression changes assessed through single-cell RNA sequencing. The progression of neuronal differentiation was estimated using a pseudotime trajectory (\Cref{fig:pseudotime}), which revealed that certain genetic perturbations impact this progression.

    While CRISPR screens are experiments, they exhibit a crucial observational nature at the single-cell level. The treatment, the presence of a specific guide RNA in a cell, is only identified post hoc from sequencing. 
    Therefore, factors that correlate with both the guide’s presence and the outcome, such as cell size, cycle stage, and microenvironmental heterogeneity, can act as confounders or mediators; see \Cref{rmk:mediator} for further discussion.
    To address these challenges, we use 4000 lowly variable genes as surrogate control outcomes for adjustment, focusing on 4163 highly variable genes for differential expression analysis on 8320 cells. 
    The data preprocessing procedure is detailed in \Cref{app:ex-real-data}.

    \begin{figure}[!t]
        \centering
        \includegraphics[width=0.7\linewidth]{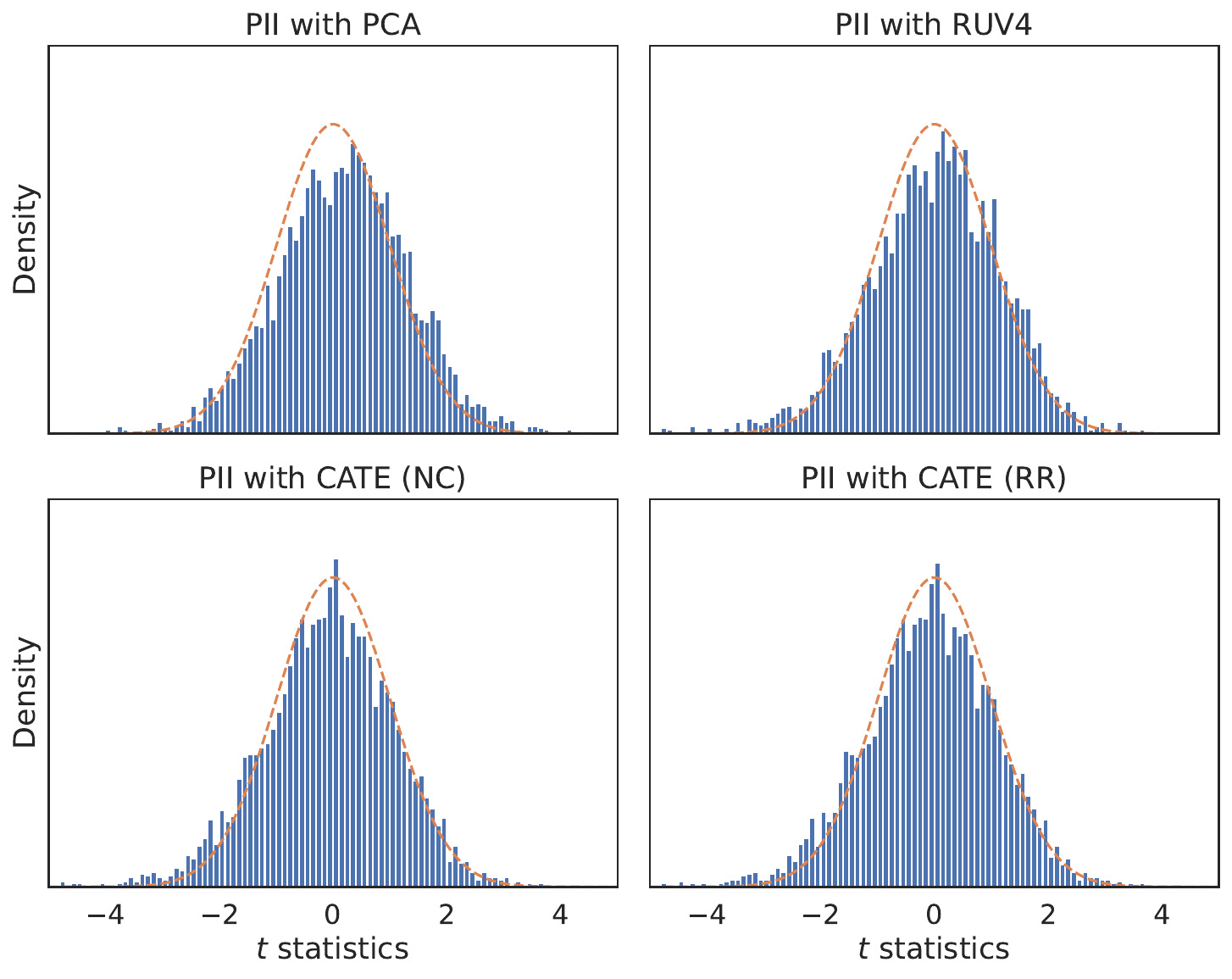}
        \caption{Histogram of $t$-statistics of \emph{PTEN} perturbation by different methods.
        PCA with 50 components, \RUV, \CATENC, and \CATERR.}
        \label{fig:stat-PII}
    \end{figure}

    We compare the proposed method with four methods for hypothesis testing:
    (1) \GLM: Score tests based on generalized linear models with Negative Binomial likelihood and log link function.
    The covariance matrix is estimated using the HC3-type robust estimator. 
    This method does not adjust for potential confounding effects.
    (2) \RUV: A heuristic method proposed by \citet{gagnon2012using} that uses principal components on the residual matrix of regressing the surrogate control outcomes on the covariate of interest to estimate the latent embeddings.
    Based on heuristic calculations, the authors claim that the RUV-4 estimator has approximately the oracle variance.
    (3) \CATENC: The deconfounding method \CATE proposed by \citet{wang2017confounder} with surrogate controls, which uses maximum likelihood estimation to estimate the latent embedding.
     Under simplified Gaussian linear models, they show that their estimator has asymptotic type I error control when the number of surrogate controls is large.
    (4) \CATERR: A variant of \CATE method \citep{wang2017confounder} with robust regression, which is similar to the heuristic algorithm LEAPP \citep{sun2012multiple} and utilizes the sparsity of effects to estimate the latent embeddings.

    For \PII, we first estimate the cell embeddings $\hat{U}$ using one of four methods (PCA, \RUV, \CATENC, or \CATERR) on the full dataset. We do not use sample splitting for this step, as our goal is to utilize all available information to construct the most accurate embeddings possible. Consequently, our target estimand is $\tbeta$. 
    The nuisance functions, including $\hat{\EE}[X \mid \hat{U}]$ and $\hat{\EE}[Y \mid \hat{U}]$, are then estimated using random forests, following the same procedure detailed in the previous simulation section.
    The first three methods utilize surrogate control to estimate the embedding, whereas the last is only valid under the assumption of sparsity in the effects.
    Before running PCA, we follow the preprocessing procedure in single-cell data analysis to adjust the library size of each cell to $10^4$, add one pseudo-count, and take the logarithm.
    We then select the top 50 principal components as the estimated embeddings.
    For the last three embedding estimation methods, we supply all 13086 genes as input, specify the set of surrogate control genes when applicable, and set the number of factors to 10, a value commonly used by researchers based on empirical evidence.
    Our sensitivity analyses show similar results with different numbers of factors or surrogate control outcomes in \Cref{app:ex-real-data}.

    \begin{figure}[!t]
        \centering
        \includegraphics[width=.7\textwidth]{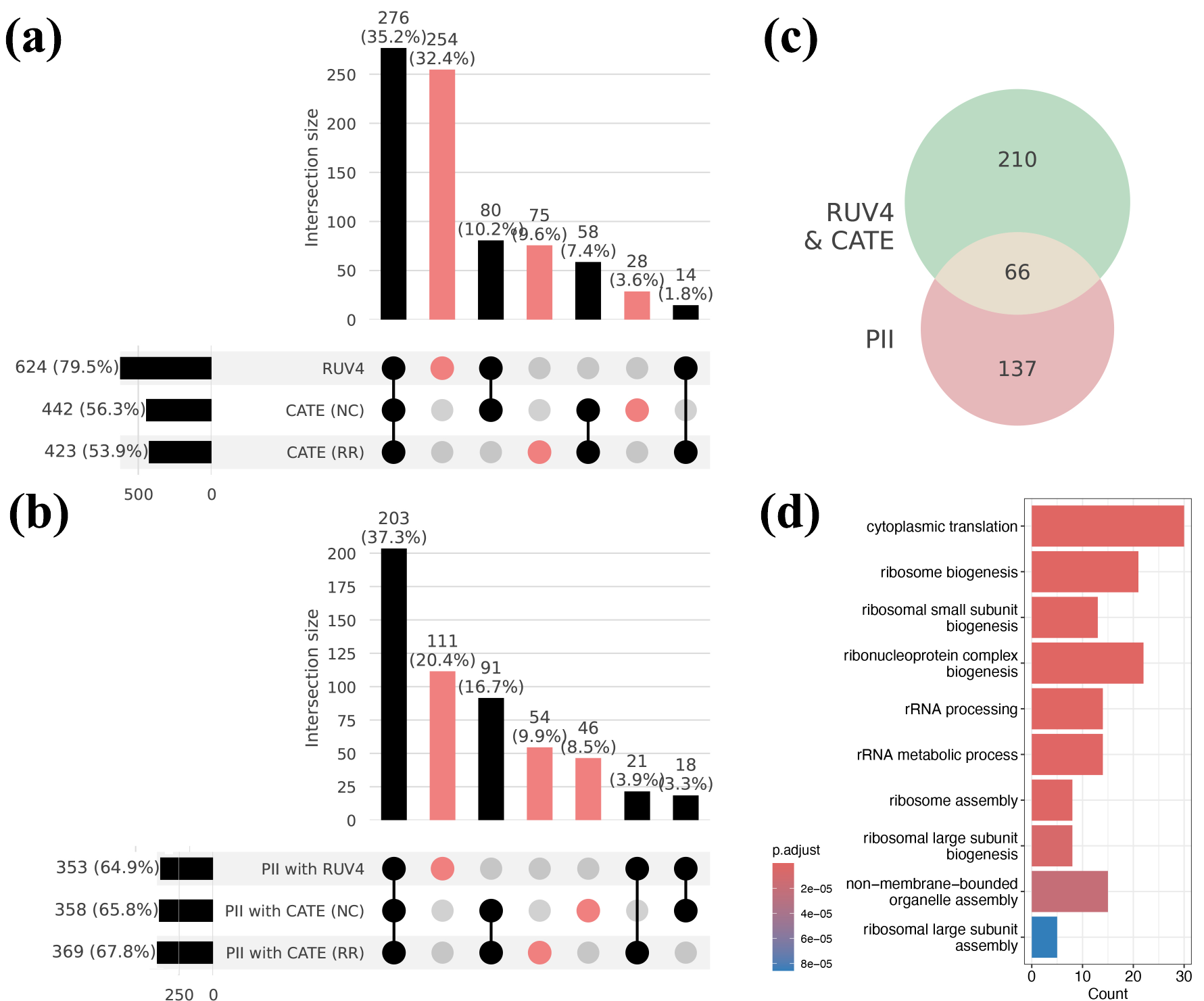}
        \caption{        
        Summary of significant genes (p-values$<0.05$) associated with \emph{PTEN} perturbation by different confounder adjustment methods.
        (a) Upset plot of discoveries by three methods: \RUV, \CATENC, and \CATERR, as in \Cref{fig:stat-all}.
        (b) Upset plot of discoveries by \PII with embedding estimated by three methods: \RUV, \CATENC, and \CATERR, as in \Cref{fig:stat-PII}.
        (c) The Venn plot of two sets of discoveries.
        One set includes 276 common discoveries by \RUV, \CATENC, and \CATERR, while the other includes 203 common discoveries by \PII with the same estimated embeddings by the three methods.
        (d) Gene ontology analysis of 137 distinct discoveries by \PII.
        }
        \label{fig:upset}
    \end{figure}

    The study by \citet{lalli2020high} suggests that specific perturbations influence changes in gene expression along pseudotime, potentially affecting the rate of development. 
    Using pseudotime as a covariate allows us to examine if perturbations explain effects beyond developmental changes; see \Cref{app:ex-real-data} for the extended background of the data. 
    Biologically, we expect more signals on pseudotime states (\Cref{fig:stat-pt}) than on perturbation conditions.

    We focus on the target gene \emph{PTEN}, which is crucial in neural development and differentiation and influences other genes in a cascading manner when it is knocked down \citep{lalli2020high}. 
    Examining the empirical distribution of test statistics for perturbation conditions using \GLM reveals conservative results for genes like \emph{CTNND2}, \emph{MECP2}, and \emph{MYT1L} (\Cref{fig:stat-all}). 
    This suggests that \GLM without adjusting for hidden confounders leads to biased hypothesis testing.    
    \PII corrects these biases (\Cref{fig:stat-other}, \Cref{fig:stat-PII}), and even PCA-based simple embedding estimation effectively calibrates test distributions. Comparing methods \RUV, \CATENC, \CATERR with \PII, we see \PII reduces distinct discoveries from 45.6\% to 38.8\% (\Cref{fig:upset}(a)), indicating more coherent outcomes with \PII.
     
    To further assess the biological significance of the discoveries, we examine 276 and 203 common discoveries (\Cref{fig:upset}(a) and \Cref{fig:upset}(b)). Significant discrepancies are noted (\Cref{fig:upset}(c)). 
    The associated gene ontology terms on biological processes using \texttt{clusterProfiler} package with default false discovery control threshold \citep{yu2012clusterprofiler} reveal that, unlike the 210 genes from \RUV, \CATENC, and \CATERR, which had no associated GO terms, the 137 genes unique to \PII align with ribosome-related processes (\Cref{fig:upset}(d)), supporting studies on \emph{PTEN}'s impact on these processes \citep{liang2017ptenbeta,cheung2023neuropathological}.

\section{Discussion}\label{sec:discussion}
    Extending beyond the CRISPR analysis demonstrated in the previous section, our method is broadly applicable to purely observational single-cell data. 
    For instance, the method is well-suited to assess the effects of a discrete exposure, like cell type or a developmental branching event, or a continuous exposure, such as the expression of a particular gene, while accounting for unwanted variations. Importantly, these estimates can be interpreted as causal effects when the necessary identifying assumptions are satisfied.

    A potential concern of the proposed method is whether the estimated embeddings might act as colliders, especially if $\hat{U}$ is influenced by both $X$ and $Y_{\cC^c}$. However, our fundamental assumption is that $Y_{\cC}$ is driven by a low-dimensional embedding $U$ but not the covariate $X$, which inherently mitigates the risk of $\hat{U}$ becoming a collider. If this foundational assumption does not hold, the direct effect estimand \eqref{eq:beta} might not align with researchers' interests, necessitating the use of domain knowledge to identify and investigate alternative target estimands.

    Our paper focuses on design-free deconfounding with surrogate control outcomes, though other strategies exist (detailed comparison in \Cref{app:comp}). While our framework allows for flexible machine learning algorithms, it introduces computational complexity, especially with increasing outcomes and hyperparameter tuning. For practical applications, specialized models like variational autoencoders for joint outcome function fitting \citep{du2022robust,moon2024augmented} and efficient cross-validation methods can be beneficial.
    
    Further extensions involve incorporating interaction effects \citep{vansteelandt2022assumption}, developing tests for nonparametric confounding \citep{miao2018identifying}. 
    Explorations into settings with high-dimensional latent embeddings and covariates \citep{miao2023identifying, zeng2024causal} could also be of interest. 
    In simulation and real data analysis, we use misspecified surrogate control outcomes, which can be viewed as one variant of the synthetic control approaches \citep{abadie2010synthetic}.
    Providing theoretical guarantees for the valid construction of surrogate control outcomes from data remains an area of practical interest.

\section*{Supplement}    
Related work, comparison with related deconfounding strategies, proofs of all results stated in the main text, experiment details, and extra results are given in the supplement.
The code for reproducing the results of this paper can be found at \url{https://github.com/jaydu1/PII}.

\bibliographystyle{biometrika}
\bibliography{ref}

\clearpage
\appendix

\counterwithin{theorem}{section}
\counterwithin{lemma}{section}
\counterwithin{remark}{section}

\renewcommand{\thesection}{\Alph{section}}
\renewcommand{\thetheorem}{\Alph{section}\arabic{theorem}}
\renewcommand{\thelemma}{\Alph{section}\arabic{lemma}}
\renewcommand{\theremark}{\Alph{section}\arabic{remark}}

\setcounter{table}{0}
\renewcommand{\thetable}{\thesection\arabic{table}}
\setcounter{figure}{0} 
\renewcommand\thefigure{\thesection\arabic{figure}}
\renewcommand{\thealgorithm}{\thesection\arabic{algorithm}}
\makeatletter
\renewcommand{\theequation}{S\@arabic\c@equation}
\makeatother

\begin{center}
\Large
{\bf Supplementary material for ``Assumption-Lean Post-Integrated\\ Inference with Surrogate Control Outcomes''}
\end{center}

The appendix includes related work, comparisons of deconfounding strategies, proofs of all theoretical results, and additional experimental results.
The outline and the summary of the notation are given below.

\paragraph{Outline} The structure of the appendix is listed below:
\begin{table}[!ht]\vspace{-2mm}
\centering \small
\begin{tabularx}{0.9\textwidth}{l l l }
    \toprule
    \multicolumn{2}{c}{\textbf{Appendix}} & \textbf{Content} \\
    \midrule \addlinespace[0.5ex] 
    \Cref{app:related-work} & \cellcolor{lightgray!25} & Related work. \\ \addlinespace[0.5ex] \cmidrule(l){1-3}\addlinespace[0.5ex] 
    \multirow{2}{*}{\Cref{app:comp}} 
    & \ref{app:comp-design-based} & Comparisons with design-based deconfounding strategies. \qquad\qquad\qquad\qquad\quad\\ 
     & \ref{app:comp-sparsity} & Comparisons with sparsity-based deconfounding strategies. \\ 
    \addlinespace[0.5ex] \cmidrule(l){1-3}\addlinespace[0.5ex]  
    \Cref{app:nonparam} & \cellcolor{lightgray!25} & Proof of \Cref{thm:iden-neg-outcomes}. \\ \addlinespace[0.5ex] \cmidrule(l){1-3}\addlinespace[0.5ex] 
    \multirow{3}{*}{\Cref{app:est}} & \ref{app:est:bias} & Proof of \Cref{thm:err-bound-beta-U}. \\
     & \ref{app:est:U_hat} & Proof of \Cref{lem:est-err-U_hat}. \\
     & \ref{app:est:aux} & Auxiliary lemmas: \Cref{lem:err-perturb}, \ref{lem:Lipschitz-Lq-norm}, and \ref{lem:err-bound-cond-exp}.\\\addlinespace[0.5ex] \cmidrule(l){1-3}\addlinespace[0.5ex] 
    \multirow{4}{*}{\Cref{app:semi}} &  \ref{app:semi:eff} & Proof of \Cref{thm:DR-linear} and \Cref{cor:inference}. \\    
    &  \ref{app:semi:simul-inference} & Proof of \Cref{prop:simul-inference}. \\
    &  \ref{app:semi:DR} & DR inference with nonlinear link function: \Cref{thm:DR}. \\
    & \ref{app:semi:aux} & Auxiliary lemmas: \Cref{lem:eff-inf}, \ref{lem:dr-est}, \ref{lem:lindeberg} and \ref{lem:chebyshev}.  \qquad \qquad \qquad \qquad \qquad  \\\addlinespace[0.5ex] \cmidrule(l){1-3}\addlinespace[0.5ex] 
    \multirow{2}{*}{\Cref{app:ex}} &  \ref{app:ex-simu} & Experiment details and extra results on simulation.\\
     &  \ref{app:ex-real-data} & Experiment details and extra results on real data.\\
    \addlinespace[0.5ex] \arrayrulecolor{black}
    \bottomrule
\end{tabularx}
\vspace*{-1.7\baselineskip}
\end{table}
\addcontentsline{toc}{part}{\appendixname}

\vfill
\paragraph{Notation}
Throughout our exposition, we will use the following notational conventions.
We use uppercase letters for random variables/vectors (e.g., $Y,X,U$) and lowercase letters for sample vectors, respectively (e.g., $y,x,u$).
For a matrix $\beta\in\RR^{d\times p}$, its $j$th column is denoted by $\beta_{\cdot j}$.
Sets are denoted by calligraphic uppercase letters ($\cA,\cC$).
Bold font is only used to denote design matrices and response matrices (e.g., $\bY,\bX,\bU$) whose first dimension equals the sample size.
For a set $\cA$, let $|\cA|$ be its cardinality.

For a random vector $X\in\RR^p$, $\cP_{X}$ denotes the projection in $L_2$.
For any matrix $A\in\RR^{n\times p}$ with full column rank, let $P_{A} = A(A^{\top}A)^{-1}A^{\top}$ and $P_{A}^{\perp} = I_p - P_{A}$ be the orthogonal projection matrices on the $A$'s column space and its orthogonal space, respectively.
For any square matrix $A\in\RR^{n\times n}$, $\lambda_i(A)$ denotes its $i$th eigenvalue.
The Gram matrix of $A^{\top}$ is denoted by $A^{\otimes2} = AA^{\top}$.
Matrix Hadamard product is denoted by $\odot$.
For two symmetric matrices $A,B\in\RR^{n\times n}$, we write that $A \preceq B$ ($A \succeq B$) if $B - A$ ($A- B$) is positive semi-definite.
    For $a\in\RR^m$, $\|a\|_q$ denotes the $\ell_q$-norm for $q=1,\ldots,\infty$.
    For $a\in\RR^m$, $A\in\RR^{m\times n}$, $\|a\|$ and $\|A\|$ denote the $\ell_2$-norm and operator norm, respectively.
    The condition number of $A$ is defined as $\kappa(A) = \|A\|\|A^{-1}\|$.
    For any random vector $X$, its $L_q$ norm is defined as $\|X \|_{\Lp{q}}=\EE[\|X\|_{q}^q]^{1/q}$ for $q=1,\ldots,\infty$.

    Let $\PP$, $\PP_n$, and $\hat{\PP}$ denote the joint probability measure of the data,  the empirical measure for a sample size $n$, and the ``plug-in'' estimated measure, where the unknown nuisance functions are replaced by their machine learning estimates. 
    For (potentially random) measurable functions $f$, we denote expectations with respect to $Z$ alone by $\PP f(Z) = \int f \rd\PP$ and with respect to both $Z$ and the observations where $f$ is fitted on by $\EE [f(Z)]$, and the empirical expectation is denoted by $\PP_n f(Z) = \frac{1}{n}\sum_{i=1}^nf(Z_i)$ for $Z_1,\ldots,Z_n\overset{i.i.d.}{\sim} \PP$.
    The notation $\sqrt{n} (\PP_n-\PP) f$ denotes the $\sqrt{n}$-scaled empirical process $n^{-1/2}\sum_{i=1}^n(f(O_i) - \EE[f(O_i)])$.
    Similarly, the population and empirical variances (or covariances) are denoted by $\VV$ and $\VV_n$, respectively.
    The identity map is denoted by $\II$.
    We write the (conditional) $L_p$ norm of $f$ as $\|f\|_{\Lp{p}} = \left[\int f(z)^p \rd \PP(z) \right]^{1/p}$ for $p\geq1$.
    
    We use ``$o$'' and ``$\cO$'' to denote the little-o and big-O notations and let ``$\op$'' and ``$\Op$'' be their probabilistic counterparts over all randomnesses.
    For sequences $\{a_n\}$ and $\{b_n\}$, we write $a_n\ll b_n$ or $b_n\gg a_n$ if $a_n=o(b_n)$; $a_n\lesssim b_n$ or $b_n\gtrsim a_n$ if $a_n=\cO(b_n)$; and $a_n\asymp b_n$ if $a_n=\cO(b_n)$ and $b_n=\cO(a_n)$.
    For $a,b\in\RR$, we write $a\vee b = \max\{a,b\}$ and $a\wedge b = \min\{a,b\}$.

\clearpage

\section{Related work}\label{app:related-work}

        \paragraph{Batch correction and data integration.}
        Large-scale single-cell transcriptomic datasets often include samples that span locations, laboratories, and conditions, leading to complex, nested batch effects in data \citep{tran2020benchmark,luecken2022benchmarking}.
        Batch correction specifically targets the removal of unwanted variation due to differences in batches within a single study, ensuring that the remaining data is comparable and reflects true biological differences.
        On the other hand, data integration focuses on combining and harmonizing multiple datasets to enhance statistical power and provide a more comprehensive analysis, dealing with both batch effects and between-dataset heterogeneity.
        Despite these differences, batch correction and data integration share the common goal of removing unwanted variation and preserving biological variation \citep{zhang2023signal}.
        The integrated data is then used for downstream analysis, such as dimension reduction, clustering, and differentially expressed gene testing.
        Integrated cellular profiles are typically represented as an integrated graph, a joint embedding, or a corrected response matrix.
        The main focus of the current paper is on the last category.

        Despite the efforts from the computational biology and machine learning community to achieve better predictive power and data alignment, most existing batch correction methods are shown to be poorly calibrated \citep{antonsson2024batch,ma2024principled}.
        For statistical inference, many heuristic methods have been proposed to remove the batch effects and unwanted variations in the past decade.
        Leading examples include Remove Unwanted Variation (RUV) \citep{gagnon2012using} and Surrogate Variable Analysis (SVA) \citep{leek2012sva}.
         RUV/SVA uses estimated latent factors from unadjusted data, which works even if the batch design is unknown.
        When the batch design is known, two-step procedures have also been proposed under parametric or mixture models \citep{li2023overcoming,luo2018batch}.

        \paragraph{Unmeasured confounders adjustment and negative control outcomes.}
        Over the past decades, researchers have been exploring methods to address the issue of unmeasured confounders in statistical analysis. In the presence of multiple outcomes, deconfounding techniques primarily employ two strategies: incorporating known negative control outcomes or leveraging sparsity assumptions \citep{wang2017confounder}; while there is also another line of research on proximal causal inference, which uses both negative control outcomes and/or exposures for deconfounding \citep{miao2018identifying}.
         For a comprehensive review of the literature on sparsity-based methods, readers are directed to \citet{du2023simultaneous} and \citet{zhou2024promises}.
        This paper focuses on the negative control approach in the context of multiple outcomes.

        Most existing works on confounder adjustment presume the knowledge of causal structure when the unobserved variable $U$ is a mediator \citep{wang2017confounder} and when $U$ is a confounder \citet{miao2023identifying}, corresponding to \Cref{fig:batch-correction}(a) and \Cref{fig:batch-correction}(b), respectively. 
        Recently developed sparsity-based methods by \citet{bing2023inference,du2023simultaneous} have tried to relax this assumption to allow for a more flexible relationship between $X$ and $U$.
        In particular, each entry of $U$ can belong to different cases in \Cref{fig:batch-correction}.

        Negative control outcomes are used in observational studies under the key assumption that exposure has no causal effect on these outcomes. \citet{rosenbaum1989role} demonstrated that negative control outcomes can be employed to test for the presence of hidden confounding in observational studies.
        By introducing an additional variable known as a negative control exposure, \citet{miao2018identifying} further showed that the average causal effect can be identified nonparametrically. 
        Building upon this work, \citet{shi2020multiply} developed a semiparametric inference procedure specifically for scenarios involving a categorical latent confounder and a binary exposure.
        Under linear latent models, \citet{galbraith2020simple} use principal components of a set of potential controls to adjust for unmeasured confounding effects.
        Under nonparametric models for a single outcome and multiple treatments \citet{miao2023identifying} derive nonparametric identification conditions.

        \paragraph{Assumption-lean semiparametric inference.}
        There is increasing interest in deriving assumption-lean inference with projection-based estimators \citep{berk2021assumption} or semiparametric estimators \citep{vansteelandt2022assumption}.
        The inferential problems we considered are also related to two-stage inference problems, such as post-sufficient dimension reduction inference \citep{kim2020post}, post-imputation inference \citep{moon2024augmented}, and inference with substituted covariate \citep{adams2024substitute} or nonparametrically generated covariates \citep{mammen2012nonparametric}.
        While these related methods offer valuable insights into two-stage procedures, they do not directly extend to address post-integrated inference problems.

\clearpage
\section{Comparisons with related deconfounding approaches}\label{app:comp}

\subsection{Design-based approaches}\label{app:comp-design-based}

    As mentioned in the introduction, our paper mainly focuses on design-free data integration approaches.
    However, it is possible to relate the design-based approaches to design-free approaches so that the proposed method can be applied, as we discuss below. 
    Design-based data integration approaches, such as Combat \citep{johnson2007adjusting} and BUS \citep{luo2018batch}, are usually based on a linear model:
    \begin{align*}
        Y_{j} &= \alpha_j+ X \beta_{j} + \gamma_{Bj}+ \epsilon_{Bj} \quad (j=1,\ldots,p),
    \end{align*}
    where $\alpha_j,\beta_j\in\RR$ are coefficients for common variations while $\gamma_{Bj}\in\RR$ is the location and $\epsilon_{Bj}\in\RR$ is a mean-zero noise with scale differences across batches, respectively, for batch $B\in[n_B-1]$ (with group $B=0$ being the baseline and $\gamma_{0j}=0$) and $n_B$ is the total number of batches.
    This implies that
    \begin{align*}
        \EE[Y_j \mid X, B] &= \alpha_j+ X \beta_j + \gamma_{Bj}\quad (j=1,\ldots,p).
    \end{align*}
    Let $U_B\in\{0,1\}^{n_B}$ be the one-hot vector with only the $B$-th entry being one and zero elsewhere.
    Then, we can rewrite the above as
    \begin{align*}
        \EE[Y \mid X, U_B] &= \alpha + \beta^{\top} X  + f(U_B)
    \end{align*}
    where $f(U_B) = \gamma^{\top} U_B $ and $\gamma = [\gamma_{bj}]_{b\in[n_B],j\in[p]}$.
    In other words, the location-and-scale model considered by \citet{johnson2007adjusting} and \citet{luo2018batch} is a special case of partial linear models with heterogeneous noises, though they have utilized empirical Bayes shrinkage to improve the estimates.
    For this reason, a generalized least square approach could be used to improve Combat, as suggested by \citet{li2023overcoming}.

    In fact, when the additive noises are normal, we can decompose the noise as $\epsilon_{Bj} = U_{\epsilon} + Z_{j} $ for $B>0$ such that $U_{\epsilon} \indep \epsilon_{j}'$ and $\epsilon_{j}'\overset{\mathrm{d}}{=} \epsilon_{0j}$.
    To see this, define $\tau^2=\VV(\epsilon_{Bj})$ and $\sigma^2 = \VV(\epsilon_{0j})$.
    Without loss of generality, we assume $\epsilon_{0j}\leq\min_{b\in[n_B]}\epsilon_{bj}$ so that $\tau^2\geq \sigma^2$.
    If we define $U_{\epsilon} = \frac{1}{\tau(\tau^2-\sigma^2)} \epsilon_{Bj} + Z_{j}$ and $\epsilon_{j} '=\epsilon_{Bj} - U_{\epsilon}$, where $Z_j \sim \cN(0,\frac{\sigma^2}{(\tau^2-\sigma^2)^2})$ is independent of $\epsilon_{Bj}$, then $U_{\epsilon}$ and $\epsilon_{Bj}-U_{\epsilon}$ are independent because $\Cov(U_{\epsilon},\epsilon_{Bj}-U_{\epsilon}) = ({\tau^2 (\tau^2 - \sigma^2)^2})^{-1} \VV(\epsilon_{Bj}) - \VV(Z_j) = 0$.
    Here, we use the fact that two jointly normal random variables are independent if they are uncorrelated.
    In other words, we can rewrite the above model as
    \begin{align*}
        \EE[Y \mid X, U] &= \alpha + \beta^{\top} X  + f(U),
    \end{align*}
    where $f(U) = [\gamma,\ind\{B>0\}]^{\top} U $ and $U = [U_B,U_{\epsilon}]$.
    By absorbing part of the randomness of the additive noises into $U$, we convert the problem with heterogeneous noises into one with homogeneous noises studied in the current paper.

    \subsection{Unknown negative control outcomes}\label{app:comp-sparsity}

    In this paper, we have focused on negative control outcomes to remove unwanted variations.
    When the negative control outcomes are unknown in advance, there are still possibilities to estimate the latent embedding and provide valid inferences.
    However, this typically requires extra sparsity assumptions on the effects of the covariate on multiple outcomes rather than utilizing the negative control outcomes.
    To illustrate the idea, we consider the following partial linear model: 
    \[\EE[Y\mid X,U] = \beta^{\top}X + h(U).\]
    Many methods start from the projected model
    \begin{align}
        \EE[\cP_{X}^{\perp}Y\mid X,U] &= \cP_{X}^{\perp} h(U) . \label{eq:nonlinear-eq-1} 
    \end{align}
    If the function $\cP_X^{\perp}h$ has a good structure, then one may be able to recover $U$ from $\cP_{X}^{\perp}Y$.
    Alternatively, we can linearize the problem and seek partial recovery of the effect, as demonstrated below.

    \begin{example}[Linear models]\label{ex:linear}
        If $h$ is a linear function such that $h:U\mapsto \eta^{\top} U$ for $\eta\in\RR^{r\times p}$, then
        \begin{align*}
            Y &= [\beta \,\, \eta]^{\top} \begin{bmatrix} X \\ U \end{bmatrix} + Z.
        \end{align*}  
        With $n$ i.i.d. samples, we obtain the following equation in matrix form:
        \begin{align}
            \bY & = \bX \beta + \bU \eta + \bZ. \label{eq:sample-linear-gau}
        \end{align}
        Early methods in the literature rely on the assumption of the functional relationship between $X$ and $U$.
        For example, \citet{wang2017confounder} assume $U$ to be a linear function of $X$ with an additive Gaussian noise:
        \begin{align}
            U = X \alpha + W,   \label{eq:U-X}
        \end{align}
        in which case the sample counterpart of \eqref{eq:nonlinear-eq-1} reduces to 
        \begin{align}
            P_{\bX}^\perp \bY =  P_{\bX}^\perp (\bW\eta + \bZ).\label{eq:linear-gau-eq-1}
        \end{align}
        Because the orthogonal projection is rank-deficient, one can further eliminate $d$ rows of the above system of equations by elementary matrix transformation.
        For this purpose, \citet{wang2017confounder} use QR decomposition by Householder rotation to derive a linear system of $n-d$ equations; e.g., Equation (2.5) and Equation (4.5) in \citet{wang2017confounder} for $d=1$ and $d>1$, respectively.
        From this, $\hat{\eta}$ is recovered from quasi-log-likelihood estimation.
        In the second step, the unknown coefficient $(\alpha,\beta)$ is estimated from \eqref{eq:linear-gau-eq-1} by plugging in the estimate $\hat{\eta}$.

        Under more general confounding machinism when \eqref{eq:U-X} does not necessarily hold, \citet{bing2022adaptive,bing2023inference} rotate the original system to consistently estimate the marginal effect, under sparsity assumption on $\beta$ and proper moment assumptions.
        They then use the residual from the lava fit to uncover the column space of $\eta$.
        Finally, the partial coefficient $\beta P_{\eta}^{\perp} $ is recovered from the rotated system:
        \[\bY P_{\eta}^{\perp} = \bX \beta P_{\eta}^{\perp} + \bZ P_{\eta}^{\perp}. \]
        These results have been extended to generalized linear models by \citet{du2023simultaneous} using joint maximum likelihood estimation.
        When $\beta$ is sparse, then it can be recovered by some estimator of $\beta P_{\eta}^{\perp} $ asymptotically.
        Note that the above approaches do not have too many restrictions on the observed covariate $X$ and the latent embedding $U$, except for certain bounded moment assumptions.
    \end{example}

    Inspired by the success of methodologies development under linear models \Cref{ex:linear}, one strategy for an extension to a nonlinear model is by linearizing the estimation problem.
    Specifically, suppose $h(U) = U\eta + R(U)$ for some remainder term $R$ that depends on $U$, similarly we have a projection-based decomposition:
    \begin{align*}
        P_{\bX}^\perp \bY &= P_{\bX}^{\perp}\bU\eta + P_{\bX}^\perp (R(\bU) + \bZ) \\
        \bY P_\eta^{\perp} &= \bX \beta P_\eta^{\perp} + (R(\bU) + \bZ) P_\eta^{\perp} ,
    \end{align*}
    from which one may seamlessly use the methods by \citet{bing2022adaptive,bing2023inference} and \citet{du2023simultaneous} when the remainder term can be well controlled.

    Another possible strategy aligned with the angle of the current paper is to detect ``weak'' negative control outcomes and perform post-integrated inference based on such negative control outcomes, as in \Cref{sec:real-data}.
    This approach is very similar to weak instrument detection and invalid instrumental variables selection; see, for example, \citet{andrews2019weak} and \citet{windmeijer2021confidence}.
    We expect the rich literature on these related problems to lead to new methodological advances in post-integrated inference problems.

\clearpage

\section{Nonparametric identification}\label{app:nonparam}

\begin{proof}[of \Cref{thm:iden-neg-outcomes}]
        Under the equivalence assumption (\Cref{asm:neg-outcomes}~\hyperlink{asm:equiv}{(ii)}), for any admissible distribution $\tf (y_{\cC}, u)$ we must have some invertible function $v$ such that $\tf(y_{\cC},u) = f \{ Y_{\cC}=y_{\cC},v(U) = u\}$.
        Note that \eqref{eq:fyx_proxy} has at least one solution $\tf(x\mid u)=f(x\mid v^{-1}(u))$;
        when this is the solution, define $\tf(y_{\cC^c} \mid x,u) = f(y_{\cC^c} \mid x,v^{-1}(u))$.
        Then, $\tf(y_{\cC^c} \mid x,u)$ is also one solution to \eqref{eq:fy_x_proxy}.

        Because $v (U)$ is invertible, the ignorability assumption (\Cref{asm:causal}~\hyperlink{asm:nuc}{(iii)}) $Y (x)\indep X \mid U$ implies that $Y (x) \indep X \mid v (U)$;
        the completeness assumption \Cref{asm:neg-outcomes}~\hyperlink{asm:comp}{(iii)} implies that $\tf(u)>0$ on $u\in v(\cU)$ and $\tf (u \mid y_{\cC} ,x; \alpha)$ is also complete in $y_{\cC}$.
        Further, from \Cref{asm:causal}~\hyperlink{asm:positivity}{(ii)}, the positivity condition $f_{X\mid v(U)}(x\mid v(u))\in(0,1)$ also holds for all $(x,u)\in\cX\times \cU$.
        Then, we have
        \begin{align}
            f_{Y(x)}(y) &= \int f_{Y(x)\mid U}(y \mid u) f(u) \rd u \notag\\
            & = \int f_{Y(x)\mid U,X}(y \mid u, x) f(u) \rd u \tag{\Cref{asm:causal}~\hyperlink{asm:positivity}{(ii)}-\hyperlink{asm:nuc}{(iii)}}\\
            & = \int f(y \mid u, x) f(u) \rd u \tag{\Cref{asm:causal}~\hyperlink{asm:consistency}{(i)}}\\
            & = \int f(y_{\cC^c} \mid u, x) f(y_{\cC} \mid u) f(u) \rd u \tag{\Cref{asm:neg-outcomes}~\hyperlink{asm:indep}{(i)}}\\
            & = \int f(y_{\cC^c} \mid u, x) f(y_{\cC},u) \rd u \notag \\
            &=\int \tf(y_{\cC^c} \mid x,u)\tf(y_{\cC}, u ) \rd u , \notag
        \end{align}
        where the last equality follows from the same derivation of the g-formula applied to random variables $(Y,X,v(U))$.
        This completes the proof for the second conclusion.

        We next show the uniqueness of the solutions to \eqref{eq:fyx_proxy} and \eqref{eq:fy_x_proxy}.
        For any candidate solutions $\tf_1 (x\mid u)$ and $\tf_2 (x\mid u)$ to \eqref{eq:fyx_proxy}, we must have that 
        \[\int (\tf_1 (x\mid u) - \tf_2 (x\mid u)) \tf(u) \rd u = 0,\]
        which implies that $\tf_1 (x\mid U) - \tf_2 (x\mid U) = 0$ almost surely because of the completeness of $\tf(u)$.
        Note that for any candidate solutions $\tf_1 (y_{\cC^c} \mid x, u)$ and $\tf_2 (y_{\cC^c} \mid x, u)$ to \eqref{eq:fy_x_proxy}, we must have that 
        \[
            \int(\tf_1 (y_{\cC^c} \mid x, u)-\tf_2 (y_{\cC^c} \mid x, u)) \tf (u \mid y_{\cC},x)\rd u \cdot f(y_{\cC},x) = 0.
        \]
        By the completeness property, this implies that $\tf_1 (y_{\cC^c} \mid x, U)-\tf_2 (y_{\cC^c} \mid x, U) =0$ almost surely. 
        Therefore, $\tf (y_{\cC^c} \mid x, u)$ is uniquely determined from \eqref{eq:fy_x_proxy}.
        This completes the proof.
    \end{proof}

\clearpage
\section{Nonlinear main effects with estimated embeddings}\label{app:est}

\subsection{Proof of \Cref{thm:err-bound-beta-U}}\label{app:est:bias}
\begin{proof}[of \Cref{thm:err-bound-beta-U}]
        Denote $A = \EE[\Cov(X\mid U)]$, $ \hat{A} = \EE[\Cov(X\mid \hat{U})]$, $B = \EE[\Cov(X,\EE[Y\mid X,U]\mid U)]$, and $\hat{B} = \EE[\Cov(X,\EE[Y\mid X,U]\mid \hat{U})]$.
        From \Cref{lem:err-perturb}, we know that the error of two linear regression coefficients $\|\tbeta_{\cdot j} - \beta_{\cdot j}\|$ is governed by $\|A-\hat{A}\|$ and $\|B_{\cdot j}-\hat{B}_{\cdot j}\|$, where the subscript $j$ indicates the $j$th column of the corresponding matrices.

        \paragraph{Part (1) Covariance estimation errors.}
        To apply \Cref{lem:err-perturb}, we first derive the error bounds for the two quantities.
        Note that $\Cov(X\mid U) = \EE[X^{\otimes 2} \mid U] - \EE[X \mid U]^{\otimes 2}$. We have
        \begin{align*}
            &\|A-\hat{A}\|\\
            =&\|\EE[\Cov(X \mid U) - \Cov(X \mid \hat{U})] \| \\
            =& \|\EE[\EE[X \mid U]^{\otimes 2} - \EE[X \mid \hat{U}]^{\otimes 2}]\| \\
            \leq& \EE[ \|\EE[X \mid U]^{\otimes 2} - \EE[X \mid \hat{U}]^{\otimes 2}\|] \tag{Jensen's inequality} \\
            \leq& \EE[ \|\EE[X \mid U](\EE[X \mid U]-\EE[X \mid \hat{U}])^{\top}\| + \|(\EE[X \mid U] - \EE[X \mid \hat{U}])\EE[X \mid \hat{U}]^{\top}\|] \tag{triangle inequality}\\
            =& \EE[ \| \EE[X \mid U] \| \|\EE[X \mid U]-\EE[X \mid \hat{U}]\| + \|\EE[X \mid U] - \EE[X \mid \hat{U}]\| \|\EE[X \mid \hat{U}] \|] \\
            \leq & (\| \EE[X \mid U] \|_{\Lp{2}} + \| \EE[X \mid \hat{U}] \|_{\Lp{2}} )\| \EE[X \mid \hat{U}] - \EE[X \mid U]  \|_{\Lp{2}} \tag{Cauchy–Schwarz inequality}\\
            \leq & 2\|X\|_{\Lp{2}}\| \EE[X \mid \hat{U}] - \EE[X \mid U]  \|_{\Lp{2}}. \tag{Jensen's inequality}
        \end{align*}
        Similarly, the second covariance estimation error can be upper bounded as
        \begin{align*}
            &\|B_{\cdot j}-\hat{B}_{\cdot j}\|\\
            =
            &\|\EE[\Cov(X,\EE[Y_j\mid X,U] \mid U) - \Cov(X, \EE[Y_j\mid X,\hat{U}] \mid \hat{U})]\| \\
            =&\|\EE[\EE[X \mid U]\EE[Y_j \mid X,U] - \EE[X \mid \hat{U}]\EE[Y_j \mid X,\hat{U}]]\| \\
            \leq& \EE[ \|\EE[X \mid U]\EE[Y_j \mid X,U] - \EE[X \mid \hat{U}]\EE[Y_j \mid X,\hat{U}]\|] \tag{Jensen's inequality} \\
            \leq& \EE[ \|\EE[X \mid U](\EE[Y_j \mid X,U]-\EE[Y_j \mid X,\hat{U}])\| + \|(\EE[X \mid U] - \EE[X \mid \hat{U}])\EE[Y_j \mid X,\hat{U}]\|] \tag{triangle inequality}\\
            =& \EE[ \| \EE[X \mid U] \| |\EE[Y_j \mid X,U]-\EE[Y_j \mid X,\hat{U}] | + \|\EE[X \mid U] - \EE[X \mid \hat{U}]\| |\EE[Y_j \mid X,\hat{U}] |] \\
            \leq & \| \EE[X \mid U] \|_{\Lp{2}} \| \EE[Y_j \mid X,\hat{U}] - \EE[Y_j \mid X,U]  \|_{\Lp{2}} + \| \EE[Y_j \mid X,\hat{U}] \|_{\Lp{2}} \| \EE[X \mid \hat{U}] - \EE[X \mid U]  \|_{\Lp{2}}  \tag{Cauchy–Schwarz inequality}\\
            \leq & \| X \|_{\Lp{2}} \| \EE[Y_j \mid X,\hat{U}] - \EE[Y_j \mid X,U]  \|_{\Lp{2}} + \| Y_j \|_{\Lp{2}} \| \EE[X \mid \hat{U}] - \EE[X \mid U]  \|_{\Lp{2}}. \tag{Jensen's inequality}
        \end{align*}

        \paragraph{Part (2) Coefficient estimation error in terms of covariance estimation errors.}
        When $\|\EE[X\mid \hat{U}] - \EE[X\mid U]\|_{\Lp{2}} < \sigma/(2M)$, from Part (1) we have $\kappa(A)\|A-\hat{A}\|/\|A\|=\|A^{-1}\|\|A-\hat{A}\|< 1$.
        From \Cref{lem:err-perturb}, we further have that
        \begin{align*}
            \max_{j\in\cC^{c}} \|\tbeta_{\cdot j} - \beta_{\cdot j}\| &\leq \frac
            {\max_{j\in\cC^{c}}\kappa(A)\left(  \dfrac{\|\beta_{\cdot j}\| \| A-\hat{A}\|}{\| A\|}+\dfrac
            {\| B_j-\hat{B}_j\|}{\| A \|}\right)  }
            {1-\kappa(A)\dfrac{\| A-\hat{A}\|}{\| A\|}} \\
            &\lesssim \| X \|_{\Lp{2}} \max_{j\in\cC^{c}} \|\beta_{\cdot j}\| \| \EE[Y_j \mid X,\hat{U}] - \EE[Y_j \mid X,U]  \|_{\Lp{2}} \\
            &\qquad + 
            (\| X \|_{\Lp{2}}  + \max_{j\in\cC^{c}} \| Y_j \|_{\Lp{2}}) \| \EE[X \mid \hat{U}] - \EE[X \mid U]  \|_{\Lp{2}} ,
        \end{align*}
        where in the last inequality, we use the boundedness of $A$'s spectrum from \Cref{asm:moment}.

        \paragraph{Part (3) Coefficient estimation error in terms of covariate estimation errors.}
        From \Cref{lem:err-bound-cond-exp}, it follows that
        \begin{align*}
            \max_{j\in\cC^{c}}\|\tbeta_{\cdot j} - \beta_{\cdot j}\|&\lesssim 
            \left(\| X \|_{\Lp{2}}(L_X^{\frac{1}{2}}+L_{Y}^{\frac{1}{2}})  + \max_{j\in\cC^c}\| Y_j \|_{\Lp{2}} \right) \| v^{-1}(\hat{U})-U)\|_{\Lp{2}} ,
        \end{align*}
        with $L_Y = \max_{j\in[p]}L_{Y_j}$.
    \end{proof}

\subsection{Proof of \Cref{lem:est-err-U_hat} (linear models)}\label{app:est:U_hat}
\begin{proof}[of \Cref{lem:est-err-U_hat}] 
    For observations $(\bX,\bY,\bU)\in\RR^{n\times d}\times \RR^{n\times p} \times \RR^{n\times r}$ and an estimate $\hat{\bU}\in\RR^{n\times \hat{r}}$ of $\bU$, we have $S = \bX^{\top}P_{\bU}^{\perp}\bX$, $\tilde{S} = \bX^{\top}P_{\hat{\bU}}^{\perp}\bX$, $\bar{\bY}=P_{\bU}^{\perp}\bY$, and $\tilde{\bY}=P_{\hat{\bU}}^{\perp}\bY$.
    Furthermore, the regression coefficient on $(P_{\bU}^{\perp}\bX, P_{\bU}^{\perp}\bY)$ can be expressed as
    \begin{align*}
        b &= (\bX^{\top} P_{\bU}^{\perp}\bX )^{-1}\bX^{\top}P_{\bU}^{\perp}\bY = S^{-1}\bX^{\top} \bar{\bY}/n,
    \end{align*}
    and the regression coefficient on $(P_{\hat{\bU}}^{\perp}\bX, P_{\hat{\bU}}^{\perp}\bY)$ can be expressed as
    \begin{align*}
        \tilde{b} &=(\bX^{\top} P_{\hat{\bU}}^{\perp}\bX )^{-1}\bX^{\top}P_{\hat{\bU}}^{\perp}\bY = \tilde{S}^{-1}\bX^{\top} \tilde{\bY} /n.
    \end{align*}

    From \Cref{lem:err-perturb}, we have
    \begin{align}
        \max_{j\in\cC^{c}}\|\tilde{b}_{\cdot j} - b_{\cdot j}\| &\leq 
         \frac{\kappa(S)}{\|S\|_{\oper}}\frac{\max_{j\in[p]} \|b_{\cdot j}\|\| \tilde{S}-S \| +  \| \bX^{\top}(\tilde{\bY}_{\cdot j} - \bar{\bY}_{\cdot j} ) /n\|}{ 1-  \kappa(S) \dfrac{\| \tilde{S}-S \|}{\|S\|_{\oper}}}. \label{eq:err-beta}
    \end{align}
    This requires verifying the assumptions therein.
    Specifically, we verify \eqref{eq:perturb1} below.
    Because
    \begin{align*}
        \| \tilde{S} - S \|&= \|\bX^{\top}(P_{\hat{\bU}}^{\perp}-P_{\bU}^{\perp})\bX/n\| \\
        &\leq \|\bX^{\top}\bX/n\| \|P_{\hat{\bU}}^{\perp}-P_{\bU}^{\perp}\|\\
        & = \|S\|_{\oper}\|P_{\hat{\bU}}^{\perp}-P_{\bU}^{\perp}\|,
    \end{align*}
    and $\kappa(S)\|P_{\hat{\bU}}^{\perp}-P_{\bU}^{\perp}\|<1$ as assumed, we have
    \[ \frac{\| \tilde{S} - S \|}{\|S\|_{\oper}} < \frac{1}{\kappa(S)},\]
    which verifies \eqref{eq:perturb1} of \Cref{lem:err-perturb}.
    On the other hand, we also have
    \[\| \bX^{\top}(\tilde{\bY} - \bar{\bY}) \|_{2,\infty} = \| \bX^{\top}(P_{\hat{\bU}}^{\perp}-P_{\bU}^{\perp}) \bY \|_{2,\infty} \leq \|P_{\hat{\bU}}^{\perp}-P_{\bU}^{\perp}\| \|\bX\|_{\oper}\|\bY\|_{2,\infty}.\] 
    Therefore, \eqref{eq:err-beta} implies that
    \begin{align*}
        \max_{j\in\cC^{c}}\|\tilde{b}_{\cdot j} - b_{\cdot j}\| &\leq \frac{\|S\|_{\oper}\|b\|_{2,\infty} + \|\bX\|_{\oper}\|\bY\|_{2,\infty}/n }{\|S\|_{\oper}}\frac{\kappa(S) \|P_{\hat{\bU}}^{\perp}-P_{\bU}^{\perp}\|}{1 - \kappa(S)  \|P_{\hat{\bU}}^{\perp}-P_{\bU}^{\perp}\|}   \\
        &\leq (\|b\|_{2,\infty} + \|S\|_{\oper}^{-\frac{1}{2}}\|\bY\|_{2,\infty}n^{-\frac{1}{2}})\frac{\kappa(S) \|P_{\hat{\bU}}^{\perp}-P_{\bU}^{\perp}\|}{1 - \kappa(S)  \|P_{\hat{\bU}}^{\perp}-P_{\bU}^{\perp}\|} 
    \end{align*}
    whenever $\|S\|_{\oper} \neq0$ and $\kappa(S)\|P_{\hat{\bU}}^{\perp}-P_{\bU}^{\perp}\|<1$.
\end{proof}

\subsection{Auxillary lemmas}\label{app:est:aux}

\begin{lemma}[Backward error of perturbed linear systems]\label{lem:err-perturb}
    Let $A\in\RR^{n\times n}$ be nonsingular, $b\in\RR^{n}$, and $x=A^{-1}b
    \in\RR^{n}$. In the following, $\Delta A\in\RR^{n\times n}$ and
    $\Delta b\in\RR^{n}$ are some arbitrary matrix and vector. We
    assume that the norm on $A$ satisfies $\| Ax\|\leq\| A\|\|x\|$ for all $A\in\RR^{n\times n}$ and all
    $x\in\RR^{n}$.
    Suppose $(A+\Delta A)\hat{x}=\hat{b}$ such that
    \begin{align}
        \hat{b} &=b+\Delta b  \neq\zero\\
        \hat{x}&=x+\Delta x  \neq\zero\\
        \frac{\|\Delta A\|}{\| A\|} & <\frac{1}{\kappa(A)}, \label{eq:perturb1}
    \end{align}
    where $\kappa(A) = \|A\|\|A^{-1}\|$ is the condition number of $A$.
    Then, it holds that
    \[
        \|\Delta x\|  \leq  \frac{\|x\| \kappa(A) \frac{\|\Delta A\|}{\|A\|} +  \kappa(A)\frac{\|\Delta b\|}{\|A\| }}{ 1-  \kappa(A) \frac{\|\Delta A\|}{\|A\|}}.
    \]
    If further, $b\neq\zero$ (or equiavlently $x\neq \zero$), then
    \[
        \frac{\|\Delta x\|}{\|x\|}\leq\frac
        {\kappa(A)\left(  \dfrac{\|\Delta A\|}{\| A\|}+\dfrac
        {\|\Delta b\|}{\|b\|}\right)  }
        {1-\kappa(A)\dfrac{\|\Delta A\|}{\| A\|}}.
    \]
\end{lemma}
\begin{proof}[of \Cref{lem:err-perturb}]
    We split the proof into two parts.

    \paragraph{Part (1)}
    We first show that when \eqref{eq:perturb1} is satisfied, $A+\Delta A$ must be nonsingular.
    If $A+ \Delta A$ is singular, then exists nonzero $v$ such that $(A+ \Delta A)v=\zero$. Since $A$ is nonsingular, we have $A^{-1}\Delta A v = - v$. So
    \begin{align*}
        \|v\|&=\|A^{-1}\Delta A v\|
        \leq \|A^{-1}\|\|\Delta A\| \|v\|,
    \end{align*}
    which implies that
    \[\|\Delta A\|\geq \frac{1}{\|A^{-1}\|}.\]
    On the other hand, since $\kappa(A)=\|A\|\|A^{-1}\|$, from \eqref{eq:perturb1} we have
    \[\frac{\|\Delta A\|}{\| A\|}<\frac{1}{\|A\|\|A^{-1}\|},\]
    or equivalently, 
    \[\|\Delta A\|<\frac{1}{\|A^{-1}\|}.\]
    This leads to contradictions. Therefore, $A+\Delta A$ must be nonsingular.
    
    \paragraph{Part (2)} Since $(A+\Delta A)\hat{x}=b+\Delta b$ and $A x= b$, we have $A\Delta x+\Delta A\hat{x}=\Delta b$. So $\Delta x =A^{-1}(\Delta b- \Delta A\hat{x})$. Then we have
    \begin{align*}
        \frac{\|\Delta x\|}{\|\hat{x}\|}&=\frac{\| A^{-1}(\Delta b- \Delta A\hat{x})\|}{\|\hat{x}\|}\\
        &\leq \frac{\| A^{-1}\| (\|\Delta A\|\|\hat{x}\|+\|\Delta b\|)}{\|\hat{x}\|}\\
        &= \|A^{-1}\|\|A\| \left(\frac{\|\Delta A\|}{\|A\|} + \frac{\|\Delta b\|}{\|A\|\|\hat{x}\|}  \right) \\
        &=\kappa(A)\left(  \frac{\|\Delta A\|}{\| A\|}+\frac
        {\|\Delta b\|}{\| A\|\|\hat{x}\|
        }\right),
    \end{align*}
    and
    \begin{align*}
        \|\Delta x\|&\leq  \kappa(A) \left( \frac{\|\Delta A\|}{\|A\|}+ \frac{\|\Delta b\|}{\|A\|\|\hat{x}\|} \right) \|\hat{x}\|\\
        &=\kappa(A) \left( \frac{\|\Delta A\|}{\|A\|}\|\hat{x}\|+ \frac{\|\Delta b\|}{\|A\|} \right) \\
        &\leq\kappa(A) \frac{\|\Delta A\|}{\|A\|}(\|x\|+\|\Delta x\|)+ \kappa(A)\frac{\|\Delta b\|}{\|A\|}.
    \end{align*}
    Rearrange the above inequality, and we have
    \begin{align*}
        \left( 1-  \kappa(A) \frac{\|\Delta A\|}{\|A\|}\right) \|\Delta x\| &\leq \kappa(A) \frac{\|\Delta A\|}{\|A\|} \|x\| +  \kappa(A)\frac{\|\Delta b\|}{\|A\|}\\
        \|\Delta x\|  &\leq  \frac{\|x\| \kappa(A) \frac{\|\Delta A\|}{\|A\|} +  \kappa(A)\frac{\|\Delta b\|}{\|A\| }}{ 1-  \kappa(A) \frac{\|\Delta A\|}{\|A\|}}.
    \end{align*}
    When $x\neq\zero$, we further have
    \begin{align*}
        \frac{\|\Delta x\|}{\|x\| }  &\leq \frac{\kappa(A) \frac{\|\Delta A\|}{\|A\|} +  \kappa(A)\frac{\|\Delta b\|}{\|A\|\|x\| }}{ 1-  \kappa(A) \frac{\|\Delta A\|}{\|A\|}}
        \leq \frac
        {\kappa(A)\left(  \dfrac{\|\Delta A\|}{\| A\|}+\dfrac
        {\|\Delta b\|}{\|b\|}\right)  }
        {1-\kappa(A)\dfrac{\|\Delta A\|}{\| A\|}},
    \end{align*}
    where the last inequality holds since $\|b\|=\|Ax\|\leq \|A\|\|x\|$.
\end{proof}

\begin{lemma}\label{lem:Lipschitz-Lq-norm}
    Suppose $X,Y$ are two random vectors in $\RR^d$ defined on probability space $(\Omega,\sF,\PP)$, and $f:\RR^d\rightarrow\RR$ is a $\sF$-measurable and satisfies the $L$-Lipschitz condition (in $\ell_q$-norm) almost surely.
    Then it holds that
    \[ \|f(X)-f(Y)\|_{L_q} \leq L^{1/q}  \|X-Y\|_{L_q}.\]
\end{lemma}
\begin{proof}[of \Cref{lem:Lipschitz-Lq-norm}]
    Note that
    \begin{align*}
        \|f(X)-f(Y)\|_{L_q}^q &= \int |f(X) - f(Y)|^q \rd\PP \\
        &\leq \int L \|X-Y\|_q^q \rd \PP \tag{Lipschiz condition}\\
        &=  L\sum_{j=1}^d\int |X_j-Y_j|^q \rd \PP \\
        &= L  \|X-Y\|_{L_q}^q 
    \end{align*}
    Then the conclusion follows by taking the $q^{-1}$-power on both sizes.
\end{proof}

\begin{lemma}[Error bound of regression function with estimated covariates]\label{lem:err-bound-cond-exp}
        On a common probability space $(\Omega,\sF,\PP)$, consider a random vector $W$ and a sequence of random vectors $\{V_m\}_{m\in\NN}$ adapted to a filtration $\{\sF_m\}_{m\in\NN}$ such that $\sF_m\subseteq\sF_{m+1}$.
        Suppose that 
        (i) $\|W\|_{\Lp{2}}<\infty$,
        (ii) $V_m\to V$ almost surely, and
        (iii) the function $h(v)=\EE[W \mid V=v]$ satisfies the $L$-Lipschitz condition in $\ell_2$-norm almost surely.
        Then, under (i)-(ii), it holds 
        \[\EE[W\mid V_m] \to \EE[W\mid V] \text{ in }\Lp{2},\]
        and under (i)-(iii), it holds that
        \begin{align*}
            \|\EE[W\mid V_m] - \EE[W\mid V]\|_{\Lp{2}} &\leq 2L^{\frac{1}{2}}\|V_m - V\|_{\Lp{2}}.
        \end{align*}
    \end{lemma}
    \begin{proof}[of \Cref{lem:err-bound-cond-exp}]
    Define $\sF_{\infty} = \sigma(\cup_m\sF_{m})$.
    There exists some $\sF_{\infty}$-measurable function $h$ and $\sF_{m}$-measurable function $h_m$ such that $h(V)=\EE[W \mid V]$ and $h_m(V_m)=\EE[W \mid V_m]$ almost surely.
    Notice that $(\EE[W\mid V_m])_{m\in\NN}$ is a Doob martingale (because $\|W\|_{\Lp{1}}<\infty$).
    From martingale convergence theorem, there exists $V_{\infty}=\EE[W\mid \sF_{\infty}]$ that is measurable with respect to $\sF_{\infty}$ such that $\|V_{\infty}\|_{\Lp{1}}<\infty$ and $\EE[W\mid V_m] \to V_{\infty}$ almost surely.
    On the other hand, because $V_m \to V$ almost surely from Assumption (ii), we know that $V\overset{\as}{=} V_{\infty}$ is $\sF_{\infty}$-measurable. 
    This implies that $\EE[W\mid V_{\infty}] = h(V_{\infty}) = h(V) = \EE[W\mid V]$ almost surely.
    Thus, we conclude that $\EE[W\mid V_m] \to \EE[W\mid V]$ almost surely.
    By Jensen's inequality and Assumption (ii), we have $\|\EE[W\mid V_m]\|_{\Lp{2}}\leq \|W\|_{\Lp{2}}<\infty$, which implies that the set of functions $\{\EE[W\mid V_m]:m\in\NN\}$ is uniformly integrable.
    Thus, we further have $\EE[W\mid V_m] \to \EE[W\mid V]$ in $\Lp{2}$ from dominated convergence theorem.

    Next, we need to derive the convergence rate. 
    We have that
    \begin{align}
        \|h_m(V_m) - h(V)\|_{\Lp{2}} & 
        \leq \|h_m(V_m) - h(V_m)\|_{\Lp{2}} + \|h(V_m) - h(V)\|_{\Lp{2}} .\label{eq:L2-cond-exp-1}
    \end{align}
    
    For the first term in \eqref{eq:L2-cond-exp-1}, from the martingale property, the function representation  $h_m(V_m) = \EE[h(V) \mid \sF_m]$ gives that
    \begin{align}
        \|h_m(V_m) - h(V_m)\|_{\Lp{2}} &= \|\EE[h(V) \mid \sF_m] - h(V_m)\|_{\Lp{2}} \leq \| h(V) -h(V_m)\|_{\Lp{2}}\label{eq:L2-cond-exp-2}
    \end{align}
    where the last inequality is from Jensen's inequality.
    
    Combining \eqref{eq:L2-cond-exp-1} and \eqref{eq:L2-cond-exp-2} yields that
    \begin{align*}
        \|h_m(V_m) - h(V)\|_{\Lp{2}} &\leq 2\|h(V_m) - h(V)\|_{\Lp{2}} \leq 2L^{\frac{1}{2}}\|V_m - V\|_{\Lp{2}},
    \end{align*}
    where the last inequality is from \Cref{lem:Lipschitz-Lq-norm} by noting that $h$ satisfies the $L$-Lipschitz condition from Assumption (iii).
    \end{proof}

\clearpage
\section{Doubly robust semiparametric inference}\label{app:semi}
\subsection{Proof of \Cref{thm:DR-linear} and \Cref{cor:inference}}\label{app:semi:eff}

\begin{proof}[of \Cref{thm:DR-linear} and \Cref{cor:inference}]
    \Cref{thm:DR-linear} is a special case of \Cref{thm:DR} with nonlinear link functions.
    The proof follows by applying \Cref{thm:DR} with $g$ being the identity.
    Meanwhile, the assumption in \Cref{thm:DR} can be relaxed under this special case by noting that $\EE[Y\mid X, U]$ can be replaced by $Y$ because the residual $Y-\EE[Y\mid X, U]$ is orthogonal to mean-zero functions of $(X, U)$ in the $L_2$ space so that $\eta(O) = Y - \EE[Y\mid U]$ under identity link.
\end{proof}

\subsection{Proof of \Cref{prop:simul-inference}}\label{app:semi:simul-inference}
\begin{proof}[of \Cref{prop:simul-inference}]
    From \Cref{thm:DR-linear}, we have
    \begin{align*}
        \sqrt{n}(\tilde{b} - \tbeta) &= \sqrt{n}\tilde{\Sigma}^{-1}(\PP_n-\PP)\{\tilde{\varphi}(O;{\PP})\} + \xi,
    \end{align*}
    where the remainder term $\xi$ satisfies that $\|\xi\|_{2,\infty} = \op(1)$ under the rate conditions in \Cref{thm:DR-linear}.
    Recall that $t_j = \sqrt{n}\VV_n\{\tilde{\varphi}_{\cdot j}(O;\hat{\PP})\}^{\frac{1}{2}}\hSigma^{-1}(\tilde{b}_{\cdot j} - \tbeta_{\cdot j})  $.
    We have that
    \begin{align*}
        v^{\top}t_j &= \sqrt{n}v^{\top}\VV_n\{\tilde{\varphi}_{\cdot j}(O;{\PP})\}^{\frac{1}{2}}(\PP_n-\PP)\{\tilde{\varphi}_{\cdot j}(O;{\PP})\} \\
        &\quad + (\sqrt{n}v^{\top}(\VV_n\{\tilde{\varphi}_{\cdot j}(O;\hat{\PP})\}^{\frac{1}{2}}-\VV_n\{\tilde{\varphi}_{\cdot j}(O;{\PP})\}^{\frac{1}{2}})(\PP_n-\PP)\{\tilde{\varphi}_{\cdot j}(O;{\PP})\} + v^{\top}\VV_n\{\tilde{\varphi}_{\cdot j}(O;\hat{\PP})\}^{\frac{1}{2}} \hSigma^{-1}\xi_{\cdot j})\\
        &= \vartheta_j + \varsigma_j.
    \end{align*}
    For the first component, $\vartheta_j$ for $j=1,\ldots,p$ are independent conditional on $(X,U)$'s.
    Furthermore, the self-normalizing term $\vartheta_j\to \cN(0,1)$ in distribution for $j\in\cN_p$.
    By the strong law of large number, $\VV_n\{\tilde{\varphi}_{\cdot j}(O;{\PP})\}\to \VV\{\tilde{\varphi}_{\cdot j}(O;{\PP})\}$, $\VV_n\{\tilde{\varphi}_{\cdot j}(O;\hat{\PP})\}\to \VV\{\tilde{\varphi}_{\cdot j}(O;{\PP})\}$ and $\hSigma \to \tSigma$ almost surely when both $m,n,p\rightarrow\infty$.
    This implies that $\max_j \|\VV_n\{\tilde{\varphi}_{\cdot j}(O;\hat{\PP})\} -\VV_n\{\tilde{\varphi}_{\cdot j}(O;{\PP})\}\|_{\oper}\to0$, $\max_j\|\VV_n\{\tilde{\varphi}_{\cdot j}(O;\hat{\PP})\}- \VV\{\tilde{\varphi}_{\cdot j}(O;{\PP})\}\|_{\oper}\to 0$, $\|\tSigma-\hSigma\|_{\oper}\to 0$ almost surely, which follows from \citet[Lemma S.8.6 (1)]{patil2023bagging} by noting that the variables $O$'s are iid in the triangluar array.
    For the second component, we have that 
    \begin{align*}
        \max_{1\leq j\leq p}|\varsigma_j| &\leq  \max_{1\leq j\leq p}\|\VV_n\{\tilde{\varphi}_{\cdot j}(O;\hat{\PP})\}^{\frac{1}{2}}-\VV_n\{\tilde{\varphi}_{\cdot j}(O;{\PP})\}^{\frac{1}{2}}\|_{\oper} \cdot \|\sqrt{n}(\PP_n-\PP)\{\tilde{\varphi}_{\cdot j}(O;{\PP})\}\|_2\\
        &\quad + \max_{1\leq j\leq p} \|\VV_n\{\tilde{\varphi}_{\cdot j}(O;\hat{\PP})\}^{\frac{1}{2}}   \hSigma^{-1}\|_{\oper} \|\xi_{\cdot j}\|_2 \\
        &= \op(1) \Op(1)  + \Op(1)\op(1) \\
        &= \op(1).
    \end{align*}

    Let $\varrho = |\cN_p|^{-1} \sum_{j \in \cN_p} \ind\{\left|v^{\top}t_j\right|>z_{\frac{\alpha}{2}}\}$.
    To prove the overall Type-I error control, we will show the expectation that $\varrho$ tends to $\alpha$, and its variance tends to zero.
    For the expectation, for any $\epsilon>0$, we have
    \begin{align*}
        \EE[\varrho] &=\frac{1}{\left|\cN_p\right|} \sum_{j \in \cN_p} \PP\left(\left|v^{\top}t_j\right|>z_{\frac{\alpha}{2}}\right)\\
        & \leq \frac{1}{\left|\cN_p\right|} \sum_{j \in \cN_p} \left[\PP\left(\left|\vartheta_j\right|>z_{\frac{\alpha}{2}}-\epsilon\right)+\PP\left(\left|\varsigma_j\right|>\epsilon\right)\right] \\
        & = \frac{1}{\left|\cN_p\right|} \sum_{j \in \cN_p} \PP\left(\left|\vartheta_j\right|>z_{\frac{\alpha}{2}} -\epsilon\right) +\frac{1}{\left|\cN_p\right|} \sum_{j \in \cN_p} \PP\left(\left|\varsigma_j\right|>\epsilon\right) \\
        & \leq \frac{1}{\left|\cN_p\right|} \sum_{j \in \cN_p} \PP\left(\left|\vartheta_j\right|>z_{\frac{\alpha}{2}} -\epsilon\right) +\PP\left(\max _{1 \leq j \leq p}\left|\varsigma_j\right|>\epsilon\right)  \rightarrow 2\left(1-\Phi\left(z_{\frac{\alpha}{2}}-\epsilon\right)\right) ,
    \end{align*}
    where the last convergence holds because the Cesaro mean converges to the same limit as 
    \[\lim_{n,p} \PP\left(\left|\vartheta_j\right|>
    z_{\frac{\alpha}{2}} -\epsilon\right) = 2\left(1-\Phi\left(z_{\frac{\alpha}{2}}-\epsilon\right)\right) , \]
    while the term $\PP(\max _{1 \leq j \leq p}\left|\varsigma_j\right|>\epsilon)$ varnishes.
    Similarly, we can show that $\liminf_{n,p\rightarrow\infty}\EE[\varrho]\geq 2\left(1-\Phi\left(z_{\frac{\alpha}{2}}-\epsilon\right)\right) $ for all $\epsilon>0$.
    Let $\epsilon\rightarrow0^+$, it follows that $\EE[\varrho]\to \alpha$ as $n,p\rightarrow\infty$.

    For any $\epsilon>0$, the second moment can be upper-bounded as follows:
    \begin{align*}
        \EE[\varrho^2] &= \frac{1}{\left|\cN_p\right|^2} \sum_{j, k \in \cN_p} \PP\left(\left|v^{\top}t_j\right|>z_{\frac{\alpha}{2}},\left|v^{\top}t_k\right|>z_{\frac{\alpha}{2}}\right) \\
         &= \frac{1}{\left|\cN_p\right|^2} \sum_{j \in \cN_p} \PP\left(\left|v^{\top}t_j\right|>z_{\frac{\alpha}{2}}\right)+\frac{1}{\left|\cN_p\right|^2} \sum_{j, k \in \cN_p, j \neq k} \PP\left(\left|v^{\top}t_j\right|>z_{\frac{\alpha}{2}},\left|v^{\top}t_k\right|>z_{\frac{\alpha}{2}}\right) \\
        &\leq  \frac{1}{\left|\cN_p\right|^2} \sum_{j \in \cN_p} \PP\left(\left|v^{\top}t_j\right|>z_{\frac{\alpha}{2}}\right) \\
        &\qquad +\frac{1}{\left|\cN_p\right|^2} \sum_{j, k \in \cN_p, j \neq k} \PP\left(\left|\vartheta_j\right|>z_{\frac{\alpha}{2}}-\epsilon,\left|\vartheta_k\right|>z_{\frac{\alpha}{2}}-\epsilon\right)  +\PP\left(\left|\varsigma_j\right|>\epsilon\right)+\PP\left(\left|\varsigma_k\right|>\epsilon\right) \\
        &= \frac{1}{\left|\cN_p\right|^2} \sum_{j, k \in \cN_p, j \neq k} \PP\left(\left|\vartheta_j\right|>z_{\frac{\alpha}{2}}-\epsilon,\left|\vartheta_k\right|>z_{\frac{\alpha}{2}}-\epsilon\right)+o(1) \\
        &= \frac{1}{\left|\cN_p\right|^2} \sum_{j, k \in \cN_p, j \neq k} \PP(\left|\vartheta_j\right|>z_{\frac{\alpha}{2}}-\epsilon\mid X,\hat{U})\PP(\left|\vartheta_k\right|>z_{\frac{\alpha}{2}}-\epsilon\mid X,\hat{U})+o(1) \\
        &\rightarrow4\left(1-\Phi\left(z_{\frac{\alpha}{2}}-\epsilon\right)\right)^2,
    \end{align*}
    where the last equality is from the independence of $\vartheta_j$ and $\vartheta_k$.
    We can similarly obtain the lower bound.
    Let $\epsilon\rightarrow0^+$, it follows that $\EE[\varrho^2]\rightarrow\alpha^2$ and $\VV(\varrho)\rightarrow0$ as $n,p\rightarrow\infty$.
    Combining the previous results yields that $\varrho\to \alpha$ in probability.
\end{proof}

\subsection{Nonlinear modeling}\label{app:semi:DR}\label{subsec:dr-semi-nonlinear}

    The natural extension of partial linear models to the nonlinear cases is the generalized partially linear models \citep{severini1994quasi,hardle1998testing}:
    \begin{align}
        g(\EE[Y \mid X, U]) = \beta^{\top} X + h(U) , \label{eq:gpls}
    \end{align}
    by introducing a proper link function $g$, applied element-wise on the conditional mean of the outcomes.
    Similar to the results in the previous sections, a nonlinear counterpart of the main effect estimand \eqref{eq:beta} is
    \begin{align}
       \beta(\PP) = \EE[\Cov(X \mid U)]^{-1} \EE[\Cov[X, g(\EE(Y | X, U)) \mid U] ].  \label{eq:beta-g}
    \end{align}
    Such an estimand has been considered in \citet{robins2008higher,newey2018cross} with the identity link and in \citet{vansteelandt2022assumption} with a single treatment.
    When the model \eqref{eq:gpls} is correctly specified, \eqref{eq:beta-g} is equivalent to the regression coefficient under model \eqref{eq:gpls}.
    On the other hand, when the model \eqref{eq:gpls} is misspecified, estimand \eqref{eq:beta-g} still represents a meaningful statistical quantity.

    With a differentiable link function $g$, the influence function (for $\tSigma\tbeta$) analoguous to \eqref{eq:tvarphi-linear} is:
    \[\tilde{\varphi}(O;{\PP}) = (X - {\EE}[X\mid \hat{U}]) ( {\eta}(O) - {\tbeta}^{\top}(X - {\EE}[X\mid \hat{U}]) )^{\top} ,\]
    where the main effect estimand with estimated embedding is defined as:
    \begin{align}
        \tbeta &= \EE[\Cov(X \mid \hat{U})]^{-1} \EE[\Cov[X, g(\EE(Y | X, \hat{U})) \mid \hat{U}] ],\label{eq:tbeta}
    \end{align}
    and the function $\eta$ is defined as:
    \begin{align*}
        \eta(O) &= g'(\EE[Y\mid X, \hat{U}]) \odot (Y - \EE[Y\mid X, \hat{U}]) + g(\EE[Y\mid X, \hat{U}]) - \EE[g(\EE[Y\mid X, \hat{U}]) \mid \hat{U}].
    \end{align*}
    The doubly robust semiparametric inference results in \Cref{thm:DR-linear} and \Cref{cor:inference} can be extended to accommodate nonlinear link functions, as shown in the next theorem.

    \begin{theorem}[Doubly robust inference with nonlinear link functions]\label{thm:DR}
        Under a nonparameteric model and a differentiable link function $g$, define the estimator of $\beta$ in \eqref{eq:beta-g} as:
        \begin{align}
            \hat{\beta}= \PP_n\{(X-\hat{\EE}(X\mid \hat{U}))^{\otimes 2}\}^{-1} \PP_n\{(X-\hat{\EE}(X\mid \hat{U})) 
        \cdot (\II-\PP_n)\{g(\hat{\EE}[Y \mid X,\hat{U}]) \} ^{\top} \}, \label{eq:b-hat}
        \end{align}
        which depends on empirical measure $\PP_n$ and two nuisance functions $\hat{\EE}[X\mid \hat{U}]$ and $\hat{\EE}[Y\mid X,\hat{U}]$ estimated from independent samples of $\PP_n$.
        Under \Cref{asm:moment,asm:nuisance} and assume that 
        
        (i) (Local Lipschitzness) There exists $L>0$ such that $\|g(\EE[Y\mid X, \hat{U}]) - g(\hat{\EE}[Y\mid X, \hat{U}]) - g'(\hat{\EE}[Y\mid X, \hat{U}]) \odot ({\EE}[Y\mid X, \hat{U}] - \hat{\EE}[Y\mid X, \hat{U}]) \|_{\infty} \leq  L \|{\EE}[Y\mid X, \hat{U}] - \hat{\EE}[Y\mid X, \hat{U}]\|_{\infty}^2 . $
        
        (ii) (Boundeness and consistency) \Cref{asm:moment,asm:nuisance} hold with additionally, $\|\eta(O)\|_{\Lp{2(1+\delta^{-1})}} < M$ and $\| \|\hat{\eta}(O)- {\eta}(O)\|_{\infty}\|_{\Lp{2(1+\delta)}}=\op(1)$.
        
        (iii) (Rate condition) $\|\EE[X\mid \hat{U}] - \hat{\EE}[X\mid \hat{U}]\|_{\Lp{2}}^2$, $\|{\EE}[Y\mid X, \hat{U}] - \hat{\EE}[Y\mid X, \hat{U}]\|_{\Lp{2},\infty}^2$, and $\| \EE[g(\EE[Y\mid X, \hat{U}]) \mid \hat{U}] -  \hat{\EE}[g(\hat{\EE}[Y\mid X, \hat{U}]) \mid \hat{U}] \|_{\Lp{2},\infty}  \|\EE[X\mid \hat{U}] - \hat{\EE}[X\mid \hat{U}]\|_{\Lp{2}}$ are of order $\op(n^{-\frac{1}{2}})$.
        
        Then, the estimator $\tilde{b}$ is asymptotically normal:
        \[\sqrt{n}(\tilde{b}_{\cdot j} - \tbeta_{\cdot j}) \to \cN_d (0, \tSigma^{-1}\VV\{\tilde{\varphi}_{\cdot j}(O;{\PP})\}\tSigma^{-1}) \text{ in distribution}\quad (j=1,\ldots,p).\]        
        Furthermore, if the conditions of \Cref{thm:err-bound-beta-U} hold with $\ell_m=o(n^{-\frac{1}{2}})$, then we have
        \begin{align*}
            \sqrt{n} (\tilde{b}_{\cdot j} - {\beta}_{\cdot j}) \to \cN_d (0, \tSigma^{-1}\VV\{\tilde{\varphi}_{\cdot j}(O;{\PP})\}\tSigma^{-1})\text{ in distribution}\quad (j=1,\ldots,p).
        \end{align*}
    \end{theorem}

    Compared to \Cref{thm:DR-linear}, \Cref{thm:DR} requires additional assumptions regarding the Lipschitzness of the link function around the true regression function, as noted by \citet{vansteelandt2022assumption}. It also requires boundedness and consistency assumptions on the first-order expansion term $\eta$. 
    Nevertheless, the overall conclusion is similar when both the estimators and the influence functions have a different link function for a different target estimand.
    The double robustness still allows efficient semiparametric inference with data-adaptive estimation procedures.

    \begin{algorithm}[t]\small
        \caption{Semiparametric inference for main effects with nonlinear link functions}
        \label{alg:semi}
        \begin{algorithmic}[1]
        \REQUIRE Reponses $Y$, covariate $X$, estimated latent embedding $\hat{U}$, and link function $g$.
        
        \STATE Use machine learning methods to obtain nuisance estimates $\hat{\EE}[Y\mid X,\hat{U}]$ and~$\hat{\EE}[X\mid\hat{U}]$.

        \STATE
        Use a data-adaptive fit $g(\hat{\EE}[Y\mid X,\hat{U}]) \sim \hat{U}$ to obtain the estimated regression function $\hat{\EE}[g(\hat{\EE}[Y\mid X,\hat{U}])\mid \hat{U}]$.
        If $X$ is categorical with finite support $|\cX|<\infty$, this simply reduces to $\hat{\EE}[g(\hat{\EE}[Y\mid X,\hat{U}])\mid \hat{U}] = \sum_{x\in\cX}g(\hat{\EE}[Y\mid X=x,\hat{U}])\hat{\EE}[X=x\mid \hat{U}]$.
        
        \STATE Fit a linear regression of $\hat{\eta}(O)\sim X - \hat{\EE}[X\mid \hat{U}]$ without an intercept to obtain an estimate $\tilde{b}$ as defined in \eqref{eq:b-hat} of $\tbeta$ as defined in \eqref{eq:tbeta}.

        \STATE Estimate the variance of $\tilde{b}_{\cdot j}$ by $\hat{S}_j/n$ based on \Cref{thm:DR}, where $\hat{S}_j = \hSigma^{-1}\VV_n\{\tilde{\varphi}_{\cdot j}(O;\hat{\PP})\}\hSigma^{-1}$. 

        \ENSURE Confidence intervals and p-values based on asymptotic null distribution $\tilde{b}_{\cdot j} \overset{\cdot}{\sim}\cN_d (\tbeta_{\cdot j}, \frac{\hat{S}_j}{n})$. 
        \end{algorithmic}
    \end{algorithm}

\begin{proof}[of \Cref{thm:DR}]

    From \Cref{lem:dr-est}, we have
    \begin{align*}
        \sqrt{n}(\tilde{b} - \tbeta) &= \sqrt{n}\PP \{(X-\EE(X\mid \hat{U}))^{\otimes 2}\}^{-1}(\PP_n-\PP)\{\tilde{\varphi}(O;{\PP})\} + \tilde{\xi}
    \end{align*}
    where $\tilde{\varphi}$ is defined as
    \begin{align}
        \tilde{\varphi}(O;\PP) &= (X - {\EE}[X\mid \hat{U}]) ( \eta(O) - \tbeta^{\top}(X - {\EE}[X\mid \hat{U}]) )^{\top} \label{eq:tvarphi}
    \end{align}
    and the remainder term $\tilde{\xi}$ satisfies that $\|\tilde{\xi}\|_{2,\infty} = \op(1)$.
    This proves the first statement.

    When $\|\hat{U}-U\|_{\Lp{2}}=\op(n^{-\frac{1}{2}})$, from \Cref{thm:err-bound-beta-U} we have $\|\tbeta-\beta\|_{2,\infty} = \op(n^{-\frac{1}{2}})$.
    Therefore, we further have
    \begin{align*}
        \sqrt{n}(\tilde{b} - \beta) &= \sqrt{n}(\tilde{b} - \tbeta) + \sqrt{n}(\tbeta - \beta) = \sqrt{n}\PP \{(X-\EE(X\mid \hat{U}))^{\otimes 2}\}^{-1}(\PP_n-\PP)\{\tilde{\varphi}(O;{\PP})\} + \xi,
    \end{align*}
    with $\|\xi\|_{2,\infty} = \op(1)$.

    To establish the asymptotic normality, we apply the triangle-array CLT in \Cref{lem:lindeberg}.
    This requires verifying the sufficient condition of the Lindeberg condition.
    Because $\VV\{\tSigma^{-1} \tilde{\varphi}_{\cdot j}(O;\PP)\} = \tSigma^{-1}\VV\{\tilde{\varphi}_{\cdot j}(O;\PP)\}\tSigma^{-1}$, we have
    \begin{align*}
        &\EE[\|\VV\{\tSigma^{-1}\tilde{\varphi}_{\cdot j}(O;\PP)\}^{-\frac{1}{2}}(\tSigma^{-1}\tilde{\varphi}_{\cdot j}(O;\PP))\|^{2+\frac{2}{\delta}}] \\
        =&\EE[\|\VV\{\tilde{\varphi}_{\cdot j}(O;\PP)\}^{-\frac{1}{2}}\tilde{\varphi}_{\cdot j}(O;\PP)\|^{2+\frac{2}{\delta}}] \\
        \leq & \|\VV\{\tilde{\varphi}_{\cdot j}(O;\PP)\}^{-\frac{1}{2}}\|_{\oper}^{2+\frac{2}{\delta}} \cdot \EE[\|\tilde{\varphi}_{\cdot j}(O;\PP)\|^{2+\frac{2}{\delta}}] \\
        \leq & \|\VV\{\tilde{\varphi}_{\cdot j}(O;\PP)\}^{-\frac{1}{2}}\|_{\oper}^{2+\frac{2}{\delta}}\cdot
        (\EE[\|(X-\EE[X\mid \hat{U}])\eta_j(O)\|^{2+\frac{2}{\delta}}] + \EE[\|(X-\EE[X\mid \hat{U}])^{\otimes2}\tbeta_{\cdot j} \|^{2+\frac{2}{\delta}}])\\        
        \leq & \|\VV\{\tilde{\varphi}_{\cdot j}(O;\PP)\}^{-\frac{1}{2}}\|_{\oper}^{2+\frac{2}{\delta}} \cdot 
        \EE[\|X-\EE[X\mid \hat{U}]\|^{1+\frac{1}{\delta}}\|\eta(O) \|^{1+\frac{1}{\delta}}] + \|\tbeta_{\cdot j}\|^{2+\frac{2}{\delta}}\\    
        \leq & \sigma^{-1-\frac{1}{\delta}} M^{2+\frac{2}{\delta}} + \|\tbeta_{\cdot j}\|^{2+\frac{2}{\delta}}.
    \end{align*}
    Now applying \Cref{lem:lindeberg} finishes the proof.
\end{proof}

\subsection{Auxillary lemmas}\label{app:semi:aux}

\begin{lemma}[Efficient influence function]\label{lem:eff-inf}
    Consider a random variable $O=(X,U,Y)\in\RR^d\times\RR^r\times\RR^p$ under a nonparameteric model and a differentiable function $g$, the main effect estimand in $\RR^{d\times p}$:
    \[\beta = \EE[{\Cov(X \mid U)}]^{-1}\EE[\Cov(X, g(\EE[Y \mid X, U]) \mid U)] ,\] (where $g$ is applied entry-wisely) has an efficient influence function $\mu:\RR^r\times\RR^d\times\RR^p\rightarrow\RR^{d\times p}$:
    \[\varphi(O) = \EE[{\Cov(X \mid U)}]^{-1} (X - \EE[X\mid U]) ( \eta(O) - \beta^{\top}(X - \EE[X\mid U]) )^{\top},\]
    where $\eta:\RR^r\times\RR^d\times\RR^p\rightarrow\RR^p$ is defined as:
    \begin{align*}
        \eta(O) &= g'(\EE[Y\mid X, U]) \odot (Y - \EE[Y\mid X, U]) + g(\EE[Y\mid X, U]) - \EE[g(\EE[Y\mid X, U]) \mid U].
    \end{align*}
\end{lemma}
\begin{proof}[of \Cref{lem:eff-inf}]
    The proof follows similarly as in \citet[Theorem 1]{vansteelandt2022assumption} for a univariate treatment and a univariate outcome, and extends the previous results to the multivariate cases.
    Below, we present a simplified derivation of the influence function.

    Under the nonparametric model for the observed data $O=(X,U,Y)$.
    We first calculate the efficient influence function of
    \begin{align*}
        \theta(\beta) &= \EE[(X - \EE[X\mid U]) )( g(\EE[Y\mid X, U]) - \beta^{\top} (X - \EE[X\mid U]) )^{\top}] \\
        &= \int (X - \EE[X\mid U]) )( g(\EE[Y\mid X, U]) - \beta^{\top} X ) \rd P(O),
    \end{align*}
    where $P(O)$ is the joint distribution of data.
    Note that by the definition of $\beta$, we have $\theta(\beta)=0$.

    Consider a one-dimensional submodel of $p(O)$ indexed by a scalar parameter $t$, and let $S_t(o) = \partial \log \rd P_t(o) / \partial t\mid _{t=0}$ denote the score function of the submodel.
    Similarly, let $S_t (Y \mid X, U)$, $S_t (X \mid U)$ and $S_t(U)$ be the scores w.r.t. $t$ in that parametric submodel, corresponding to the distributions $p(Y \mid X, U)$, $p(X \mid U)$ and $p(U)$, respectively
    Taking the derivative of $\theta$ w.r.t. $t$, we obtain
    \begin{align*}
        \frac{\partial\theta(\beta)}{\partial t}\Big|_{t=0} 
        &=    \int \frac{\partial(X - \EE_t[X\mid U])}{\partial t} \Big|_{t=0} ( g(\EE[Y\mid X, U]) - \beta^{\top} X )^{\top}   \rd P(O) \\
        &\qquad + \int (X - \EE[X\mid U]) )\left( g'(\EE[Y\mid X, U]) \odot \frac{\partial \EE[Y\mid X, U] }{\partial t} \Big|_{t=0} \right)^{\top}\rd P(O)\\
        &\qquad + \int (X - \EE[X\mid U]) )( g(\EE[Y\mid X, U]) - \beta^{\top} X )^{\top}  \frac{\partial p_{t}(X,U)}{\partial t} \Big|_{t=0} \rd O\\        
        &= - \int (X - \EE[X \mid U])\EE[ g(\EE[Y\mid X, U]) - \beta^{\top} X \mid U] ^{\top}  S_t(X \mid U)  \rd P(O)\\
        &\qquad + \int (X - \EE[X\mid U]) )(g'(\EE[Y\mid X, U]) \odot (Y - \EE[Y \mid X,U]))^{\top} S_t(Y\mid X,U) \rd P(O)\\
        &\qquad + \int (X - \EE[X\mid U]) )( g(\EE[Y\mid X, U]) - \beta^{\top} X )^{\top}  S_t(X,U) \rd P(O), 
    \end{align*}
    where in the first equality, we apply the product and chain rules \citep[Section 3.4.3]{kennedy2022semiparametric}; and in the second equality, we use the identity $S_t(Z) = \partial \log p_t(Z) /\partial t = (\partial p_t(Z) /\partial t) / p_t(Z)$ for score functions.
    
    Note that
    \begin{align*}
        S_t(O) &= S_t (Y \mid X, U) + S_t (X \mid U) + S_t(U).
    \end{align*}
    From the zero mean properties of scores and $\theta(\beta)=0$, we further have
    \begin{align*}
        \frac{\partial\theta(\beta)}{\partial t}\Big|_{t=0} 
        &= - \int (X - \EE[X \mid U])\EE[ g(\EE[Y\mid X, U]) - \beta^{\top} X \mid U] ^{\top}  S_t(O)  \rd P(O)\\
        &\qquad + \int (X - \EE[X\mid U]) )(g'(\EE[Y\mid X, U]) \odot (Y - \EE[Y \mid X,U]))^{\top} S_t(O) \rd P(O)\\
        &\qquad + \int (X - \EE[X\mid U]) )( g(\EE[Y\mid X, U]) - \beta^{\top} X )^{\top}  S_t(O) \rd P(O)\\
        &= \int (X - \EE[X\mid U]) )( \eta(O) - \beta^{\top}(X - \EE[X\mid U]) )^{\top} S_t(O) \rd P(O) ,
    \end{align*}
    which implies that $(X - \EE[X\mid U]) )( \eta(O) - \beta^{\top}(X - \EE[X\mid U]) )^{\top}$ is an influence function for $\theta$.
    From a similar argument in the proof of Theorem 1 in \citet{vansteelandt2022assumption}, it is also the efficient influence function of $\theta( \beta )$ under the nonparametric model.
    Consequently, by chain rule $\partial\theta/\partial t= (\partial\theta/\partial \beta) (\partial\beta/\partial t)$, the conclusion follows by taking the inverse of $\partial\theta/\partial \beta$.    
\end{proof}

\begin{remark}[Alternative expression of the estimand]
    Note that the first part of the influence function also gives an alternative expression for $\beta$:
    \begin{align}
        \beta &=  \EE[{\Cov(X \mid U)}]^{-1} \EE[(X - \EE[X\mid U])  \eta(O)^{\top}] \label{eq:beta-mu}
    \end{align}
    because 
    \begin{align}
        \EE[(X - \EE[X\mid U])(g'(\EE[Y\mid X, U]) \odot (Y - \EE[Y\mid X, U]))^{\top}] &= 0, \label{eq:iter-exp-1}
    \end{align}
    by the law of iterated expectation.    
\end{remark}

\begin{lemma}[Doubly robust estimation]\label{lem:dr-est}
    Consider the setting in \Cref{lem:eff-inf}.
    Define a plug-in estimator of $\beta$:
    \[\hat{\beta}= \PP_n\{(X-\hat{\EE}(X\mid U))^2\}^{-1} \PP_n\{(X-\hat{\EE}(X\mid U)) 
    \cdot (\II-\PP_n)\{g(\hat{\EE}[Y \mid X,U]) \} ^{\top} \}\]
    which depends on empirical measure $\PP_n$ and two nuisance functions $\hat{\EE}[X\mid U]$ and $\hat{\EE}[Y\mid X,U]$ estimated from independent samples of $\PP_n$.
    Define the population and empirical variance by
    \begin{align*}
        \Sigma &=\PP \{(X-\EE(X\mid U))^{\otimes 2}\}\\
        \hSigma &=\PP_n \{(X-\hat{\EE}(X\mid U))^{\otimes 2}\},
    \end{align*}
    the empirical influence function (for $\Sigma\beta$) by:
    \[\varphi(O;\hat{\PP}) = (X - \hat{\EE}[X\mid U]) ( \hat{\eta}(O) - \hat{\beta}^{\top}(X - \hat{\EE}[X\mid U]) )^{\top} .\]
    Suppose the following conditions hold:
    \begin{itemize}
        \item (Regularity conditions) 
        There exists $\sigma>0$ such that $\Sigma\succeq\sigma I_d$, $\hSigma\succeq\sigma I_d$.        

        \item (Bounded moments and consistency)
        There exists $\delta\in(0,1]$ and $M>0$, such that 
        \[
            \|\beta\|_{2,\infty}
            \vee
            \|X - \EE[X\mid U]\|_{\Lp{2(1+\delta^{-1})}} 
            \vee
            \|X - \hat{\EE}[X\mid U]\|_{\Lp{2(1+\delta^{-1})}}
            \vee
            \|\eta(O)\|_{\Lp{2(1+\delta^{-1})}} < M
        \]
        \[
            \|\EE[X\mid U] - \hat{\EE}[X\mid U]\|_{\Lp{2(1+\delta)}}, \| \|\hat{\eta}(O)- {\eta}(O)\|_{\infty}\|_{\Lp{2(1+\delta)}}=\op(1)
        \]

        \item (Local Lipshitzness) There exists $L>0$ such that
        \begin{align}
            &\|g(\EE[Y\mid X, U]) - g(\hat{\EE}[Y\mid X, U]) -
            g'(\hat{\EE}[Y\mid X, U]) \odot ({\EE}[Y\mid X, U] - \hat{\EE}[Y\mid X, U]) \|_{\infty} \notag \\
            \leq & L \|{\EE}[Y\mid X, U] - \hat{\EE}[Y\mid X, U])\|_{\infty}^2 . \label{eq:g-taylor}
        \end{align}
    \end{itemize}
    Then, it holds that
    \begin{align*}
        \sqrt{n}(\hat{\beta} - \beta) &=  \sqrt{n}\Sigma^{-1}(\PP_n-\PP)\{\varphi(O;{\PP})\} + \xi,
    \end{align*}
    where for any $\epsilon>0$, there exists a constant $C=C(\epsilon,\sigma,M,L)$, such that with probability at least $1- 3\epsilon$, the remainder term satisfies that 
    \begin{align*}
        \|\xi\|_{2,\infty}
        &\leq 
         C \{ \|(\PP_n-\PP)\{(X-\EE[X\mid U])^{\otimes 2}\}\|_{\oper}+ 
          \| {\EE}[X\mid U] - \hat{\EE}[X\mid U]\|_{\Lp{2(1+\delta)}} + \| \|\eta(O)- \hat{\eta}(O)\|_{\infty}\|_{\Lp{2(1+\delta)}} \} \\
        &\qquad + C\sqrt{n}\{
        \|\EE[X\mid U] - \hat{\EE}[X\mid U]\|_{\Lp{2}}^2\\
        &\qquad\qquad + ML\|{\EE}[Y\mid X, U] - \hat{\EE}[Y\mid X, U]\|_{\Lp{2},\infty}^2 \notag\\
        &\qquad\qquad + \| \EE[g(\EE[Y\mid X, U]) \mid U] -  \hat{\EE}[g(\hat{\EE}[Y\mid X, U]) \mid U] \|_{\Lp{2},\infty}  \|\EE[X\mid U] - \hat{\EE}[X\mid U]\|_{\Lp{2}} \},
    \end{align*}
    When 
    $\|\EE[X\mid U] - \hat{\EE}[X\mid U]\|_{\Lp{2}}^2$, $\|{\EE}[Y\mid X, U] - \hat{\EE}[Y\mid X, U]\|_{\Lp{2},\infty}^2$, and $\| \EE[g(\EE[Y\mid X, U]) \mid U] -  \hat{\EE}[g(\hat{\EE}[Y\mid X, U]) \mid U] \|_{\Lp{2},\infty}  \|\EE[X\mid U] - \hat{\EE}[X\mid U]\|_{\Lp{2}}$ are of order $\op(n^{-\frac{1}{2}})$, we further have that $\|\xi\|_{2,\infty} = \op(1)$ and hence 
    \[\sqrt{n}(\hat{\beta}_{\cdot j} - \beta_{\cdot j}) \to \cN_d (0, \Sigma^{-1}\VV\{\varphi_{\cdot j}(O;{\PP})\})\text{ in distribution}\quad (j=1,\ldots,p).\]
\end{lemma}
\begin{proof}[of \Cref{lem:dr-est}]

    From the definition of $\hat{\beta}$, we have $\PP_n\{\varphi(O;\hat{\PP})\} = 0$.
    Therefore, $\hat{\beta} $ is also a one-step estimator.
    We begin with a three-term decomposition of the estimation error (see, for example, \citet[Equation (2.2)]{du2024causal} and \citet[Equation (10)]{kennedy2022semiparametric}):
    \begin{align}
        \hSigma \sqrt{n}(\hat{\beta} - \beta) &= \sqrt{n}(\PP_n-\PP)\{\varphi(O;\PP)\} \notag\\
        &\qquad + \sqrt{n}(\PP_n-\PP)\{\varphi(O;\hat{\PP}) - \varphi(O;\PP)\} + \sqrt{n}(\hSigma - \tilde{\Sigma})(\hat{\beta} - \beta) \notag\\
        &\qquad + \sqrt{n}\tilde{\Sigma}(\hat{\beta} - \beta)  + \sqrt{n}\PP\{\varphi(O;\hat{\PP})\} \notag\\
        &= C + T_1 + T_2, \label{eq:decom}
    \end{align}
    where $\tilde{\Sigma} = \PP\{(X - \hat{\EE}[X\mid U])^{\otimes 2}\}$.
    By the central limit theorem, each entry of the first term $C$ is $\Op(1)$.
    We next derive finite-sample deviation bounds for the other terms and show that they are $\op(1)$ under the extra rate conditions as assumed.

    \paragraph{Part (1) Controlling the empirical process term $T_1$.}
    We begin by decomposing $T_1$:
    \begin{align*}
        &\varphi(O;\hat{\PP}) - \varphi(O;\PP) + (X - \hat{\EE}[X\mid U])^{\otimes 2}(\hbeta-\beta)\\
        =&(X - \hat{\EE}[X\mid U]) ( \hat{\eta}(O) - \hat{\beta}^{\top}(X - \hat{\EE}[X\mid U]) )^{\top}
        - (X - {\EE}[X\mid U]) ( {\eta}(O) - {\beta}^{\top}(X - {\EE}[X\mid U]) )^{\top}\\
        &\qquad + (X - \hat{\EE}[X\mid U])^{\otimes 2}(\hbeta-\beta)\\
        =& [(X-\EE[X\mid U])^{\otimes 2}  - (X- \hat{\EE}[X\mid U])^{\otimes 2} ] {\beta} + [(X - \hat{\EE}[X\mid U]) \hat{\eta}(O)^{\top} -  (X - {\EE}[X\mid U]) {\eta}(O)^{\top}]\\
        =&: S_1+S_2.
    \end{align*}
    Note that each term above takes the form of $\hat{a}\hat{b} - a b = \hat{a}(\hat{b} - b) + (\hat{a} - a)b$, which we will next use to derive the upper bound.

    For the first term, we have
    \begin{align*}
        &\sqrt{n}\|(\PP_n-\PP)S_1\|_{2,\infty} \\
        =& \sqrt{n}
        \|
        (\PP_n-\PP)[(X- \hat{\EE}[X\mid U])^{\otimes 2} - (X-\EE[X\mid U])^{\otimes 2}] {\beta} \|_{2,\infty}\\
        =& \sqrt{n}
        \| (\PP_n-\PP)\{A_1\} {\beta} \|_{2,\infty},
    \end{align*}
    where
    \begin{align*}
        A_1 &= (X- \hat{\EE}[X\mid U])^{\otimes 2} - (X-\EE[X\mid U])^{\otimes 2}.
    \end{align*}
    From \Cref{lem:chebyshev}, we have
    \begin{align*}
        \sqrt{n}\|(\PP_n-\PP)S_1\|_{2,\infty} 
        \leq \epsilon^{-\frac{1}{2}}\EE[\|A_1\|_{\oper}^2\|\beta\|_{2,\infty}^2]^{\frac{1}{2}}  \leq \epsilon^{-\frac{1}{2}}\EE[\|A_1\|_{\oper}^2]^{\frac{1}{2}} \|\beta\|_{2,\infty} ,
    \end{align*}
    with probability at least $1-\epsilon$.
    Now, it remains to derive the upper bound of the expected squared operator norm:
    \begin{align*}
        \EE[\|A_1\|_{\oper}^2]^{\frac{1}{2}} &\leq \EE[\| {\EE}[X\mid U] - \hat{\EE}[X\mid U]\|_2^2 (\|X- {\EE}[X\mid U]\|_2 + \|X- \hat{\EE}[X\mid U]\|_2)^2]^{\frac{1}{2}} \notag\\
        &\leq \| {\EE}[X\mid U] - \hat{\EE}[X\mid U]\|_{\Lp{2(1+\delta)}} (\|X- {\EE}[X\mid U]\|_{\Lp{2(1+\delta^{-1})}} + \|X- \hat{\EE}[X\mid U]\|_{\Lp{2(1+\delta^{-1})}}). 
    \end{align*}
    Therefore, we have
    \begin{align*}
        \sqrt{n}\|(\PP_n-\PP)S_1\|_{2,\infty} 
        \leq 2M^2\| {\EE}[X\mid U] - \hat{\EE}[X\mid U]\|_{\Lp{2(1+\delta)}}
    \end{align*}
    with probability at least $1-\epsilon$.

    For the second term, similarly, we have 
    \begin{align*}
        &\sqrt{n}\|(\PP_n-\PP)S_2\|_{2,\infty} \\        
        \leq& \epsilon^{-\frac{1}{2}}\max_{j\in[p]}\EE[\|(X - \hat{\EE}[X\mid U]) [(\hat{\eta}(O)- {\eta}(O))^{\top}]_{\cdot j} + 
        (\EE[X\mid U] - \hat{\EE}[X\mid U])[{\eta}(O)^{\top}]_{\cdot j} \|^2]^{\frac{1}{2}}\\
        \leq& \epsilon^{-\frac{1}{2}}(\|X - \hat{\EE}[X\mid U]\|_{\Lp{2(1+\delta^{-1})}} \| \|(\hat{\eta}(O)- {\eta}(O))^{\top}\|_{2,\infty} ^2\|_{\Lp{1+\delta}}^{\frac{1}{2}} + \|\EE[X\mid U] - \hat{\EE}[X\mid U] \|_{\Lp{2(1+\delta)}} \|{\eta}(O)^{\top}\|_{\Lp{2(1+\delta^{-1})},\infty})\\
        \leq& \epsilon^{-\frac{1}{2}}M (\| \|\hat{\eta}(O)- {\eta}(O)\|_{\infty}\|_{\Lp{2(1+\delta)}} + \|\EE[X\mid U] - \hat{\EE}[X\mid U] \|_{\Lp{2(1+\delta)}})
    \end{align*}
    with probability at least $1-\epsilon$.
            
    Combining the above results, with probability at least $1-2\epsilon$, we have 
    \begin{align}
        \|T_1\|_{2,\infty} & \leq 2M^2
          \| {\EE}[X\mid U] - \hat{\EE}[X\mid U]\|_{\Lp{2(1+\delta)}} \notag\\
         &\qquad + M (\| \|\hat{\eta}(O)- {\eta}(O)\|_{\infty}\|_{\Lp{2(1+\delta)}} + \|\EE[X\mid U] - \hat{\EE}[X\mid U] \|_{\Lp{2(1+\delta)}}) \notag\\
         &\leq 2M(M\vee 1)
          \| {\EE}[X\mid U] - \hat{\EE}[X\mid U]\|_{\Lp{2(1+\delta)}}  + M \| \|\hat{\eta}(O)- {\eta}(O)\|_{\infty}\|_{\Lp{2(1+\delta)}}. \label{eq:T1}
    \end{align}

    \paragraph{Part (2) Controlling the bias term $T_2$.}
    For the third term $T_2$ in \eqref{eq:decom}, we have
    \begin{align}
        T_2 &= \sqrt{n} 
        \tilde{\Sigma}(\hat{\beta} -\beta ) + \sqrt{n}\PP\{\varphi(O;\hat{\PP})\}  \notag \\
        &= \sqrt{n}\PP\{ (X - \hat{\EE}[X\mid U]) \hat{\eta}(O)^{\top} \} - \sqrt{n}\tSigma\beta \notag\\
        &=  \sqrt{n}\PP\{ (X - \hat{\EE}[X\mid U]) (g'(\hat{\EE}[Y\mid X, U]) \odot (Y - \hat{\EE}[Y\mid X, U]) + g(\hat{\EE}[Y\mid X, U]) - \hat{\EE}[g(\hat{\EE}[Y\mid X, U]) \mid U])^{\top} \} \notag\\
        &\qquad - \sqrt{n}\tSigma\Sigma^{-1}\Sigma\beta , \label{eq:T3-expand}
    \end{align}
    where the last equality is because of \Cref{eq:beta-mu,eq:iter-exp-1}.
    Denote the second-order remaining term by $Q = g(\EE[Y\mid X, U]) - g(\hat{\EE}[Y\mid X, U]) -
            g'(\hat{\EE}[Y\mid X, U]) \odot ({\EE}[Y\mid X, U] - \hat{\EE}[Y\mid X, U])$.
    Then, we further have
    \begin{align}
        T_2 &= - \sqrt{n}\PP\{ (X - \hat{\EE}[X\mid U]) (g(\EE[Y\mid X, U]) + \hat{\EE}[g(\hat{\EE}[Y\mid X, U]) \mid U] + Q )^{\top} \} \notag\\
        &\qquad + \sqrt{n}\PP\{(X - \EE[X\mid U])  (g(\EE[Y\mid X, U]) - \EE[g(\EE[Y\mid X, U]) \mid U])^{\top}\} \notag \\
        &\qquad + \sqrt{n}(I_d-\tSigma\Sigma^{-1})\Sigma\beta \notag \\
        &= \sqrt{n}\PP\{ (X - \hat{\EE}[X\mid U])  Q ^{\top} \} \notag\\
        &\qquad + \sqrt{n}\PP\{ (\EE[X\mid U] - \hat{\EE}[X\mid U]) (\EE[g(\EE[Y\mid X, U])\mid U] - \hat{\EE}[g(\hat{\EE}[Y\mid X, U]) \mid U] )^{\top}\}  \notag\\
        &\qquad + \sqrt{n}(\Sigma-\tSigma)\beta. \label{eq:T2-expand}
    \end{align}

    Because by the law of iterative expectation,
    \begin{align}
        \PP\{(X - \EE[X\mid U])(\EE[X\mid U] - \hat{\EE}[X\mid U])^{\top}\} &= 0,
    \end{align}
    we have
    \begin{align*}
        1- \tSigma\Sigma^{-1} &= 1- \PP\{(X - \EE[X\mid U] + \EE[X\mid U] - \hat{\EE}[X\mid U])^{\otimes 2}\} \Sigma^{-1}\\
        &=-  \PP\{(X - \EE[X\mid U])(\EE[X\mid U] - \hat{\EE}[X\mid U])^{\top}\} - 
        \PP\{(\EE[X\mid U] - \hat{\EE}[X\mid U])(X - \EE[X\mid U])^{\top}\}\\
        &\qquad + \PP\{(\EE[X\mid U] - \hat{\EE}[X\mid U])^{\otimes 2}\} \Sigma^{-1}\\
        &=\PP\{(\EE[X\mid U] - \hat{\EE}[X\mid U])^{\otimes 2}\} \Sigma^{-1}
    \end{align*}
    and
    \begin{align}
        \|\Sigma - \tSigma\|_{\oper} 
        &=\|(I_d - \tSigma\Sigma^{-1})\Sigma\|_{\oper} \\
        &= \|\PP\{(\EE[X\mid U] - \hat{\EE}[X\mid U])^{\otimes 2}\} \|_{\oper} \notag\\
        &\leq \|\PP\{(\EE[X\mid U] - \hat{\EE}[X\mid U])^{\otimes 2}\}\|_{\oper} \tag{Jensen's inequality}\\
        &\leq \PP\{\|(\EE[X\mid U] - \hat{\EE}[X\mid U])^{\otimes 2}\|_{\oper}\} \notag\\
        &= \|\EE[X\mid U] - \hat{\EE}[X\mid U]\|_{\Lp{2}}^2 \label{eq:oper-norm-Sigma}
    \end{align}

    Combining \Cref{eq:T2-expand,eq:oper-norm-Sigma,eq:g-taylor} yields that
    \begin{align}
        \|T_2\|_{2,\infty} & \leq \sqrt{n}\|Q\|_{\Lp{2},\infty} \|X - \hat{\EE}[X\mid U]\|_{\Lp{2}} \notag\\
        &\qquad + \sqrt{n}\| \EE[g(\EE[Y\mid X, U]) \mid U] -  \hat{\EE}[g(\hat{\EE}[Y\mid X, U]) \mid U] \|_{\Lp{2},\infty}  \|\EE[X\mid U] - \hat{\EE}[X\mid U]\|_{\Lp{2}} \notag\\
        &\qquad + \sqrt{n}\|\Sigma- \tSigma\|_{\oper} \|\beta\|_{2,\infty} \notag\\
        & \leq  \sqrt{n}ML\|{\EE}[Y\mid X, U] - \hat{\EE}[Y\mid X, U]\|_{\Lp{2},\infty}^2\notag\\
        &\qquad + \sqrt{n}\| \EE[g(\EE[Y\mid X, U]) \mid U] -  \hat{\EE}[g(\hat{\EE}[Y\mid X, U]) \mid U] \|_{\Lp{2},\infty}  \|\EE[X\mid U] - \hat{\EE}[X\mid U]\|_{\Lp{2}} \notag\\
        &\qquad + M \sqrt{n} \|\EE[X\mid U] - \hat{\EE}[X\mid U]\|_{\Lp{2}}^2  .\label{eq:T2}
    \end{align}

    \paragraph{Part (3) Combining the above results.}
    Finally, from \Cref{eq:decom,eq:T1,eq:T2}
    \begin{align*}
         \hSigma \sqrt{n}(\hat{\beta} - \beta) &=  \sqrt{n}(\PP_n-\PP)\{\varphi(O;\PP)\} + \varsigma
    \end{align*}
    for some $\varsigma\in\RR^{d\times p}$ with 
    \begin{align*}
        \|\varsigma\|_{2,\infty} &\leq 2M(M\vee 1)
          \| {\EE}[X\mid U] - \hat{\EE}[X\mid U]\|_{\Lp{2(1+\delta)}} + \sqrt{n}M \|\EE[X\mid U] - \hat{\EE}[X\mid U]\|_{\Lp{2}}^2 \notag\\
         &\qquad + M \| \|\hat{\eta}(O)- {\eta}(O)\|_{\infty}\|_{\Lp{2(1+\delta)}} \\
         &\qquad + \sqrt{n}ML\|{\EE}[Y\mid X, U] - \hat{\EE}[Y\mid X, U]\|_{\Lp{2},\infty}^2 \notag\\
        &\qquad + \sqrt{n}\| \EE[g(\EE[Y\mid X, U]) \mid U] -  \hat{\EE}[g(\hat{\EE}[Y\mid X, U]) \mid U] \|_{\Lp{2},\infty}  \|\EE[X\mid U] - \hat{\EE}[X\mid U]\|_{\Lp{2}}.
    \end{align*}
    Note that
    \begin{align*}
        \|\Sigma^{-1}-\hSigma^{-1}\|_{\oper} &= \|\hSigma^{-1} (\hSigma - \Sigma) \Sigma^{-1}\|_{\oper} \\
        &\leq \|\hSigma^{-1}\|_{\oper} \|\hSigma - \Sigma\|\|_{\oper} \| \Sigma^{-1}\|_{\oper} \\
        &\leq \sigma^2 \|\hSigma - \Sigma\|_{\oper} \\
        &\leq \sigma^2\|(\PP_n-\PP)\{(X-\EE[X\mid U])^{\otimes 2}\}\|_{\oper} +  \sigma^2\|\EE[X\mid U] - \hat{\EE}[X\mid U]\|_{\Lp{2}}^2,
    \end{align*}
    where the first equality is from $\hSigma^{-1} (\hSigma - \Sigma) \Sigma^{-1} = \Sigma^{-1}-\hSigma^{-1}$, the second inequality is from the positivity assumption that $\|\hSigma^{-1}\|_{\oper}\leq \sigma $, $\|\Sigma^{-1}\|_{\oper}\leq \sigma $, and the last inequality is from \eqref{eq:oper-norm-Sigma}.
    We further have
    \begin{align*}
         \sqrt{n}(\hat{\beta} - \beta) &=  \sqrt{n}\Sigma^{-1}(\PP_n-\PP)\{\varphi(O;\PP)\} + \xi,
    \end{align*}
    with 
    \begin{align*}
        \xi &= \sqrt{n}(\hSigma^{-1}-\Sigma^{-1})(\PP_n-\PP)\{\varphi(O;\PP)\} +  \hSigma^{-1} \varsigma.
    \end{align*}
    By the multidimensional Chebyshev inequality and union bound, with probability at least $1-3\epsilon$,
    \begin{align*}
        \|\xi\|_{2,\infty}
        &\leq
        \sigma^2
        (\|(\PP_n-\PP)\{(X-\EE[X\mid U])^{\otimes 2}\}\|_{\oper} + \|\EE[X\mid U] - \hat{\EE}[X\mid U]\|_{\Lp{2}}^2)\\
        &\qquad \cdot \|\EE[X\mid U] - \hat{\EE}[X\mid U] \|_{\Lp{2(1+\delta^{-1})}}
        ( 
        \|\eta(O)\|_{\Lp{2(1+\delta)},\infty}
        +
       \|\beta\|_{2,\infty}
        )  + \sigma\|\varsigma\|_{2,\infty}\\
        &\leq 
        2\sigma^2 M^2 \|(\PP_n-\PP)\{(X-\EE[X\mid U])^{\otimes 2}\}\|_{\oper} \\
        &\qquad + 2M(M\vee 1)
          \| {\EE}[X\mid U] - \hat{\EE}[X\mid U]\|_{\Lp{2(1+\delta)}} + M \| \|\hat{\eta}(O)- {\eta}(O)\|_{\infty}\|_{\Lp{2(1+\delta)}} \\
        &\qquad + \sqrt{n}2(\sigma^2\vee1)M(M\vee 1)\|\EE[X\mid U] - \hat{\EE}[X\mid U]\|_{\Lp{2}}^2\\
        &\qquad + \sqrt{n}ML\|{\EE}[Y\mid X, U] - \hat{\EE}[Y\mid X, U]\|_{\Lp{2},\infty}^2 \notag\\
        &\qquad + \sqrt{n}\| \EE[g(\EE[Y\mid X, U]) \mid U] -  \hat{\EE}[g(\hat{\EE}[Y\mid X, U]) \mid U] \|_{\Lp{2},\infty}  \|\EE[X\mid U] - \hat{\EE}[X\mid U]\|_{\Lp{2}}.
    \end{align*}
    Under the extra rate conditions as assumed, we further have $\|\xi\|_{2,\infty} = \op(1) $.
    This completes the proof.
\end{proof}

\begin{lemma}[Multivariate Lindeberg CLT for triangular array]\label{lem:lindeberg}
    Let $m=m_n$ and $p=p_n$ be two sequences indexed by $n$.
    Consider the influence-function-based linear expansion for estimator $\hat{\tau}_j$ of $\tau_j\in\RR^d$:
    \[\sqrt{n}(\hat{\tau}_j - \tau_j) = \sqrt{n}(\PP_n-\PP)\{\varphi_{m_nj}\} + \varsigma_{m_nj}\quad (j=1,\ldots,p)\]
    where $\varphi_{m_nj}$ is the influence function that depends on $m$ and the residual $\varsigma_j$'s satisfy that $\|\varsigma_{m_n}\|_{2,\infty}=\op(1)$ as $n\rightarrow\infty$.
    Further assume that (i) there exists a constant $c>0$, such that $\VV(\varphi_{m_nj}(O_1))\geq c$, and (ii) $\max_{k\in[n]}\EE[\|\VV\{\varphi_{m_nj}(O_k)\}^{-\frac{1}{2}}\varphi_{m_nj}(O_k)\|^{2+\frac{2}{\delta}}]\leq M$, then 
    \[\sqrt{n} \VV\{\varphi_{m_nj}\}^{-1/2} ({\hat{\tau}_j - \tau_j}) \to \cN_d(0,I_d) \text{ in distribution}.\]
\end{lemma}
\begin{proof}[of \Cref{lem:lindeberg}]
    Note that $\varphi_{m_nj}$ is the centered influence function such that $\EE[\varphi_{m_nj}(O)] = 0$.
    Let $X_{nk} = \VV\{\varphi_{m_nj}(O_k)\}^{-\frac{1}{2}}\varphi_{m_nj}(O_k)$.
    From assumption (ii) that $\max_{k\in[n]}\EE[\|X_{nk}\|^{2+\frac{2}{\delta}}]\leq M$, we have that, for any $\xi>0$,
    \begin{align*}
        \frac{1}{n}\sum_{k=1}^{n} \EE\left[ \|X_{nk}\|^2 \ind\{\|X_{n k}\| \geq \xi \sqrt{n}\} \right] \leq \frac{1}{n^{1+\frac{1}{\delta}}}\sum_{k=1}^{n} \EE\left[ \|X_{nk}\|^{2+\frac{2}{\delta}}\right] \leq M\frac{n}{n^{1+\frac{1}{\delta}}}\to 0.
    \end{align*}
    This verifies Lindeberg's condition for a triangular array of random variables.
    From the multivariate Lindeberg's theorem (e.g., \citet[Theorem 29.5]{billingsley1995probability}), it follows that
    \[ \sqrt{n}\VV\{\varphi_{m_nj}\}^{-\frac{1}{2}}(\PP_n-\PP)\{\varphi_{m_nj}\} \to \cN(0,1) \text{ in distribution},\]
    as $n\rightarrow\infty$.
    From assumption (i) that $\VV(\varphi_{m_nj})\geq c>0$ and $\max_{j\in p_n}\|\varsigma_{m_nj}\|=\op(1)$, we further have 
    \[\max_{j\in p_n} \| \VV\{\varphi_{m_nj}\}^{-\frac{1}{2}}\varsigma_{m_nj}\| = \op(1)\]
    as $n\rightarrow\infty$.
    Consequently, the conclusion follows.
\end{proof}

\begin{lemma}[Matrix Chebyshev inequality]\label{lem:chebyshev}
    Let $A$ denote a random matrix in $\RR^{d\times r}$ and $\beta \in\RR^{r\times p}$ such that $\EE[A\beta] = 0_{d\times p}$.
    Then with probability at least $1-\epsilon$, it holds that
    \begin{align*}
        \sqrt{n}\|(\PP_n-\PP)\{A\beta\}\|_{2,\infty} 
        &\leq \epsilon^{-\frac{1}{2}} \EE[ \|A\|_{\oper}^2 \|\beta\|_{2,\infty}^2]^{\frac{1}{2}},
    \end{align*}
    and
    \begin{align*}
        \sqrt{n}\|(\PP_n-\PP)\{A\}\|_{\oper} 
        &\leq \epsilon^{-\frac{1}{2}} \EE[ \|A\|_{\oper}^2]^{\frac{1}{2}}.
    \end{align*}
\end{lemma}
\begin{proof}[of \Cref{lem:chebyshev}]
    By Chebyshev's inequality, we have
    \begin{align*}
        \PP(\sqrt{n}\|(\PP_n-\PP)\{A\beta\}\|_{2,\infty} > t) &\leq \frac{n\EE[\|(\PP_n-\PP)\{A\beta\}\|_{2,\infty}^2]}{t^2}\\
        &\leq \frac{n\EE[\|A\beta - \EE[A\beta]\|_{2,\infty}^2]}{nt^2} \\
       &= \frac{\EE[\|A\beta \|_{2,\infty}^2]}{t^2}.
    \end{align*}
    Choosing $t=\EE[\|A\beta \|_{2,\infty}^2]^{\frac{1}{2}}\epsilon^{-\frac{1}{2}} $ yields that, with probability at least $1-\epsilon$,
    \begin{align*}
        \sqrt{n}\|(\PP_n-\PP)\{A\beta\}\|_{2,\infty} \leq \epsilon^{-\frac{1}{2}}\EE[\|A\beta \|_{2,\infty}^2]^{\frac{1}{2}} \leq \epsilon^{-\frac{1}{2}}\EE[ \|A\|_{\oper}^2 \|\beta\|_{2,\infty}^2]^{\frac{1}{2}},
    \end{align*}
    which finishes the proof of the first statement.

    Similarly, considering all unit vectors in the unit sphere $\SSS^{r-1}$ (i.e., the set of vector $v\in\RR^r$ such that $\|v\|_2=1$), it holds that
    \begin{align*}
        \PP\left(\sqrt{n}\|(\PP_n-\PP)\{A\}\|_{\oper} > t\right) = \PP\left(\sup_{v\in\SSS^{r-1}}\sqrt{n}\|(\PP_n-\PP)\{Av\}\|_{2} > t\right) \leq \frac{\EE[\|A\|_{\oper}^2]}{t^2}.
    \end{align*}
    The second conclusion follows by choosing $t=(\EE[\|A\|_{\oper}^2]/\epsilon)^{\frac{1}{2}}$.
\end{proof}

\clearpage
\section{Experiment details and extra results}\label{app:ex}

\subsection{Simulation}\label{app:ex-simu}

\paragraph{Convergence rate}

    As in \Cref{thm:DR-linear,thm:DR}, the nuisance functions need to be estimated fast enough such that valid inference can be guaranteed.

    For the estimation of the latent embedding $U$, we acknowledge your valid point regarding sparse loadings in standard factor analysis, which is common in genomics. Our claim relies on using methods appropriate for the data type.
    For sparse count data, as often encountered in single-cell studies, one can use {generalized factor models} (e.g., for Poisson or Bernoulli data), which assume a structure like $g(\EE[Y_{\cC}\mid U]) = U \Gamma^{\top}$, where $g(\cdot)$ is a link function.
    A key advantage is that the loadings matrix $\Gamma$ on the natural parameter scale is not necessarily sparse, even if the observed data is. This allows classical results on (generalized) factor analysis \citep{bai2012statistical,chen2020structured} to apply, under which the desired rate of $\|\hat{U} - U\|_{\fro} = \Op(m^{-1/2})$ can be achieved with a sufficiently large number of surrogate control outcomes ($m$).

    For the nuisance functions, such as $f = \EE[Y \mid U]$, our theory requires a rate of $\|\hat f - f\| = \op(n^{-1/4})$. This rate is readily satisfied by many modern, flexible machine learning estimators under standard assumptions. Below are three concrete examples.
    \begin{enumerate}[label=(\arabic*)]
        \item {Smooth functions:} If the true function $f$ is $s$-smooth and lies in a H\"older class of dimension $d$ (the dimension of $U$), minimax optimal estimators (like local polynomials or kernels) achieve a rate of $\|\hat f - f\| = \Op(n^{-s/(2s+d)})$ \citep{kennedy2022semiparametric}. This rate is faster than $n^{-1/4}$ as long as the dimension is not excessively large relative to the smoothness (specifically, when $d < 2s$).
        
        \item {Sparse high-dimensional functions:} If $f$ is approximately sparse with $s$ active features out of a total of $p$, estimators like the Lasso can achieve a rate of $\|\hat f - f\| = \Op(\sqrt{s\log p /n})$ \citep{kennedy2022semiparametric}. This rate is faster than $n^{-1/4}$ provided the sparsity $s$ is not too large (i.e., when $s = o(\sqrt{n}/\log p)$).
        
        \item {Random forests:} The $L_2$ consistency of random forests has been examined in various studies; see, for example, \citep{biau2012analysis,scornet2015consistency}.
        The rate of convergence is closely related to the minimax rate of $\Op(n^{-2/(q+2)})$ for nonparametric estimation involving $q$ features.
        In a simplified setting, \citet{biau2012analysis} demonstrated that this rate can be improved to $\Op(n^{-0.75/(s+0.75)})$, where $s$ represents the intrinsic dimension, which can be substantially smaller than the total feature dimension $q$.
        This often allows random forests to adapt to lower-dimensional structure and achieve rates faster than $n^{-1/4}$ in practice.
    \end{enumerate}
    Similar guarantees can be established for other flexible methods, such as neural networks, under appropriate regularity conditions.

    Below, we provide a numerical examination of the convergence rate for nuisance estimation using random forests with \ECV. Our findings indicate a $L_2$ convergence rate of approximately $n^{-1/4}$ for both nuisance functions on the simulated data, as illustrated in \Cref{fig:nuisance-est}.
    This supports the appropriate use of doubly robust estimators.

\begin{figure}[!ht]
    \centering
    \includegraphics[width=0.9\linewidth]{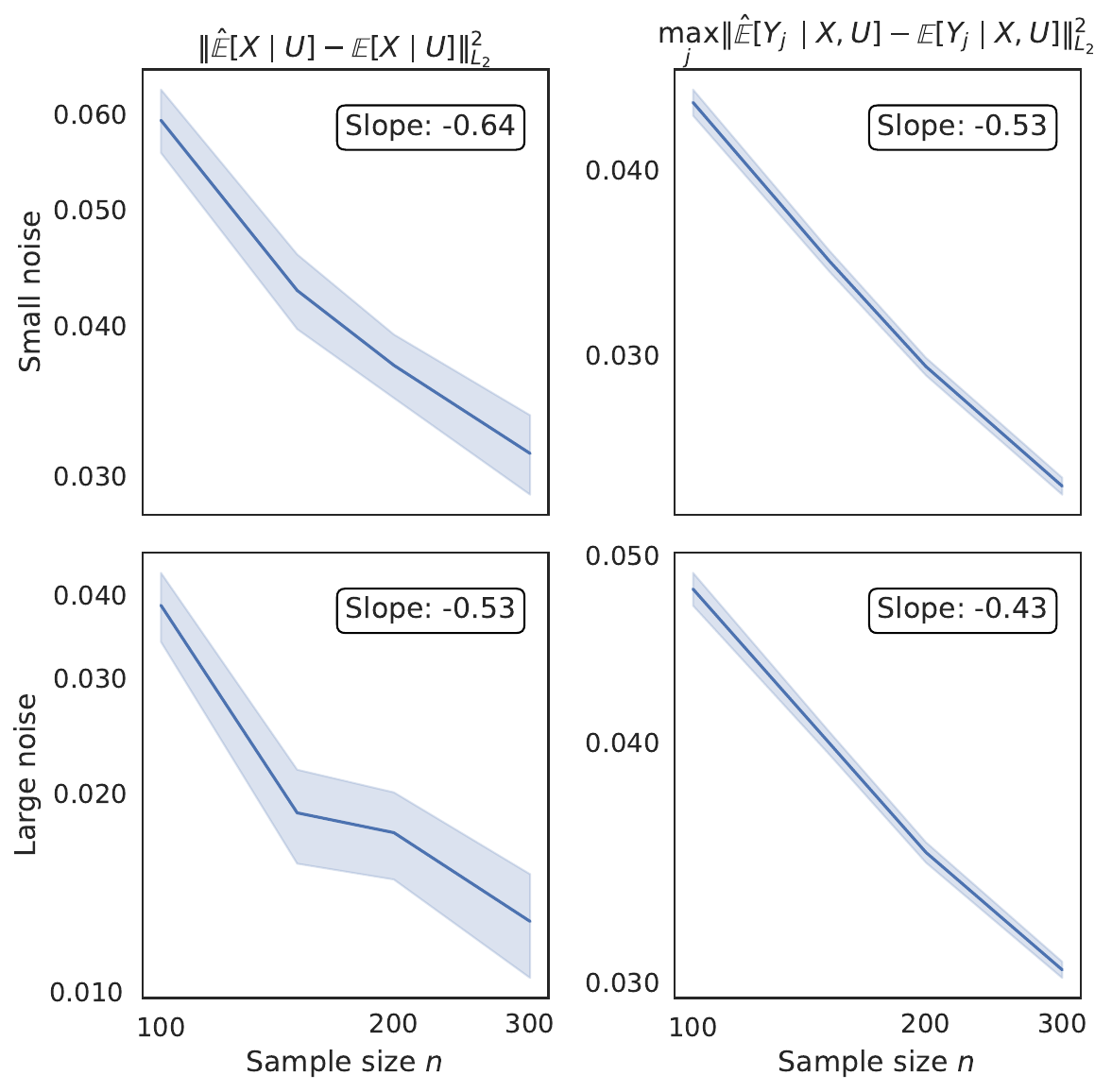}
    \caption{Estimation error of the nuisance regression function on simulated data using random forests.
    The axes are shown on a logarithmic scale, and the slope represents the estimated rate of convergence.
    The data-generating process is given in \Cref{sec:simu}, and we use the true latent embedding $U$ so that the ground truth regression function is computable.
    The errors are computed based on 1000 test observations without irreducible additive noise.
    }
    \label{fig:nuisance-est}
\end{figure}

\paragraph{Misspecification} 
To examine the effect of misspecification of the surrogate control outcomes, we contaminated the negative control set by including 5\%, 10\%, and 20\% non-null genes. Furthermore, we mimicked our real data analysis by using the least variable genes (based on $t$-statistics from GLMs that adjust for observed covariates) as negative controls. 
The results, presented in \Cref{fig:mis}, demonstrate our method's strong robustness. Even with 20\% contamination, the Type I error remains well-controlled with only a slight loss of power. Similarly, when using misspecified controls, the method's performance is nearly identical to the oracle case where the true controls are known.
Collectively, these simulations confirm that our method is highly robust to two realistic forms of negative control misspecification.

    \begin{figure}[!ht]
        \centering
        \includegraphics[width=0.9\linewidth]{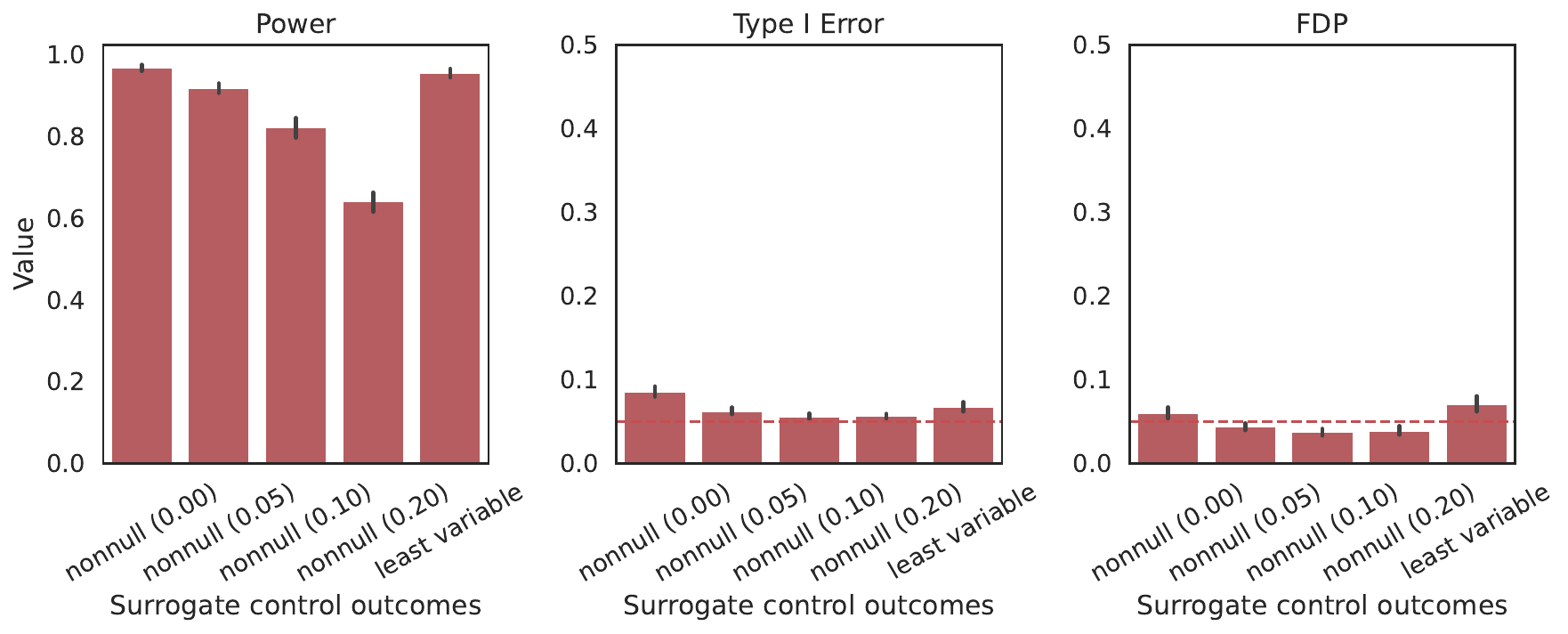}
        \caption{Performance of \PII under misspecification of surrogate controls.
        The setting `nonnull ($p$)' shows results when the control set is contaminated with a specified proportion ($p\in\{0\%, 5\%, 10\%, 20\%\}$) of non-null genes, and the setting `least variable' shows results when controls are selected as the 500 genes with the smallest-magnitude Wald statistics from a naive GLM. 
        The simulation setting and evaluation metrics are identical to those of \Cref{sec:simu}, with noise level $\sigma_{\epsilon}=1$. 
        The method demonstrates strong robustness in both misspecification scenarios.}
        \label{fig:mis}
    \end{figure}

    \paragraph{Comparison between \PII and \GLMoracle}
    In \Cref{fig:simu-1}, we observe that \PII outperforms \GLMoracle in terms of power.
    Our hypothesis is that \PII's advantage stems from its two-step, semiparametric nature, which provides a more stable estimation procedure than the one-step parametric \GLM in the presence of collinearity and model nonlinearity. We can reason through this as follows:

    In a simple linear model $Y=X\beta+U\eta+\epsilon$, the one-step regression ($Y\sim X+U$) and the two-step double-residual regression ($Y-\EE[Y \mid U] \sim X-\EE[X\mid U]$) are known to produce identical point estimates for $\beta$.
    However, the standard error estimates of the two methods satisfy the equation 
    \[(n-p-d)\hat{\sigma}^2_{one-step} = (n-p)\hat{\sigma}^2_{two-step}\] 
    when assuming homoskedasticity \citep[Theorem 2]{ding2021frisch}.
    Note that if allowing for heteroskedasticity, the HC0-type variance estimates of the two coincide \citep[Theorem 2]{ding2021frisch}.
    In other words, the variance estimate from the two-step procedure is smaller than the one from the one-step procedure, which may explain the power gain under linear models.

    In the nonlinear (Logistic) setting, the difference between one-step and two-step procedures can be even bigger because of numerical reasons.
    The \GLM-oracle is a one-step parametric procedure that simultaneously estimates the coefficients for $X$ and $U$. When $X$ and $U$ are highly correlated, the information matrix can become near-singular, leading to unstable estimates and high variance for $\hat{\beta}$, which in turn reduces statistical power.
    Our \PII method, in contrast, is a two-step semiparametric procedure. 
    It first uses a flexible machine learning model (random forests) to nonparametrically estimate and remove the complex, nonlinear influence of $U$ from both $X$ and the outcome $Y$. The final step estimates $\beta$ using only the remaining variation (the residuals). This orthogonalization procedure effectively mitigates the instability caused by collinearity, resulting in a more stable, lower-variance estimator for $\beta$ and consequently, greater statistical power for hypothesis testing.

    In summary, both the difference in variance estimation and numerical stability contribute to better power.
    On the other hand, we also acknowledge that \PII slightly inflates the type-I error, as we observe in the simulation for nonlinear models (\Cref{fig:simu-1}).

\subsection{Real data}\label{app:ex-real-data}
\paragraph{Extended background}
In a recent single-cell CRISPR perturbation study, \citet{lalli2020high} investigated the molecular mechanisms of genes associated with neurodevelopmental disorders, particularly Autism Spectrum Disorder (ASD). Using a modified CRISPR-Cas9 system, they performed gene suppression experiments on 13 ASD-linked genes in the Lund Human Mesencephalic (LUHMES) neural progenitor cells. The experiment comprised 14 groups: 13 treatment groups with individual gene knockdowns and one control group. Single-cell RNA sequencing was employed to assess gene expression changes resulting from each knockdown. The authors estimated a pseudotime trajectory, which approximates the progression of neuronal differentiation.  The analysis of \citet{lalli2020high} suggests that some perturbations cause changes in pseudotime (slow or fast development); see \Cref{fig:pseudotime}.
A scientific question of interest not answered by \citet{lalli2020high} is whether some perturbation explains anything beyond the changes in expression levels caused by cell development.
    
In single-cell CRISPR perturbation experiments, confounding factors can significantly impact the interpretation of results. 
    Unlike controlled experiments, these studies often resemble observational data, where confounding variables such as cell size, cell cycle stage, or microenvironment heterogeneity may influence gene expression patterns.
    These confounders can mask or mimic the effects of the intended genetic perturbations, potentially leading to erroneous conclusions about gene function or regulatory networks. Addressing these confounding issues is crucial for the accurate interpretation of CRISPR perturbation data and for distinguishing true biological effects from technical artifacts.

    To adjust for possible confounding effects, we may take advantage of the multiple surrogate control genes.
    Even though tens of thousands of genes are measured, one typically restricts the differential expression analysis to the top thousands of highly variable genes.
    For the remaining genes with low variations, it is believed that there will not be sufficient power to differentiate the response from the null distribution. 
    But even with low power, it is likely that, in total, one can detect the impact of confounding. 
    For this reason, we use such genes as surrogate control outcomes; even if this choice is incorrect, we still target meaningful statistical estimands, provided that the estimated embedding captures the common variability of all cells under control.
    Alternatively, we can also use housekeeping genes as surrogate control outcomes.
    The main goal here is to demonstrate a practical procedure for post-integrated inference and show that our asymptotic results are reasonably accurate in real data.

\paragraph{Data}
    After filtering out low-quality cells and genes that expressed in less than 10 cells, we retained 8320 cells and 13086 genes under 14 perturbation conditions (including control) from \citet{lalli2020high}.
    Following the routine selection procedure of highly variable genes in genomics \citep{seurat}, we select 4163 genes whose standardized variance is larger than 1, and the last 4000 genes with the lowest standardized variances are treated as surrogate control outcomes.
    The covariates we measured include the logarithm of library sizes, cell cycle scores (`S.Score' and `G2M.Score'), batches (3 categories), and pseudotime states (normalized to range from 0 to 1).
    After one-hot encoding of the categorical features, we have 19 covariates (including 13 perturbation indicators) and 4163 genes for model fitting.
    For each highly variable gene, we aim to test whether its gene expressions vary along the pseudotime state under perturbation conditions.

When restricted to a small subset of significant genes discovered by \PII, their expression levels are visualized as a function of pseudotime states and perturbation conditions in \Cref{fig:pten}.
We observe an increasing trend of the expression and the overexpression in the perturbed group at the very late stage of pseudotime.
The significance suggests that these genes may be affected not only by cell development but also by \emph{PTEN} repression.
\emph{NEFM} is involved in neurite outgrowth and axon caliber \citep{cheung2023neuropathological}, \emph{TUBB2B} and \emph{TUBA1A} encode critical structural subunits of microtubules that are enriched during brain development \citep{jaglin2009mutations}, \emph{HN1} is related to cancer and senescence \citep{jia2019hnrnpa1}.
Given the role of \emph{PTEN} on neural differentiation and related processes, these genes could be affected. Further research would be needed to establish any direct links between \emph{PTEN} repression and the expression or function of these specific genes during neural differentiation.

\paragraph{Sensitivity analyses}
 To assess the self-consistency of our method with respect to its parameters, we perform sensitivity analyses. 
 First, we varied the number of principal components used for the PCA-based embedding, running \PII with 20, 30, and 50 PCs of the surrogate control outcomes. 
 Second, we varied the number of negative controls, using the bottom 3000, 4000, and 5000 lowest-variance genes.

We compare the pairwise Jaccard similarity indices (defined as $J(A,B)= |A\cap B|/|A\cup B|$ for two sets $A$ and $B$) of discoveries under different hyperparameter settings.
From \Cref{fig:sensitivity}, we observe a high Jaccard similarity among different hyperparameters, which indicates high degrees of overlap in the significant genes discovered across these different parameter settings. 
These analyses demonstrate that the discoveries made by our method are highly stable across a reasonable range of user-specified parameters. This reinforces the reliability of the biological conclusions drawn from our data analysis.

\begin{remark}[Hidden confounder versus hidden mediators]\label{rmk:mediator}
    A key example of mediation ($X \rightarrow U \rightarrow Y$) arises when studying a gene perturbation's effect on neuronal differentiation, under the crucial assumption that cells start from a homogeneous developmental stage. If the initial cell states were heterogeneous, the starting state would act as a confounder, influencing both treatment efficacy and final outcome.

    In a synchronized population, the treatment ($X$) is the knockdown of a developmental gene, such as \textit{PTEN}. This perturbation causally alters the cell's subsequent developmental progress. This new differentiation status, captured by the latent embedding, serves as the mediator ($U$). In turn, the change in the cell's internal state ($U$) drives the expression of the outcome ($Y$), such as the late-stage neuronal marker \textit{MAP2}. The effect of the knockdown is therefore mediated through its impact on the cell's progression along the differentiation pathway.

    The causal pathway is as follows: the \textit{PTEN} knockdown ($X$) alters the cell's developmental program, pushing it along the differentiation trajectory ($U$). This change in differentiation status, in turn, drives the expression of the mature neuronal marker ($Y$). Therefore, the effect of the \textit{PTEN} knockdown on \textit{MAP2} expression is mediated through its impact on the cell's overall developmental state.

    Other realistic examples of mediation include a drug treatment ($X$) inducing a cellular stress state ($U$) that leads to apoptosis ($Y$), or a CRISPR knockdown of a cell cycle gene ($X$) altering the cell cycle distribution ($U$), which then affects the expression of phase-specific genes ($Y$).
\end{remark}

\begin{figure}[!ht]
    \centering
    \includegraphics[width=0.8\linewidth]{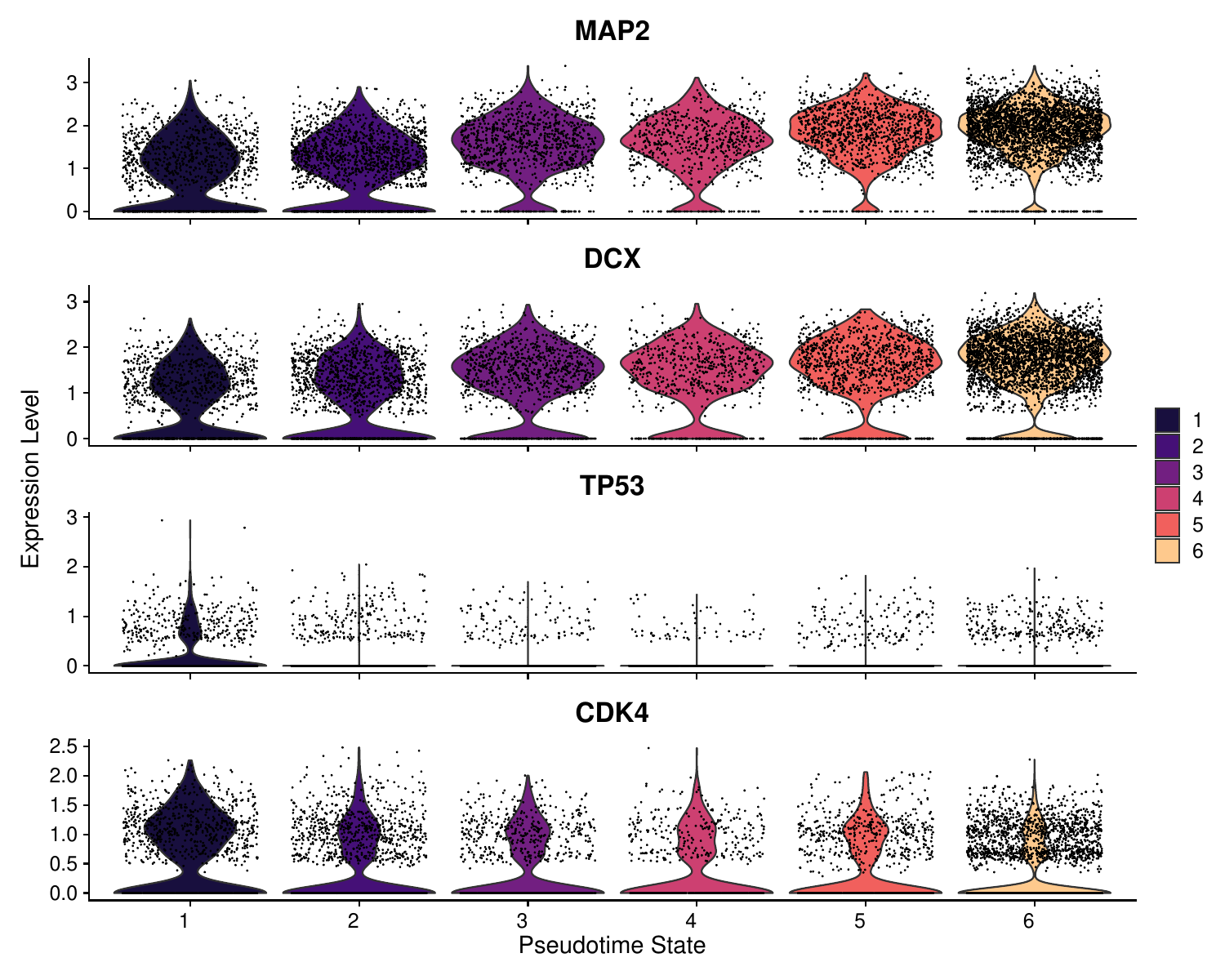}
    \caption{Expression levels of marker genes in different estimated pseudotime states.
    Genes \emph{MAP2} and \emph{DCX} are neuronal markers (expressed in more differentiated cells)
    while genes \emph{TP53} and \emph{CDK4} are progenitor markers (expressed in less differentiated cells).}
    \label{fig:pseudotime}
\end{figure}

\begin{figure}[!ht]
    \centering
    \includegraphics[width=0.8\linewidth]{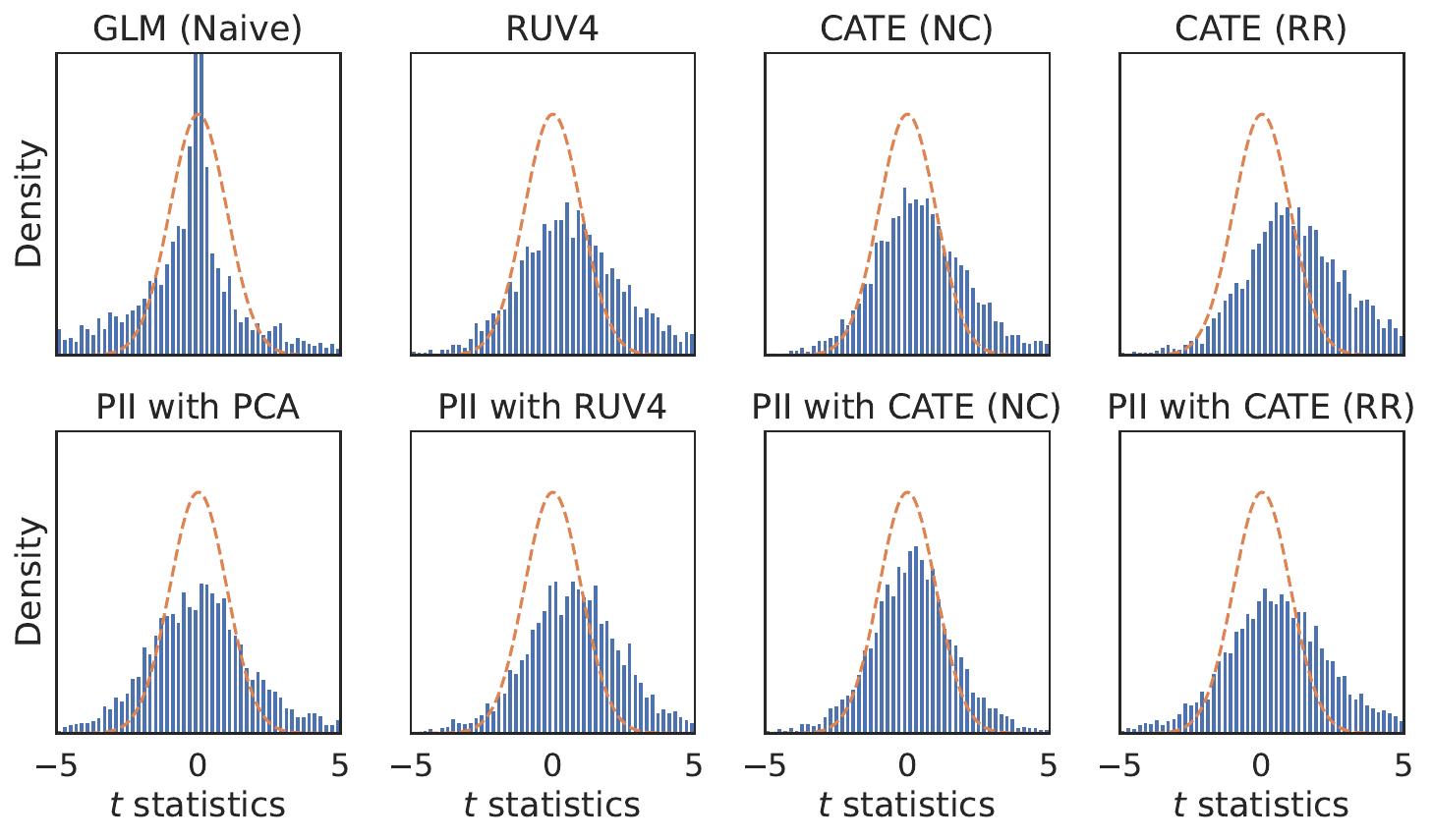}
    \caption{
    Histogram of test statistics for main effects of pseudotime states on the expressions of 4163 genes.
    Many genes are significant because the expression levels are expected to change during neural differentiation.
    }
    \label{fig:stat-pt}
\end{figure}

\begin{figure}[!ht]
    \centering
    \includegraphics[width=\linewidth]{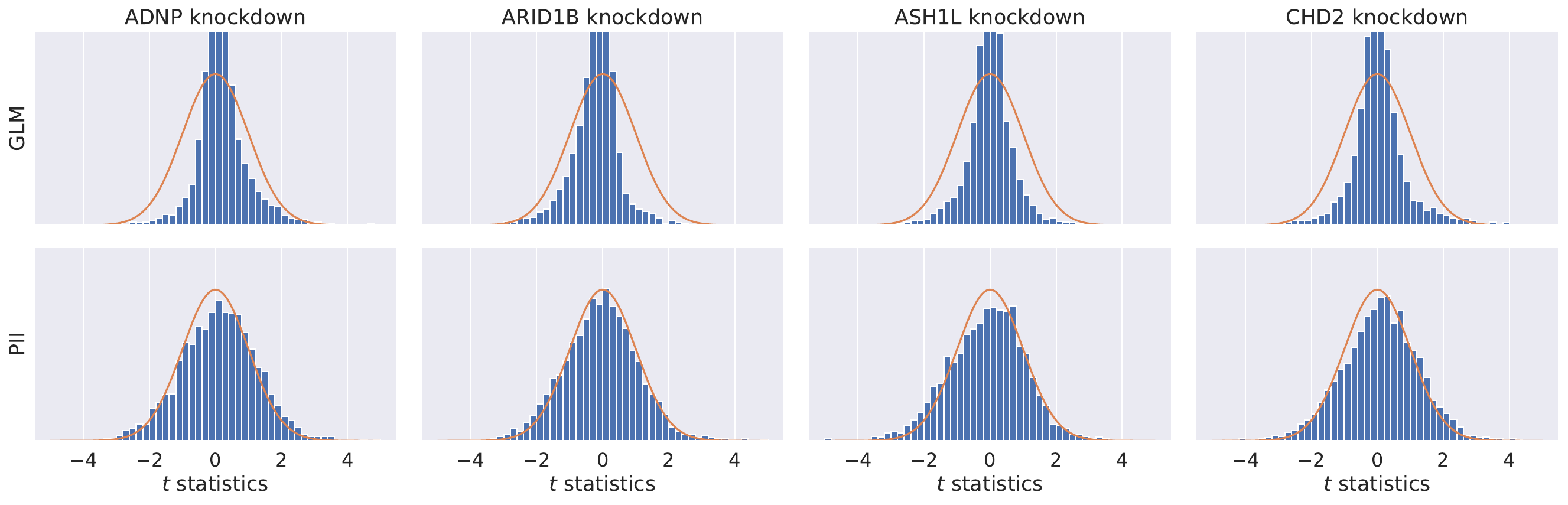}
    \vspace{5mm}
    \includegraphics[width=\linewidth]{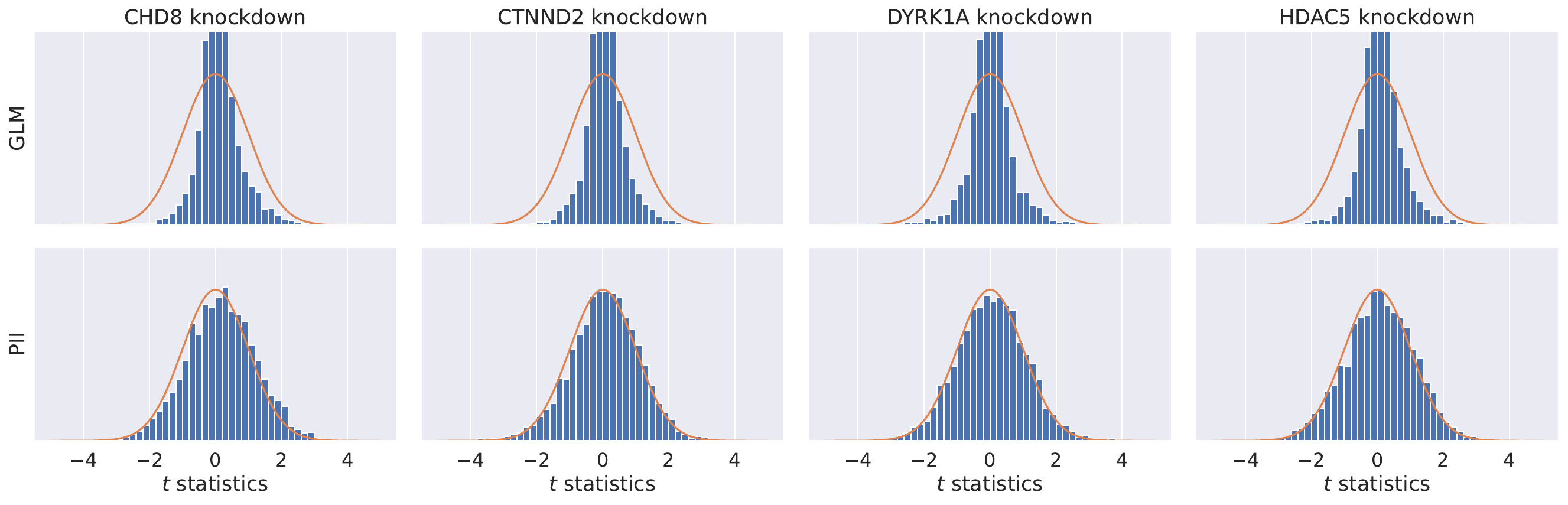}
    \vspace{5mm}
    \includegraphics[width=\linewidth]{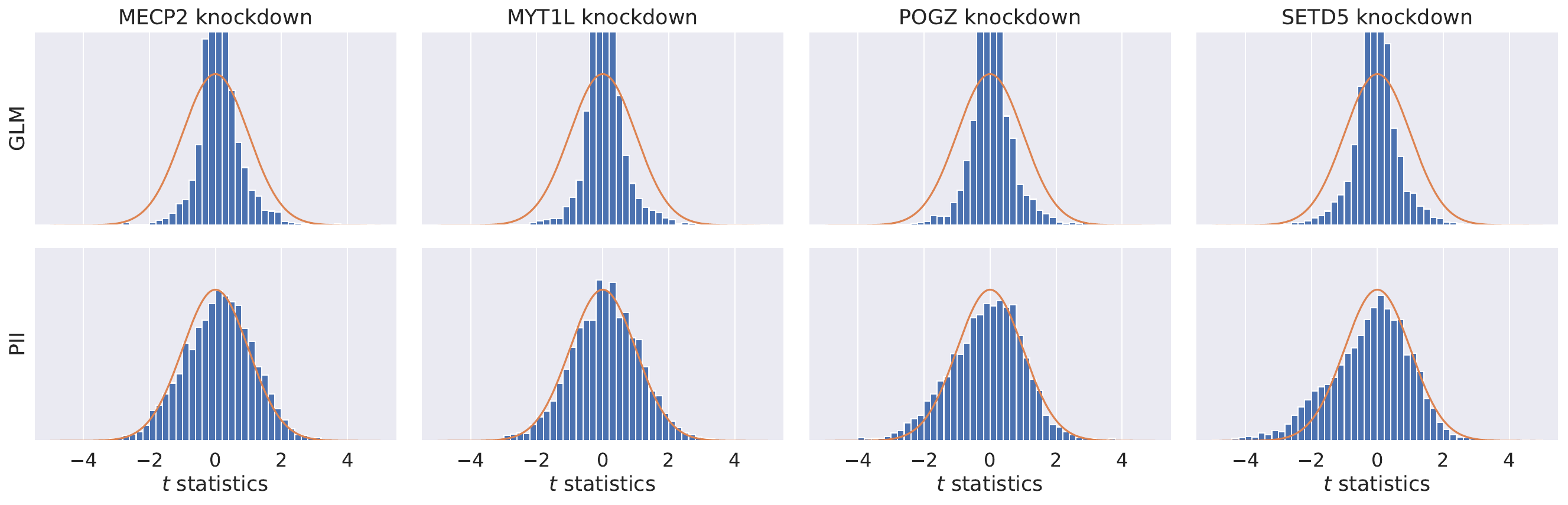}
    \caption{
    Histogram of test statistics on 4163 genes for 12 different perturbation conditions.
    Different rows represent the results of different methods:
    \GLM: Score tests by generalized linear models with Negative Binomial likelihood and log link function.
    The covariance matrix is estimated using the HC3-type robust estimator.
    \PII: The proposed post-integrated inference with 50 principal components as the estimated embeddings.
    }
    \label{fig:stat-all}
\end{figure}

\begin{figure}
        \centering
        \includegraphics[width=0.8\linewidth]{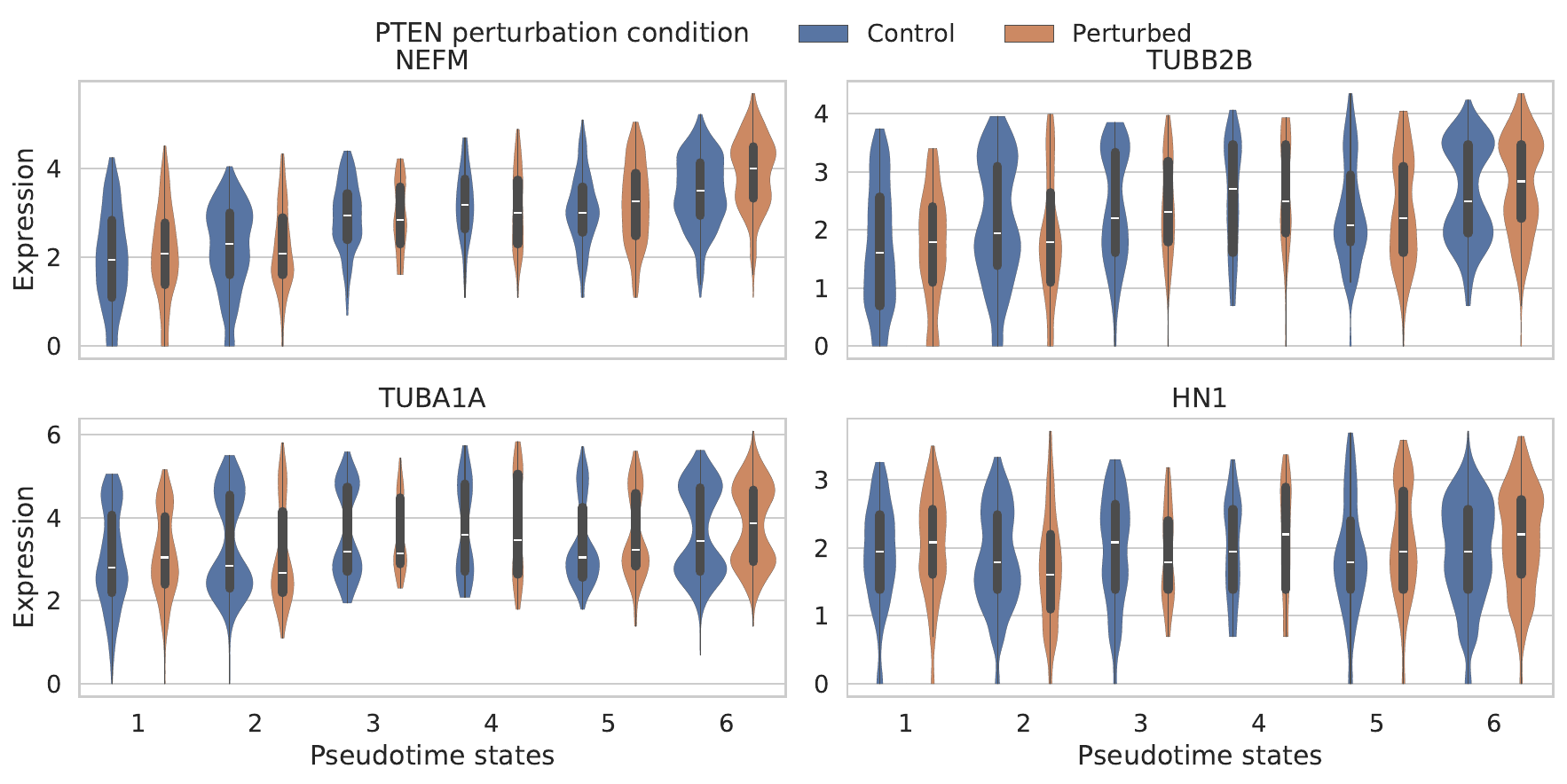}
        \caption{
        Gene expressions of significant genes in the control group and the \emph{PTEN} knockdown group.
        Four genes with positive estimated effect sizes are selected with a p-value threshold of 0.01 for both pseudotime states and \emph{PTEN} knockdown for three \PII methods in \Cref{fig:upset}(b) and a median expression level larger than zero.
        }
        \label{fig:pten}
\end{figure}

\begin{figure}
    \centering
    \includegraphics[width=0.9\linewidth]{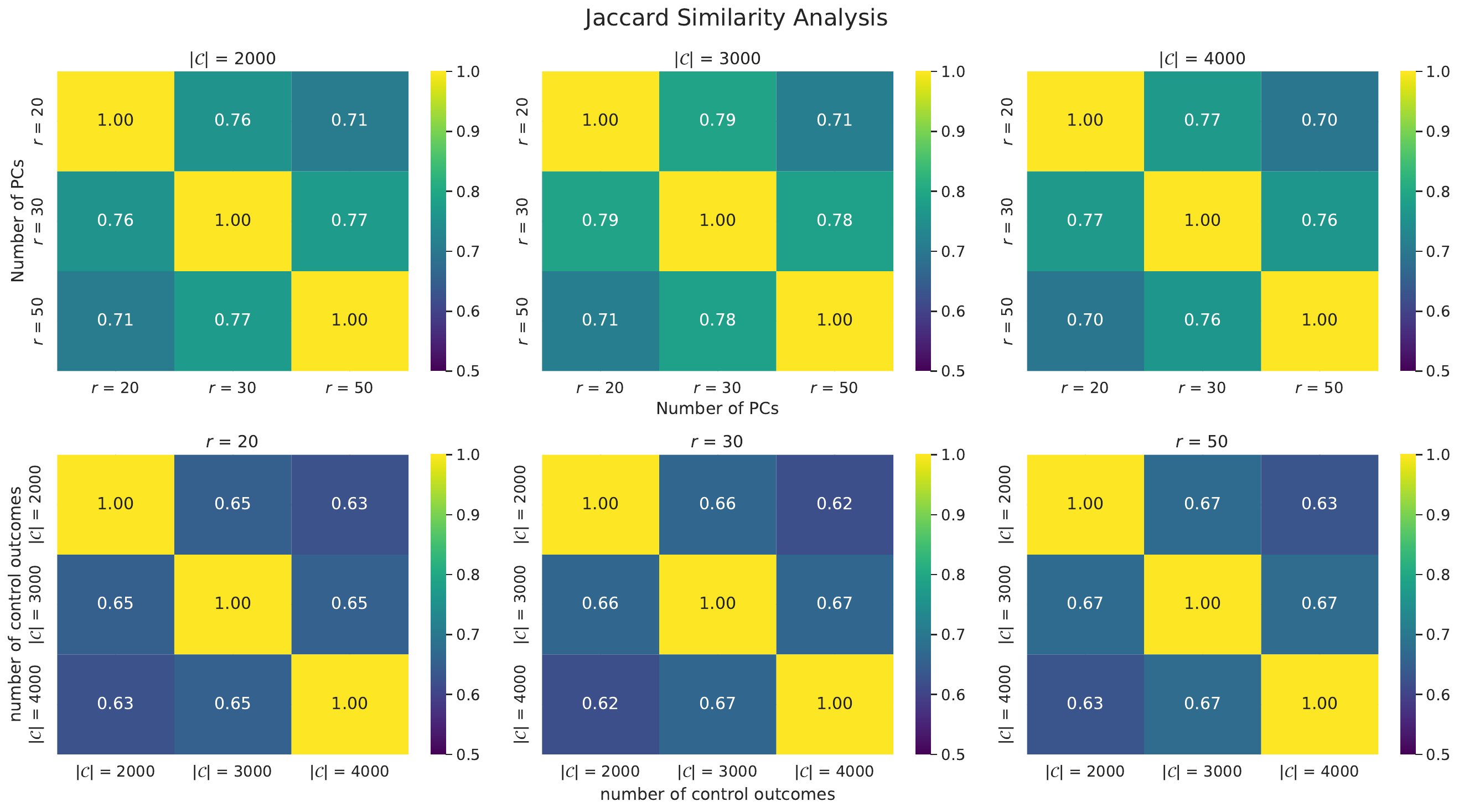}
    \caption{Sensitivity analysis of \PII's discoveries using principal components of surrogate control outcomes. 
    The heatmaps show the Jaccard similarity between sets of significant genes (p-values$ < 0.05$) when varying the number of principal components and the number of negative controls.}
    \label{fig:sensitivity}
\end{figure}

\end{document}